\documentclass[acmsmall,screen,nonacm]{acmart}

\AtBeginDocument{%
  }

\usepackage{ifthen}
\usepackage{xspace}

\newboolean{talk}
\setboolean{talk}{false}
\newboolean{paper}
\setboolean{paper}{false}

\newboolean{IEEEstyle}
\setboolean{IEEEstyle}{false}
\newboolean{lipicsstyle}
\setboolean{lipicsstyle}{false}
\newboolean{eptcsstyle}
\setboolean{eptcsstyle}{false}
\newboolean{entcsstyle}
\setboolean{entcsstyle}{false}
\newboolean{sigplanstyle}
\setboolean{sigplanstyle}{false}
\newboolean{easychairstyle}
\setboolean{easychairstyle}{false}
\newboolean{scrartclstyle}
\setboolean{scrartclstyle}{true}
\newboolean{lncsstyle}
\setboolean{lncsstyle}{false}

\newboolean{acmartstyle}
\setboolean{acmartstyle}{false}

\newboolean{needstheorems}
\setboolean{needstheorems}{false}
\newboolean{withimages}
\setboolean{withimages}{false}
\newboolean{withproofs}
\setboolean{withproofs}{true}

\newboolean{french}
\setboolean{french}{false}

\setboolean{acmartstyle}{true}
\newboolean{colored}
\setboolean{colored}{false}
\setboolean{withproofs}{true}

\usepackage{amsthm}
    \newtheorem{theorem}{Theorem}[section]
    \newtheorem{lemma}[theorem]{Lemma}
    
    \newtheorem{corollary}[theorem]{Corollary}
    \newtheorem{proposition}[theorem]{Proposition}
    \newtheorem{definition}[theorem]{Definition}
    
    \ifthenelse{\boolean{scrartclstyle}}{
  \newtheorem{example}[theorem]{Example}
}{}

\theoremstyle{remark}
    
\usepackage[utf8]{inputenc}

\ifthenelse{\boolean{acmartstyle}}{}{
\usepackage{amssymb}	
}
\usepackage{marvosym}
\usepackage{stmaryrd}
\usepackage{bussproofs}
\usepackage{mathtools}
\usepackage{bbold}
\usepackage{enumerate}
\usepackage{qtree}
\usepackage{tikz-cd}
\usetikzlibrary{calc}

\usepackage{ifthen}
\usepackage{xspace}
\usepackage{fancybox}
\usepackage{hhline}

\usepackage{soul}

\usepackage{adjustbox}

\usepackage{xcolor}
\usepackage{multirow}
\usepackage{microtype}

\usepackage{hyperref}

\hypersetup{colorlinks=true}

\usepackage[colorinlistoftodos]{todonotes}
\usepackage{ifthen}
\usepackage{proof}
\usepackage{pifont}

\usepackage{appendix}
\usepackage{marginnote}

\newboolean{appendix}
\setboolean{appendix}{false}
\newboolean{withlinkedproofs}
\setboolean{withlinkedproofs}{true}

\capturecounter{theorem}

\newcommand{\NoteProof}[1]{
	\ifthenelse{\boolean{withproofs}}{\ifthenelse{\boolean{appendix}}{
	\marginnote{Originally at p. \pageref{#1}}
	}{
	\marginnote{{Proof at p.\,{\pageref{app:#1}}}}
	}
	}{}
}

\newcommand{\applabel}[1]{$\phantomsection\label{app:#1}$}

\newcommand{\sem}[1]{\interp{#1}}

\newcommand{\semint}[1]{\interp{#1}^{\intsym}}

\newcommand{\ignore}[1]{}
 
\newcommand{\colspace}{@{\hspace{.5cm}}}
\newcommand{\hcolspace}{@{\hspace{.25cm}}}
\newcommand{\myinput}[1]{\ifthenelse{\boolean{withimages}}{\input{#1}}{}}

\newcommand{\reflemma}[1]{Lemma~\ref{l:#1}}

\newcommand{\reflemmaeq}[1]{{L.\ref{l:#1}}}

\newcommand{\refth}[1]{Thm.~\ref{th:#1}}

\newcommand{\refthm}[1]{Thm.~\ref{thm:#1}}

\newcommand{\refprop}[1]{Prop.~\ref{prop:#1}}
\newcommand{\refpropp}[2]{Prop.~\ref{prop:#1}\eqref{p:#1-#2}} 
\newcommand{\refsect}[1]{Sect.~\ref{sect:#1}}

\newcommand{\refapp}[1]{Appendix~\ref{app:#1} (p.~\pageref{app:#1})}

\renewcommand{\refeq}[1]{(\ref{eq:#1})} 
\newcommand{\reffig}[1]{Fig.~\ref{fig:#1}}
\newcommand{\refcoro}[1]{Cor.~\ref{coro:#1}}

\newcommand{\refnot}[1]{Notation~\ref{not:#1}}
\newcommand{\refdef}[1]{Definition~\ref{def:#1}}

\newcommand{\refex}[1]{Ex.~\ref{ex:#1}}

\newcommand{\ie}{\textit{i.e.}\xspace}
\newcommand{\eg}{\textit{e.g.}\xspace}
\newcommand{\ih}{\textit{i.h.}\xspace}

\newcommand{\blue}[1]{{\color{blue} {#1}}}

\newcommand{\darkgreen}[1]{{\color{green!50!black} {#1}}}

\newcommand{\defeq}{\coloneqq} 
\newcommand{\grameq}{\Coloneqq} 
\newcommand{\set}[1]{\{#1\}}

\newcommand{\nat}{\mathbb{N}}
\newcommand{\size}[1]{|#1|}

\renewcommand{\l}{\lambda}
\newcommand{\isub}[2]{\{#1/#2\}}
\renewcommand{\isub}[2]{\{#1\,{:=}\,#2\}}
\newcommand{\esub}[2]{[#1/#2]}
\renewcommand{\esub}[2]{[#1{\shortleftarrow}#2]}

\newcommand{\rootRew}[1]{\mapsto_{#1}}
\newcommand{\Rew}[1]{\rightarrow_{#1}}

\newcommand{\tob}{\Rew{\beta}}

\newcommand{\shufeqext}{\shufeqext} 

\newcommand{\tm}{t}
\newcommand{\tmtwo}{u}
\newcommand{\tmthree}{s}
\newcommand{\tmfour}{r}

\newcommand{\tmp}{\tm'}

\newcommand{\var}{x}
\newcommand{\vartwo}{y}
\newcommand{\varthree}{z}
\newcommand{\varfour}{w}

\newcommand{\ctxholep}[1]{\langle #1\rangle}
\newcommand{\ctxhole}{\ctxholep{\cdot}}

\newcommand{\ctx}{C}
\newcommand{\ctxtwo}{\ctx'}

\newcommand{\ctxp}[1]{\ctx\ctxholep{#1}}
\newcommand{\ctxtwop}[1]{\ctxtwo\ctxholep{#1}}

\newcommand{\Cctx}{\mathcal{C}}
\newcommand{\Cctxtwo}{\mathcal{D}}
\newcommand{\Hctx}{\mathcal{H}}
\newcommand{\Cctxtwop}[1]{\Cctxtwo\ctxholep{#1}}
\newcommand{\Cctxp}[1]{\Cctx\ctxholep{#1}}

\newcommand{\hctx}{H}

\newcommand{\hctxp}[1]{\Hctx\ctxholep{#1}}

\newcommand{\hauxctx}{P}

\newcommand{\hauxctxp}[1]{\hauxctx\ctxholep{#1}}

\newcommand{\arbctxp}[1]{\arbctxp{#1}}
\newcommand{\arbctxtwop}[1]{\arbctxtwop{#1}}

\newcommand{\tctx}{T}

\ifthenelse{\boolean{acmartstyle}}{
	
	}{
	
}

\AtEndPreamble{%
\theoremstyle{acmdefinition}
\newtheorem{notation}[theorem]{Notation}
\newtheorem*{remark*}{Remark}

}

\newcommand{\la}[1]{\lambda #1.}

\newcommand{\myproof}[1]{
\ifthenelse{\boolean{omitproofs}}{\begin{IEEEproof} Proof available but omitted for readability. \end{IEEEproof}}{#1}}

\newcommand{\levy}{{L{\'e}vy}\xspace}

\newcommand{\node}{\mathtt{n}}

\newcommand{\withproofs}[1]{\ifthenelse{\boolean{withproofs}}{#1}{}}

\newcommand{\withoutproofs}[1]{\ifthenelse{\boolean{withproofs}}{}{#1}}

\newcommand{\doubt}[1]{}

\newcounter{numberone}
\newcounter{numberoneroman}
\newcounter{numberonealph}

\newcommand{\mset}[1]{[#1]}
\newcommand{\emptymset}{\zero}

\newcommand{\zero}{\mathbf{0}}

\newcommand{\typctx}{\Gamma}
\newcommand{\typctxtwo}{\Delta}

\newcommand{\tder}{\pi}
\newcommand{\tdertwo}{\sigma}

\newcommand{\ruleAp}{@}
\newcommand{\ruleAx}{\mathsf{ax}}

\newcommand{\hastype}{\!:\!}

\newcommand{\domain}[1]{\mathsf{dom}(#1)}

\makeatletter
\newsavebox{\@brx}
\newcommand{\llangle}[1][]{\savebox{\@brx}{\(\m@th{#1\langle}\)}%
  \mathopen{\copy\@brx\kern-0.7\wd\@brx\usebox{\@brx}}}
\newcommand{\rrangle}[1][]{\savebox{\@brx}{\(\m@th{#1\rangle}\)}%
  \mathclose{\copy\@brx\kern-0.7\wd\@brx\usebox{\@brx}}}
\makeatother

\newcommand{\symfont}[1]{\mathtt{#1}}

\newcommand{\rtob}{\rootRew{\beta}}

\newcommand{\ntm}{n}

\newcommand{\Id}{\symfont{I}}

\let\cal\undefined
\newcommand{\cal}[1]{\mathcal{#1}}
\newcommand{\relsym}{{\cal R}}
\newcommand{\rel}{~\relsym~}

\usepackage{proof}

\newcommand{\FV}[1]{\mathsf{fv}(#1)}

\newcommand{\bshs}{\Downarrow_{\mathrm{h}}\,} 
\newcommand{\bsh}[1]{\Downarrow_{\mathrm{h}}^{#1}\,} 
\newcommand{\bshdiv}{\not\bsh{}}

\newcommand{\bohm}{B{\"o}hm\xspace}

\newcommand{\mtype}{\typefont{M}}
\newcommand{\mtypetwo}{\typefont{N}}
\newcommand{\mtypethree}{\typefont{O}}

\newcommand{\emptytype}{[~]}
\renewcommand{\emptytype}{\zero}

\newcommand{\multitype}[2]{[{#2}_{1},\ldots,{#2}_{#1}]}

\newcommand{\ltype}{\typefont{L}}
\newcommand{\ltypetwo}{\ltype'}
\newcommand{\ltypethree}{\ltype''}

\newcommand{\gtype}{\typefont{T}}
\newcommand{\gtypetwo}{\gtype'}

\newcommand{\vartype}{X}
\renewcommand{\vartype}{\typefont{A}}

\newcommand{\typectx}{\Gamma}
\newcommand{\typectxtwo}{\Delta}

\newcommand{\emptytypectx}{\emptyset}

\newcommand{\typingruleApp}{@}
\newcommand{\typingruleAx}{\mathsf{ax}}
\newcommand{\typingruleAbs}{\lambda}
\newcommand{\typingruleMany}{\mathsf{many}}

\newcommand{\typesym}{\mathrm{typ}}
\newcommand{\leqtype}{\sqsubseteq^{\typesym}_{\intsym}}
\newcommand{\equivtype}{\equiv^{\typesym}}
\newcommand{\leqtypepl}{\sqsubseteq^{\typesym}}

\newcommand{\interp}[1]{\llbracket #1 \rrbracket}

\newcommand{\tderiv}{\tder}
\newcommand{\tderivp}{\tderiv'} 
\newcommand{\tderivtwo}{\tdertwo}

\newcommand{\derive}[2]{#1 \derives #2}
\newcommand{\concl}[4]{\derive{#1}{#2 \vdash #3 \hastype #4}}

\newcommand{\derives}{\vartriangleright}

\newcommand{\eqth}{\mathcal{T}}

\newcommand{\lang}{\mathcal{L}}

\newcommand{\htm}{h}
\newcommand{\htmtwo}{\htm'}

\newcommand{\red}[1]{{\color{red} {#1}}}

\newcommand{\blueclr}{\blue{\mathsf{b}}}
\newcommand{\redclr}{\red{\mathsf{r}}}

\newcommand{\rla}[1]{\red{\lambda_{\redclr}} #1.}
\newcommand{\bla}[1]{\blue{\lambda_{\blueclr}} #1.}

\newcommand{\rapp}[2]{#1 \red{\bullet_{\redclr}} #2}
\newcommand{\bapp}[2]{#1 \blue{\bullet_{\blueclr}} #2}

\newcommand{\hsym}{\symfont{h}}

\newcommand{\toh}{\Rew{\hsym}}

\newcommand{\hcolsym}{\hsym\redclr\blueclr}

\newcommand{\tohcol}{\to_{\hcolsym}}

\newcommand{\bshcols}{\Downarrow_{\hchsym}\,} 
\newcommand{\bshcol}[1]{\Downarrow_{\hchsym}^{\intsym#1}\,} 
\newcommand{\bshcoldiv}{\not\bshcols}

\newcommand{\Lambdac}{\Lambda_{\RB}}

\newcommand{\equivchcol}{\sqsubseteq_{\hcolsym}^{\qcontextualsym}}
\newcommand{\leqchcol}{\sqsubseteq_{\hcolsym}^{\qcontextualsym}}
\renewcommand{\equivchcol}{\equiv_{\intsym}^{\mathrm{ctx}}}
\renewcommand{\leqchcol}{\sqsubseteq_{\intsym}^{\mathrm{ctx}}}

\newcommand{\intleq}{\sqsubseteq^{\mathrm{int}}}
\newcommand{\inteq}{\equiv^{\mathrm{int}}}

\newcommand{\wintleq}{\sqsubseteq^{\redclr\mathrm{int}}}

\newcommand{\abscolorone}{p}
\newcommand{\abscolortwo}{q}
\newcommand{\clr}[1]{\mathsf{\abscolorone}_{#1}}
\newcommand{\clrtwo}[1]{\mathsf{\abscolortwo}_{#1}}
\newcommand{\clrp}[1]{\mathsf{\abscolorone}'_{#1}}

\newcommand{\colr}{\mathsf{\abscolorone}}
\newcommand{\colrtwo}{\mathsf{\abscolortwo}}

\newcommand{\cla}[2]{\lambda_{\clr{#1}} #2.}

\newcommand{\ccapp}[3]{#2 \bullet_{\clr {#1}} #3}
\newcommand{\capp}[3]{#2 \bullet_{\clrtwo {#1}} #3}

\newcommand{\appsym}{\bullet}
\renewcommand{\capp}[3]{#2 \appsym_{\clrtwo {#1}} #3}

\newcommand{\manyclam}[2]{\lambda_{\clr{1}\cdots\clr{#1}} #2_1\ldots#2_{#1}.\,}
\newcommand{\manyblam}[2]{\lambda_{\blueclr\cdots \blueclr}#2_1\ldots#2_{#1}.\,}
\newcommand{\manyrlam}[3][1]{\lambda_{\redclr\cdots \redclr}#3_{#1}\ldots#3_{#2}.\,}

\newcommand{\manycapp}[3]{\capp{#1}{\capp{1}{#2}{#3_1} \cdots}{#3_{#1}}}

\newcommand{\clavec}[2]{\lambda_{\vec{#1}}\vec{#2}.}
\newcommand{\cappvec}[3]{#2\appsymp{\vec{#1}}\vec{#3}}

\newcommand{\typearrow}{\rightarrow}

\newcommand{\ctypes}[1]{\vdash^{#1}_{\intsym}}

\newcommand{\leqctype}{\precsim_{\red{ty}\blue{pe}}}
\newcommand{\equivctype}{\simeq_{\red{ty}\blue{pe}}}

\renewcommand{\leqctype}{\leqtype}
\renewcommand{\equivctype}{\equivtype_{\intsym}}
\newcommand{\leqbtype}{\sqsubseteq^{\blueclr\typesym}}
\newcommand{\eqbtype}{\equiv^{\blueclr\typesym}}
\newcommand{\equivbtype}{\eqbtype}

\newcommand{\typearrowpp}[2]{\xrightarrow{{#1}{#2}}}

\newcommand{\monoToPlayer}[2]{\overline{#2}^{#1}}

\newcommand{\monoToBlue}[1]{\monoToPlayer{\blueclr}{#1}}
\newcommand{\monoToRed}[1]{\monoToPlayer{\redclr}{#1}}

\newcommand{\spsym}{\mathrm{sp}}

\newcommand{\laxsym}{\eta^\infty}
\newcommand{\laxbsym}{\cB\laxsym}

\newcommand{\btleq}{\btpre}
\newcommand{\etabtle}{\sqsubseteq_{\laxbsym}}
\newcommand{\etabteq}{=_{\laxbsym}}

\newcommand{\exder}{\,\triangleright}

\newcommand\mplus{\uplus}

\newcommand{\bsub}{\begin{enumerate}[(i)]}
\renewcommand{\esub}{\end{enumerate}}
\newcommand{\obsle}[1][]{\sqsubseteq_{#1}^{\mathrm{ctx}}}
\newcommand{\obseq}{\equiv^{\mathrm{ctx}}}
\newcommand{\obseqpar}[1]{\equiv_{#1}^{\mathrm{ctx}}}
\newcommand{\obseqh}{\obseqpar{\hsym}}
\newcommand{\costeq}{\equiv^{\mathrm{cost}}}
\newcommand{\lam}{\ensuremath{\lambda}}
\newcommand{\las}[2]{\lambda #1_{1}\dots #1_{#2}.}
\newcommand{\apps}[2]{{#1_1}\cdots #1_{#2}}
\newcommand{\Tupler}[1]{{\sf T}_{#1}}
\newcommand{\Tuple}[1]{\langle #1\rangle}
\newcommand{\Proj}[2]{{\sf S}^{#1}_{#2}}
\newcommand{\Var}{\textsc{Var}}
\newcommand{\conv}[1]{\!\Downarrow_{#1}\,} 
\renewcommand{\conv}[1]{\Downarrow_{#1}} 
\newcommand{\convc}[2]{\!\Downarrow_{#1}^{#2}\,} 
\newcommand{\relation}[1]{{\sf #1}}

\newcommand{\iset}{I}

\newcommand{\emptytypctx}{\epsilon_\typctx}

\newcommand{\yinyang}[1][1]{%
    \begin{tikzpicture}[scale=#1*0.07]
      \draw[line width = #1*0.05mm,transform canvas={yshift=0.02cm}] (0,0) circle (1cm);
      \path[fill=black,transform canvas={yshift=0.02cm}] (90:1cm) arc (90:-90:0.5cm)
                        (0,0)    arc (90:270:0.5cm)
                        (0,-1cm) arc (-90:-270:1cm);

    \end{tikzpicture}}
\newcommand{\intsym}{\yinyang}
\newcommand{\nointsym}{\tau}

\newcommand{\absblue}[1]{
	\ifthenelse{\boolean{colored}}{
		\blue{#1}}{
		#1}
}
\newcommand{\absred}[1]{
	\ifthenelse{\boolean{colored}}{
		\red{#1}}{
		#1}
}
\newcommand{\absdarkgreen}[1]{
	\ifthenelse{\boolean{colored}}{
		\darkgreen{#1}}{
		#1}
}

\renewcommand{\appsym}{\cdot}
\newcommand{\appsymp}[1]{\appsym^{#1}}

\renewcommand{\blueclr}{\absblue{\bullet}}
\renewcommand{\redclr}{\absred{\circ}}

\newcommand{\blackcol}{\blueclr}
\newcommand{\whitecol}{\redclr}

\newcommand{\bappsym}{\blueclr}
\newcommand{\rappsym}{\redclr}

\renewcommand{\rla}[1]{\absred{\lambda_{\redclr}} #1.}
\renewcommand{\bla}[1]{\absblue{\lambda_{\blueclr}} #1.}

\renewcommand{\rapp}[2]{#1 \rappsym #2}
\renewcommand{\bapp}[2]{#1 \bappsym #2}

\newcommand{\bnointsym}{\beta_{\nointsym}}
\newcommand{\bintsym}{\beta_{\intsym}}
\newcommand{\bchsym}{\beta_{\checkerssym}}

\newcommand{\rtobnoint}{\rootRew{\bnointsym}}
\newcommand{\rtobint}{\rootRew{\bintsym}}
\newcommand{\tobnoint}{\Rew{\bnointsym}}
\newcommand{\tobint}{\Rew{\bintsym}}

\newcommand{\checkerssym}{\absred\redclr \absblue\blueclr}
\renewcommand{\Lambdac}{{\Lambda_{\checkerssym}}}

\newcommand{\lap}[2]{\lambda_{#1} #2.}
\newcommand{\appp}[3]{#2 \appsymp{#1} #3}
\renewcommand{\ccapp}[3]{\appp{\clr{#1}}{#2}{#3}}
\renewcommand{\capp}[3]{\appp{\clrtwo{#1}}{#2}{#3}}

\newcommand{\typingruleAppInt}{@_{\intsym}}
\newcommand{\typingruleAppNoInt}{@_{\nointsym}}

\newcommand{\hnointsym}{\hsym_{\nointsym}}
\newcommand{\hintsym}{\hsym_{\intsym}}
\newcommand{\hchsym}{\hsym_{\checkerssym}}

\newcommand{\tohint}{\Rew{\hintsym}} 
\newcommand{\tohnoint}{\Rew{\hnointsym}}

\newcommand{\BT}[1]{{\sf BT}(#1)}
\newcommand{\restr}{\!\!\upharpoonright}
\newcommand{\head}{\relation{h}}
\newcommand{\comb}[1]{\symfont{#1}}
\newcommand{\emptyseq}{\langle\rangle}

\newcommand{\cB}{\mathcal{B}}

\newcommand{\rtoeta}{\rootRew{\eta}}
\newcommand{\toeta}{\Rew{\eta}}

\newcommand{\eqb}{=_\beta}
\newcommand{\eqeta}{=_\eta}

\newcommand{\clrd}[1]{\clr{#1}^\bot}

\newcommand{\bId}{\Id_{\blueclr}}
\newcommand{\bOne}{\comb{1}_{\blueclr}}
\newcommand{\rId}{\Id_{\redclr}}

\newcommand{\tobch}{\Rew{\bchsym}}
\newcommand{\tohch}{\Rew{\hchsym}}

\newcommand{\silconv}{=_{\bnointsym}}
\newcommand{\bconv}{=_{\beta}}

\newcommand{\btpre}{\sqsubseteq_{\cB}}

\newcommand{\speq}{=_{\spsym}}

\newcommand{\httsym}{\mathtt{h}}
\newcommand{\hdsize}[1]{| #1|_{\httsym} }

\newcommand{\hnf}{h}

\newcommand{\insize}[1]{\size{#1}_{@}}

\newcommand{\type}{\typefont{T}}

\newcommand{\bigmplus}{\biguplus}

\newcommand{\ruleMany}{\mathsf{many}}

\newcommand{\typefont}[1]{{\mathsf{#1}}}
\newcommand{\tvar}{\typefont{A}}

	\newcommand{\nftm}{\htm_{\tm}}
	\newcommand{\nftmtwo}{\htm_{\tmtwo}}

\newcommand{\med}[2]{
($(#1)!.5!(#2)$)
}

\tikzset{
ocenter/.style={baseline={([yshift=-.5ex, xshift=-.5ex]current bounding box)}},  
}

\newcommand{\leqchcolr}{\sqsubseteq_{\intsym}^{\mathrm{ctx\cdot imp}}}
\newcommand{\intrleq}{\sqsubseteq^{\mathrm{int\cdot imp}}}

\newcommand{\sbtsym}{\cB\eta_{\mathrm{red}}^\infty}
\newcommand{\slbtleq}{\sqsubseteq_{\sbtsym}}

\newcommand{\chcontexts}{\mathcal{C}_{\checkerssym}}
\newcommand{\chhcontexts}{\mathcal{H}_{\checkerssym}}

\newcommand{\evseq}{e}

\begin{document}

\title{Interaction Equivalence}

\author{Beniamino Accattoli}
\orcid{0000-0003-4944-9944}
\affiliation{%
  \institution{Inria - École Polytechnique}
  \city{Palaiseau}
  \country{France}
  }
\email{beniamino.accattoli@inria.fr}

\author{Adrienne Lancelot}
\orcid{0009-0009-5481-5719}
\affiliation{%
	\institution{Inria - École Polytechnique - Université Paris Cité}
	\city{Palaiseau}
	\country{France}
	}
\email{lancelot@irif.fr}

\author{Giulio Manzonetto}
\orcid{0000-0003-1448-9014}
\affiliation{
  \institution{Université Paris Cité}
  \city{Paris}
  \country{France}
}
\email{gmanzone@irif.fr}

\author{Gabriele Vanoni}
\orcid{0000-0001-8762-8674}
\affiliation{
  \institution{Université Paris Cité}
  \city{Paris}
  \country{France}
}
\email{gabriele.vanoni@irif.fr}

\renewcommand{\shortauthors}{Accattoli, Lancelot, Manzonetto, and Vanoni}

\begin{abstract}
Contextual equivalence is the \emph{de facto} standard notion of program equivalence. A key theorem is that contextual equivalence is an \emph{equational theory}. Making contextual equivalence more intensional, for example taking into account the time cost of the computation, seems a natural refinement. Such a change, however, does \emph{not} induce an equational theory, for an apparently essential reason: cost is not invariant under reduction.

In the paradigmatic case of the untyped $\lambda$-calculus, we introduce \emph{interaction equivalence}. Inspired by game semantics, we observe the number of interaction steps between terms and contexts but---crucially---ignore their internal steps. We prove that interaction equivalence is an equational theory and characterize it as $\mathcal{B}$, the well-known theory induced by B{\"o}hm tree equality. It is the first observational characterization of $\mathcal{B}$ obtained \emph{without} enriching the discriminating power of contexts with extra features such as non-determinism. To prove our results, we develop interaction-based refinements of the B{\"o}hm-out technique and of intersection types.
\end{abstract}



\maketitle

\section{Introduction}
A cornerstone of the theory of programming languages is the acceptance of contextual equivalence as a meaningful notion---if not as \emph{the} notion---of program equivalence.  Introduced by \citet{morris1968lambda} to study the untyped $\l$-calculus, contextual equivalence has the advantage that its definition is almost language-agnostic. Given a language $\lang$, one simply needs the notion of contexts $\ctx$ of $\lang$, usually defined as terms with a single occurrence of a special additional constant $\ctxhole$ (the \emph{hole} of the context), and a chosen predicate $\Downarrow$ of \emph{observation} for $\lang$, typically some notion of termination, which is the only language-specific ingredient. Then, contextual equivalence is defined as:
\[\begin{array}{cccc}
	t \obseq u &\mbox{ if for all contexts }C\,. &[\ 
	C\ctxholep{t}\conv{} \quad\Leftrightarrow\quad
	C\ctxholep{u}\conv{}
	\ ]
	\end{array}
\]

Contextual equivalence embodies a \emph{black box} behavioral principle with respect to termination. Nothing is said about the structure of $\tm$ and $\tmtwo$, it is only prescribed that any use of $\tm$ in a larger program can be replaced by $\tmtwo$ (and vice-versa) without affecting the observables. One of the most studied contextual equivalences is \emph{head contextual equivalence} $\obseqh$ for the untyped $\l$-calculus~\cite{Barendregt84}, where one observes termination of head reduction $\toh$. Another contextual equivalence that was studied at length is the one of Plotkin's \textsf{PCF}~[\citeyear{DBLP:journals/tcs/Plotkin77}], where the observation is termination on the same value of ground type, typically a natural number. Both played crucial roles in the historical development of denotational semantics of programming languages. 

Contextual equivalence is also relevant in more applied settings. The typical example being program transformations at work in compilers, for which contextual equivalence guarantees a strong form of soundness, see, for instance, the survey by \citet{DBLP:journals/csur/PatrignaniAC19}.

\paragraph{Equational Contextual Equivalence} In order for contextual equivalence $\obseq$ to be a proper semantical notion, one actually has to show that it is an \emph{equational theory} for $\mathcal{L}$. This means that the following two properties must be satisfied:
\begin{enumerate}
	\item \emph{Compatibility}: $\obseq$ is stable under context closure, that is, $\tm\obseq\tmtwo$ implies $\ctxp\tm\obseq\ctxp\tmtwo$ for all contexts $\ctx$, and; 
	\item \label{p:bconv} 
	\emph{Invariance}: $\obseq$ includes $\beta$-conversion $\eqb$, that is, if $\tm \eqb \tmtwo$ then $\tm \obseq \tmtwo$. 
\end{enumerate}
While compatibility holds by definition, invariance is a non-trivial theorem, because of the universal quantification on contexts in $\obseq$, which is notoriously hard to manage. For  head contextual equivalence $\obseqh$, for instance, a rewriting-based proof of the invariance property requires to prove the confluence of $\beta$-reduction and various non-trivial properties of head reduction\footnote{For an overview of such a proof, see \withproofs{\refapp{invariance}}\withoutproofs{Appendix A of the longer version on arXiv \cite{DBLP:journals/corr/abs-2409-18709}}.}. 

An equational theory for $\lang$ can be seen as a \emph{semantics} for $\lang$. The abstract requirement defining denotational models is that their induced equivalence\footnote{The equivalence induced by a model $\mathcal{M}$ is defined as $\tm =_{\mathcal{M}} \tmtwo$ if $\sem\tm=\sem\tmtwo$, where $\sem\cdot$ is the interpretation in the model.} is an equational theory. Notably, head contextual equivalence $\obseqh$ is exactly the equational theory of Scott's $\mathcal{D}_\infty$ model~[\citeyear{Scott72}], the first model of the untyped $\l$-calculus.

\paragraph{Termination $vs$ Time Cost, Equationally} It is natural to wonder whether contextual equivalence can be refined by replacing the termination predicate $\Downarrow$ with a finer one $\Downarrow^k$ indicating termination in $k$ evaluation steps. Intuitively, the number of evaluation steps is taken as a time cost model. With the compiler analogy in mind, it is indeed natural to ask that a program transformation preserves the time behavior. This refinement is exactly Sands' notion of \emph{cost equivalence}, the symmetric form of his more famous \emph{improvement preorder} \citeyearpar{SandsImprovementTheory,SandsToplas,DBLP:journals/tcs/Sands96}:
\[\begin{array}{cccc}
	t \costeq u &\mbox{ if for all contexts }C,\exists k\ge 0\,. &[\ 
	C\ctxholep{t}\convc{}{k} \quad\Leftrightarrow\quad
	C\ctxholep{u}\convc{}{k}
	\ ]
	\end{array}
\]
Cost equivalence falls short, however, of being a refinement of contextual equivalence. The price to pay for the added quantitative information is that it can no longer be seen as a semantics, that is, it is \emph{not} an equational theory. The reason is both simple and deep. It is simple because (in, say, the $\l$-calculus) if $\tm \tob \tmtwo$, then $\ctxp\tmtwo$ might take one step less to terminate than $\ctxp\tm$, so $\beta$-conversion $\eqb$ cannot be included in cost equivalence $\costeq$, breaking the invariance of equational theories.

\paragraph{Internal and External} The issue is deep as it reflects a more general tension between quantitative intensional properties such as time cost, which by their nature cannot be invariant under evaluation, and semantical notions such as equational theories, that are expected to hide the computational process. 
More generally, one might identify two perspectives on programs. The \emph{internal view} studies programs \emph{in isolation}, and it is concerned with qualitative properties such as, \eg, termination, confluence, productivity, or quantitative properties such as time or space cost. The \emph{external view}, instead, studies how programs \emph{interact} with other programs, as it is the case for contextual equi\-va\-len\-ces, labelled transition systems (LTSs), or process calculi. Usually, the external view hides the internal dynamics of programs, for instance via the invariance requirement for contextual equivalence, or via the silent $\tau$ transitions of LTSs. 

The mentioned issue with cost equivalence is then an instance of a more general question: is there a way of analyzing quantitative properties externally, without interference from the hidden internal dynamics? Concretely, is there a way of measuring the amount of \emph{interactions} between a term and its context \emph{modulo} the internal dynamics of the term and of the context?

\paragraph{Interaction Equivalence and the Checkers Calculus} In this paper, we provide a positive answer.  The literature on improvements has focused on application-oriented $\l$-calculi based on call-by-need evaluation. Our focus is different, more theoretical, we are rather interested in the semantical aspect and the understanding of the internal/external dilemma. Therefore, we place ourselves in a minimalistic setting, namely the ordinary (call-by-name) untyped $\l$-calculus, the denotational semantics of which has been studied in depth---see the classic book \cite{Barendregt84} and its recent extension \cite{BarendregtM22}---and aim at refining head contextual equivalence.

We introduce a new notion of cost equivalence, called \emph{(head) interaction equivalence}, whose key property is being---as we prove---an equational theory, in contrast to what happens for cost equivalence. This is obtained via a reconciliation of the internal and the external perspectives, loosely inspired by \emph{game semantics} \cite{DBLP:journals/iandc/HylandO00,DBLP:journals/iandc/AbramskyJM00}.

Interaction equivalence is based on the new \emph{checkers calculus} $\Lambdac$, a $\l$-calculus enriched with two players, black $\blackcol$ and white $\whitecol$. The idea is to duplicate the abstraction and application constructors of the $\l$-calculus as to have white ($\rla\var\tm$ and $\rapp\tm\tmtwo$) and black variants ($\bla\var\tm$ and  $\bapp\tm\tmtwo$) of each. Next, borrowing from LTSs, one defines two variants of $\beta$-reduction, the silent one (internal to a player) and the interaction (external) one, depending on whether the constructors involved in the redex belong to the same player or not:
\begin{center}
$\begin{array}{ccc\colspace |\colspace cccc}
\multicolumn{3}{c}{\textsc{Silent }\beta} & \multicolumn{3}{c}{\textsc{Interaction }\beta}
\\
\bapp{(\bla\var\tm)}\tmtwo& \tobnoint & \tm\isub\var\tmtwo
&
\bapp{(\rla\var\tm)}\tmtwo& \tobint & \tm\isub\var\tmtwo
\\
\rapp{(\rla\var\tm)}\tmtwo& \tobnoint & \tm\isub\var\tmtwo
&
\rapp{(\bla\var\tm)}\tmtwo& \tobint & \tm\isub\var\tmtwo
\end{array}$
\end{center}
Interaction equivalence of $\tm$ and $\tmtwo$ is then defined in the checkers calculus exactly as cost equivalence \emph{except} that one  uses the predicate $\bshcol{k}$ holding when checkers head reduction terminates using $k$ interaction steps $\tobint$ and---crucially---\emph{arbitrarily many} (possibly zero) silent steps $\tobnoint$.

Checkers interaction equivalence is then transferred to the ordinary $\l$-calculus via a \emph{paint-and-wash} construction: two ordinary $\l$-terms $\tm$ and $\tmtwo$ are interaction equivalent when their uniformly, say, black-painting $\monoToBlue\tm$ and $\monoToBlue\tmtwo$ are interaction equivalent in the checkers calculus, \ie:
\[\begin{array}{cccc}
	t \inteq u &\mbox{ if for all checkers contexts }C,\exists k\ge 0\,. &[\ 
	C\ctxholep{\monoToBlue{t}}\bshcol{k} \quad\Leftrightarrow\quad
	C\ctxholep{\monoToBlue{u}}\bshcol{k}
	\ ]
\end{array}
\]
Our first main result is that interaction equivalence is an equational theory of the ordinary untyped $\l$-calculus. As for contextual equivalence, proving that $\inteq$ is an equational theory is non-trivial, and the counting of interaction steps adds a further difficulty. We do it via a \emph{multi type system}, as explained after our second contribution. We also prove that two preorder variants of interaction equivalence, one of which is the interaction reformulation of Sands's improvements, are \emph{in}equational theories, that is, preorders verifying the compatibility and invariance of equational theories.

\paragraph{Inspecting Black Boxes} Contextual equivalences, as already mentioned, are black-box principles. While the concept is natural, it is hard to establish the contextual equivalence of two given terms, because of the universal quantification over all the contexts. It is then common to look for alternative, more explicit reformulations that relate the inner structure of equivalent terms. Historically, these reformulations were presented via semantic trees such as \emph{\bohm trees}, introduced by~\citet{Bare77}.
Head contextual equivalence $\obseqh$ for the untyped $\l$-calculus is one of the few cases for which an explicit description is available, as they are not easy to obtain.
Namely, $\obseqh$ was characterized by \citet{Hyland76} and \citet{Wadsworth76} as the equational theory $\laxbsym$ induced by the equality of \bohm trees up to possibly infinite $\eta$-equivalence. 
A natural question then arises: is there an explicit description of interaction equivalence? Does it correspond to anything already known?

\paragraph{Interaction Equivalence and \bohm Trees} The second contribution of our work is the answer to these two questions. We prove that interaction equivalence is exactly the well-known equational theory $\cB$ induced by \bohm tree equality. In particular, interaction equivalence is \emph{not} extensional, that is, it does not validate any form of $\eta$-equivalence, given that $\var$ and $\la\vartwo\var\vartwo$ clearly have a different number of interactions with any context providing an argument. Therefore, interaction equivalence $\inteq$ has a simpler explicit description than head contextual equivalence $\obseqh$, and provides an operational insight about $\eta$-equivalence. The failure of $\eta$-equivalence also reveals that, despite the analogy, our framework is not that close to game semantics, where $\eta$ is naturally validated unless additional constraints are imposed to avoid it, as in \cite{DBLP:conf/icalp/OngG02,DBLP:journals/tcs/KerNO03}.

Our results actually turn out to solve an open problem in the semantical theory of the untyped $\l$-calculus, namely the problem of finding a satisfying description of (non-extensional) \bohm tree equality $\cB$ as an observational equivalence. Partial solutions have been presented in the literature and are discussed among related works, in Section~\ref{sec:relatedworks}.
Ours is the first solution that does not require adding any extra features such as concurrency or non-determinism, but only refining the analysis. 

\paragraph{Normal Form Bisimilarities} Some readers might prefer a coinductive formulation of the problem. The equational theory $\cB$, indeed, corresponds to \emph{head normal form bisimilarity}, as shown by \citet{lassen1999bisimulation}, building over \citeauthor{DBLP:journals/iandc/Sangiorgi94}'s work~\citeyearpar{DBLP:journals/iandc/Sangiorgi94}. Normal form bisimilarities are known to be sound but---in general---\emph{not} complete for contextual equivalences. The open problem from the literature then might be recast as follows: is there an observational equivalence for which normal form bisimilarities are complete? We show that, for the head case, interaction equivalence is the answer.

\paragraph{Technical Development} Our results are proved via new interaction-based refinements of standard proof methods in the literature, namely the \bohm-out technique and multi types, overviewed in the next subsection. While the proofs are non-trivial, we believe that they are compact and neat, hopefully reassuring that the introduced framework is not \emph{ad-hoc}.

\withoutproofs{\paragraph{Proofs} Some proofs are omitted; they are in the proof appendix on arXiv \cite{DBLP:journals/corr/abs-2409-18709}}.


\subsection*{Overview of the Proof Techniques}
In the study of program equivalences, proving the inclusion of an equivalence into another one is always challenging. As it is customary, we study \emph{preorders} rather than equivalences, and prove the equality of the \bohm preorder $\btleq$ and the interaction preorder $\intleq$ by showing the two inclusions. 

\paragraph{Proof Technique 1: \bohm-Out}  For $\intleq \,\subseteq\, \btleq$, we prove the contrapositive: if $\tm$ and $\tmtwo$ have different \bohm trees then they are not interaction equivalent. We adapt the \emph{\bohm-out technique} at work in \bohm's separation theorem [\citeyear{Boehm68}]---a classic result of the untyped $\l$-calculus---thus building a context that separates $\tm$ and $\tmtwo$. The original technique, used also more recently in \cite{Dezani98,BOUDOL199683,IntrigilaMP19,BarendregtM22}, cannot distinguish $\eta$-equivalent terms. We refine it by counting interaction steps: when interacting with a context, $\eta$-equivalent terms give rise to different amounts of interaction.

\paragraph{Proof Technique 2: Multi Types} For the inclusion $\btleq \,\subseteq\, \intleq$, we use a different technique. The main tool is a new \emph{multi type} system (also known as non-idempotent intersection types) for the checkers calculus. Multi types are an established tool for the study of untyped $\l$-calculi, mediating between operational and denotational studies; see \cite{BKV17} for an introduction. 

Similarly to intersection types, they \emph{characterize} various termination properties, in the sense that $\tm$ ``converges'' if and only if $\tm$ is typable\footnote{The notion of convergence that is captured depends on the system that is considered. Here, we consider head normalization.}. In contrast to intersection types, however, multi types are \emph{quantitative}, which can be expressed in at least two ways. Firstly, the fact that typability implies (head) termination can be proved easily using as a decreasing measure the size of type derivations, while intersection types require more involved techniques. Secondly, multi types allow one to extract a bound of the number of head steps to normal form, as first shown by \citet{Carvalho07,DBLP:journals/mscs/Carvalho18}, which is impossible with (idempotent) intersection types. 

As a disclaimer, please note that multi types are a theoretical tool \emph{not} meant to be used in real life programming languages, since typability in multi type systems is an undecidable property. Rather, they provide handy type-theoretic presentations of denotational models.

We prove that our multi type system for the checkers calculus characterizes head termination, and we use the type system to introduce a type preorder $\leqbtype$ and show that $\btleq \,\subseteq\, \leqbtype \,\subseteq\, \intleq$. The first of these two inclusions is proved exploiting the quantitative properties of the multi type system. The second inclusion requires more, namely to characterize the multi type judgements from which one can measure the \emph{exact} number of interaction steps, rather than simply providing a bound. For that, we adapt the \emph{tight} technique of \citet{DBLP:journals/jfp/AccattoliGK20}, in its turn refining previous work by \citet{Carvalho07,DBLP:journals/mscs/Carvalho18}. A possibly interesting point is that the literature uses the tight technique to measure the number of internal steps, while here we use it dually, for measuring interaction steps. The following diagram sums up the technical development:
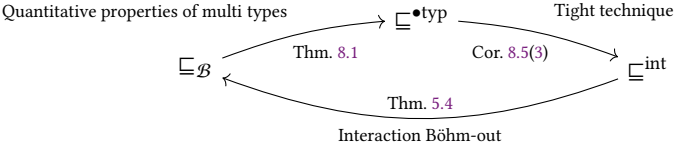
\begin{center}
\begin{tikzpicture}
		\node at (0,0)[align = center](bt){$\btleq$};
		\node at (bt.center)[right = 160pt](ctxc){ $\intleq$};
		\node at \med{bt}{ctxc}[above = 11pt](relc){ $\leqbtype$};
		
		\draw[->, out=20, in=185](bt) to node[above left= 1pt and 2pt, overlay] {\scriptsize Quantitative properties of multi types}node[below right = 1pt and -7pt] {\scriptsize Thm. \ref{thm:bisimulation-preserves-typeder}} (relc);
		\draw[->, out=-5, in=160](relc) to node[above right= 1pt and 2pt, overlay] {\scriptsize Tight technique}node[below left = 1pt and -7pt] {\scriptsize Cor. \ref{coro:type-preorder-included-ctxc}\eqref{coro:type-preorder-included-ctxc3}}(ctxc);
		\draw[out=-160, in=-20,->](ctxc) to node[below] {\scriptsize Interaction \bohm-out}node[above] {\scriptsize Thm. \ref{th:intleq-included-in-bohm}}(bt);
	\end{tikzpicture}
\end{center}
In fact, we end up providing \emph{two} characterizations of interaction equivalence/preorder, one as \bohm trees \emph{and} one as multi types.
Our first main result, namely the fact that interaction equivalence is an equational theory (Cor. \ref{coro:requirements-for-inequational}\eqref{coro:requirements-for-inequational3}), is also proved using tight multi types. Essentially, it is obtained by symmetry from the inclusion $\leqbtype\,\subseteq\,\intleq$ above and the fact that $\leqbtype$ is an inequational theory.

\ignore{\ben{\paragraph{Perspective on the Theory of Multi Types} Our use of multi types is novel on two aspects. Firstly, in the theory of normal form bisimilarities, the main technique for the inclusion $\btleq \subseteq\, \intleq$---usually referred to as the soundness of the bisimilarity---is Lassen's variant of Howe's method for applicative bisimilarities. We here show that multi types are an alternative proof technique for this kind of results.

Secondly, the tight technique of \citet{DBLP:journals/jfp/AccattoliGK20} aims at counting the exact number of steps for the \emph{internal} dynamics of terms, and for that it provides types that forbid external interactions. In other words, the internal and external views are kept sealed and apart. Our use of the tight technique in the checkers type system aims at counting the exact number of interaction/\emph{external} steps. can be seen as an improvement, since therein the players tags allow one 
}}
\begin{figure}[t!]
\begin{tabular}{c|c}
$\begin{array}{rrllll}
\Lambda\ni &  t,u,s& ::=& x\,(\in\Var)\mid \lam x.t\mid tu
\\
\textsc{hnfs} & \htm & \grameq & \lambda x_{1}\ldots x_{n}.y\,t_{1}\cdots t_{k}
\\[4pt]
\Cctx\ni & 	C& ::=& \ctxhole\mid \lam x.C\mid \tm\,C \mid C\,\tmtwo
\\
\Hctx\ni & 	\hctx& ::=& \lam x_1\dots x_n.\ctxhole t_1\cdots t_k

\\
\\
\end{array}$
&
$\begin{array}{rrlllllll}
\textsc{$\beta$-rule} & (\lam x.t)u &\rtob& t\isub{x}{u}
\\
\textsc{$\eta$-rule} & \lam x.tx &\rtoeta& t, \textrm{if }x\notin\FV{t}
\\[4pt]
\textsc{$\beta$} & \tob &\defeq& \Cctxp\rtob
\\
\textsc{head} & \toh &\defeq& \hctxp\rtob
\\
\textsc{$\eta$} & \toeta &\defeq& \Cctxp\rtoeta

\end{array}$

\end{tabular}
\vspace{-5pt}
\caption{The $\l$-calculus.
$\Lambda$ is the set of \lam-terms, $\Cctx$ the class of contexts, $\Hctx$ the subclass of head contexts.}
\label{fig:lambda}
\end{figure}
\section{The $\lambda$-Calculus} 
To keep this article as self-contained as possible, we summarize some definitions and results concerning \lam-calculus that we shall use in the paper.
For more information, see \cite{Barendregt84}.
\begin{definition}
The syntax and rewriting rules of $\l$-calculus are given in \reffig{lambda}.
\end{definition}
\paragraph{$\l$-Terms and Contexts} 
The set $\Lambda$ of \emph{\lam-terms} is constructed over a countable set $\Var$ of variables. 
We assume that application associates to the left and has a higher precedence than abstraction.
Given $k,n\ge 0$ and $t,u, s_1,\dots,s_k\in\Lambda$, we write $\lam \vec x.t$ as an abbreviation of $\lam x_1\dots x_n.t$, and $t\vec s$ for $ts_1\cdots s_k$.
For instance, $\lam xywz.xy(wz)$ stands for $\lam x.(\lam y.(\lam w.(\lam z.((xy)(wz)))))$.

The set $\FV{t}$ of \emph{free variables} of $t$ and \emph{$\alpha$-conversion} are defined as in~\cite[\S1.2]{Barendregt84}.
Hereafter, we consider \lam-terms up to $\alpha$-conversion and we denote $\alpha$-conversion simply by $=$. The usual meta-level capture-avoiding substitution of $u$ for $x$ in $t$ is noted $t\isub{x}{u}$.

The class $\Cctx$ of \emph{contexts} contains $\lambda$-terms built using exactly one occurrence of the \emph{hole} $\ctxhole$, standing for a removed sub-term. 
We shall also use the subclass $\Hctx$ of \emph{head contexts}.
Given a context $C$ and a \lam-term $t$, we denote by $C\ctxholep{t}$ the \lam-term obtained by  replacing $t$ for the hole $\ctxhole$ in $C$, possibly with capture of free variables. 

\paragraph{Rewriting.}
The \lam-calculus can be endowed with several \emph{notions of reduction} $\to_{\relation{R}}$, turning the set $\Lambda$ into a higher-order term rewriting system.
Given a notion of reduction $\to_{\relation{R}}$: 
\bsub\item $\to^*_{\relation{R}}$ stands for the reflexive and transitive closure of $\to_{\relation{R}}$ (\emph{multistep $\relation{R}$-reduction});
\item $=_{\relation{R}}$ stands for its reflexive, symmetric, and transitive closure (\emph{$\relation{R}$-conversion});
\esub
We write $\tm\to_{\relation{R}}^n u$ to indicate a reduction sequence $\tm\to_{\relation{R}}\tm_1\to_{\relation{R}}\cdots\to_{\relation{R}}\tm_{n-1}\to_{\relation{R}} u$ of length $n$.
A \lam-term $\tm$ is an \emph{$\relation{R}$-normal form} (or \emph{$\relation{R}$-nf}) if there is no term $u$ such that $\tm\to_{\relation{R}} u$.
Given $\tm\in\Lambda$, we write $\tm\conv{\relation{R}} u$ if $\tm\to^*_{\relation{R}}u$ and $u$ is a $\relation{R}$-nf.
The term $\tm$ is \emph{$\relation{R}$-normalizable}, written $\tm\conv{\relation{R}}$, if $\tm\conv{\relation{R}} u$, for some $u$.


\paragraph{Contextual Closure and Reductions.} We distinguish between the \emph{(rewriting) rule} $\relation{R}$ and the \emph{notion of reduction} $\to_{\relation{R}}$, where the latter is obtained from the former via some form of context closure.
The context closure defined next is applicable more generally to every relation $\relsym\subseteq\Lambda^2$ on $\l$-terms.

\begin{definition}[Contextual closure]\label{def:ctxclosure}
Let $\relsym\subseteq\Lambda^2$ be a relation and $\Cctxtwo$ be a class of contexts. The \emph{$\Cctxtwo$-closure of $\relsym$} is the least relation $\Cctxtwop{\relsym}$ such that $\tm\rel\tmtwo$ entails $\ctxp\tm \ \Cctxtwop{\relsym}\ \ctxp\tmtwo$, for all $\ctx\in\Cctxtwo$. $\relsym$ is called \emph{$\Cctxtwo$-compatible} if $\Cctxtwop{\relsym} \subseteq\relsym$, and simply \emph{compatible} when $\Cctxtwo$ is the class of all contexts $\Cctx$. 
\end{definition}

The \emph{$\beta$-reduction} $\tob$ (resp.\ \emph{$\eta$-reduction} $\toeta$) is the closure of the $\beta$-rule ($\eta$-rule) under all contexts. 

\begin{notation}\label{not:lam-terms} 
Concerning specific \lam-terms, we fix the following notations:
\begin{center}$
	\begin{array}{cccccc}
	\comb{I}\defeq\lam x.x,&	\comb{1}\defeq\lam xy.xy,&\comb{K}\defeq\lam xy.x,&\comb{F}\defeq\lam xy.y,&\comb{Y} \defeq \lam f.(\lam x.f(xx))(\lam x.f(xx)),&\Omega \defeq \comb{YI},
	\end{array}$
\end{center}
where $\comb{I}$ is the identity, $\comb{1}$ is an $\eta$-expansion of $\comb{I}$, $\comb{K}$ and $\comb{F}$ are the projections, $\comb{Y}$ is a fixed point operator satisfying $\comb{Y}f =_\beta f(\comb{Y}f)$, and $\Omega=_\beta (\lam x.xx)(\lam x.xx)$ is the paradigmatic looping combinator.
\end{notation}

Head reduction $\toh$, defined in \reffig{lambda}, is the evaluation strategy adopted in this paper. We choose head reduction because of its key role in the semantics of $\l$-calculus \cite{Barendregt84}, but our construction of interaction equivalence could be adapted smoothly to weak (\ie not under abstraction) head reduction. We shall discuss it in Section \ref{sect:future-work}.

\reffig{lambda} gives the standard characterization of head normal forms (or \emph{hnfs}). Given a hnf $\htm=\lambda x_{1}\ldots x_{n}.y\,t_{1}\cdots t_{k}$, we refer to $\vartwo$ (which may possibly be one of $x_{1},\ldots, x_{n}$) as to its \emph{head variable}.

\paragraph{Conversion and Equational Theories.} The equational theories of the \lam-calculus, called \emph{\lam-theories}, are compatible equivalence relations containing $\beta$-reduction. 
They arise naturally when one aims at equating \lam-terms displaying the same operational behavior.
Similarly, inequational theories express the fact that the behavior of a \lam-term is somewhat \emph{less defined} than the behavior of another term.

\begin{definition} 
\bsub
\item A relation $\relsym\subseteq\Lambda^2$ is called a \emph{congruence} if it is a compatible equivalence.
\item We say that $\relsym$ is \emph{$\beta$-invariant} if it contains $\beta$-conversion $\eqb$.
\item	An \emph{equational theory}, or \emph{\lam-theory}, is any $\beta$-invariant congruence $=_{\eqth}$. 
\item
	An \emph{inequational theory} is any compatible $\beta$-invariant preorder $\sqsubseteq_\eqth$.
\item
	An (in)equational theory is \emph{consistent} if it is different from $\Lambda^2$, \emph{extensional} if it contains $\eqeta$.
\item An inequational theory $\sqsubseteq_\eqth$ is \emph{semi-extensional} if it contains $\eta$-reduction $\to_\eta$.
\esub
%
\end{definition}
Note that inequational \lam-theories are not required to be symmetric---they are preorders---and yet they are required to contain $\beta$-conversion $\eqb$, which is symmetric.
Any inequational \lam-theory $\sqsubseteq_\eqth$, induces a \lam-theory denoted by $=_{\eqth}$ by setting: $\tm=_\eqth\tmtwo$ if both $\tm\sqsubseteq_\eqth\tmtwo$ and $\tmtwo\sqsubseteq_\eqth\tm$ hold.

\paragraph{Contextual Preorders and Equivalences} A nowadays standard notion in the study of programming languages and $\l$-calculi is contextual equivalence, and its asymmetric variant of contextual preorder. The idea is to observe the termination of \lam-terms when plugged in the same context. Thus, they depend on a notion of termination. In this paper, we focus on termination of head reduction.
\begin{definition}
The \emph{(head) contextual preorder} $\obsle$ on $\l$-terms is defined as follows:
\begin{center}$\begin{array}{cccc}
	t \obsle u &\mbox{ if for all contexts }C\,. &[\ 
	C\ctxholep{t}\conv{\hsym} \quad\Rightarrow\quad
	C\ctxholep{u}\conv{\hsym}
	\ ]
	\end{array}
$\end{center}
The associated \emph{(head) contextual equivalence} is defined by setting $t \obseq u$ if $t \obsle u$ and $u \obsle t$.
\end{definition}
Often the literature requires $\ctx$ to be a \emph{closing} context, that is, such that both $\ctxp\tm$ and $\ctxp\tmtwo$ are closed. In the ordinary $\l$-calculus, adding/removing the closing requirement does not change the defined relation, which is why we do not add it\ignore{ (proving that the two notions are equivalent is however not immediate, see \cite{DBLP:conf/fossacs/AccattoliL24})}.
It turns out that the head contextual preorder provides an inequational $\l$-theory.

\begin{theorem}[\cite{Barendregt84}]
\label{thm:head-ctx-inequational}
The head contextual preorder $\obsle$ is an inequational $\l$-theory. Moreover, it is consistent and extensional.
\end{theorem}

As mentioned in the introduction, the proof of the previous theorem is non-trivial. Proving that $\beta$-conversion $\eqb$ is included in $\obsle$ indeed is not immediate, because of the famously fastidious universal quantification on contexts in $\obsle$. A rewriting-based proof requires the use of both confluence of $\beta$ and the untyped normalization theorem of head reduction (if $\tm\tob^*\tmtwo$ and $\tmtwo$ is normal for head reduction, then $\tm$ is head normalizing). For an overview, see \withproofs{\refapp{invariance}}\withoutproofs{Appendix A of the longer version on arXiv \cite{DBLP:journals/corr/abs-2409-18709}}. Similarly, a rewriting-based proof that $\eta$-conversion $\eqeta$ is included in $\obsle$ also rests on non-trivial theorems about $\eta$.  Also in this case, for an overview see \withproofs{\refapp{invariance}}\withoutproofs{Appendix A of the longer version on arXiv \cite{DBLP:journals/corr/abs-2409-18709}}.

Semantic proofs are possible but not easy anyway, as they rest on soundness of the model. For the similar case of interaction equivalence, we shall develop a semantic proof.

\paragraph{Background and Notable Variants} The head contextual preorder has been studied in-depth because it captures the (in)equational theory of Scott's denotational model $\mathcal{D}_\infty$~\cite{Scott72}, the first model of the untyped $\l$-calculus. 
The preorder $\obsle$ is also the maximal contextual preorder of interest: any strictly larger inequational \lam-theory is inconsistent \cite[Lemma~12.5]{BarendregtM22}.
In 1976, \citeauthor{Hyland76} and \citeauthor{Wadsworth76} have shown that the associated equivalence $\obseq$ coincides with Böhm trees equality up to inﬁnite $\eta$-expansions (Cf.\ Theorem~\ref{thm:HW76}, below).

Another relevant contextual preorder is obtained by considering termination with respect to $\beta$-reduction, instead of head reduction. This alternative preorder was originally introduced in Morris's PhD thesis [\citeyear{morris1968lambda}], but has been the subject of fewer investigations than the head one (until recently, see  \cite{IntrigilaMP19} and \cite[Ch.~12]{BarendregtM22}). \citet{Hyland75} showed that the associated equivalence captures B\"ohm trees equality up to \emph{finite} $\eta$-expansions.

\section{The Checkers Calculus}
\label{sec:checkers-calculus}

In this section we introduce the \emph{checkers calculus}, obtained from the $\l$-calculus by duplicating the abstraction and application constructors, that now both come in white and black dresses.

\begin{definition} The syntax and operational semantics of the checkers calculus are defined in \reffig{col-def}.
\end{definition}

\paragraph{Black and White Constructors.} 
The constructors of abstraction and application receive a tag, called \emph{player}, that can be either white $\redclr$ or black $\blueclr$. 
The terms populating the set $\Lambdac$ are called \emph{checkers terms}, and inherit the notions of free variables, $\alpha$-conversion, and substitution from \lam-calculus.
Note that variables do \emph{not} receive a player tag. The main reason is simplicity: this way, we do not need to enforce the uniform tagging of all the occurrences of the same variable.

\begin{notation}
When the players are actually specified, we simply denote the applications by $\rapp\tm\tmtwo$ (and $\bapp\tm\tmtwo$) instead of $\appp\redclr\tm\tmtwo$ (and $\appp\blueclr\tm\tmtwo$).
Hereafter, we often need to refer to constructors of a fixed but arbitrary player $\colr$, thus we use $\lap{\colr}{\var}\tm$ and $\appp{\colr}{\tm}\tmtwo$ for abstractions and applications of player $\colr\in\{\redclr,\blueclr\}$. 
We also use $\colr^\bot$ to denote the opposite player, defined as $\redclr^\perp \defeq \blueclr$ and $\blueclr^\perp \defeq\redclr $. 
In case of many consecutive abstractions or applications, we shorten the notations to $\clavec{\colr}{\var}\tm$ and $\cappvec{\colr}{\tm}{\tmtwo}$, respectively. 
The former is sometimes slightly expanded to $\manyclam{k}{\var}\tm$, and the latter to $\appp{\clr{1}\cdots\clr{k}}{\tm}{\tmtwo_1\cdots \tmtwo_{k}}$. 
\end{notation}


\begin{figure}
\arraycolsep=3.5pt
\tabcolsep=2pt
\centering
		\begin{tabular}{c|c}
		\textsc{Terms and Contexts}			
			&			
			\textsc{Beta rules and reductions}
			\\[2pt]\hline\\[-8pt]
			\begin{tabular}{c}
			$\begin{array}{rrllllll}
			\textsc{Players} &			\clr{},\clrtwo{} & \grameq & \redclr \mid \blueclr
			\\
			\textsc{Terms} \ \Lambdac\ni & \tm,\tmtwo,\tmthree & \grameq& \var \mid \lap\colr{\var}\tm  \mid \appp\colr\tm\tmtwo\\[4pt]
			\textsc{$\hchsym$-nfs} & \htm,\htmtwo &\grameq &
			\clavec\colr\var\appp{\clrtwo{1}\cdots\clrtwo{k}}\vartwo{\tm_1 \cdots \tm_k}
		\end{array}$ 
		\\[28pt]
			$\begin{array}{rrllllll}
			\chcontexts\ni& \ctx & \grameq& \ctxhole \mid \cla{}{\var}\ctx \mid \ccapp{}\ctx\tmtwo \mid \ccapp{}\tm\ctx 
			\\
			\chhcontexts\ni& \hctx & \grameq & \clavec\colr\var\appp{\clrtwo{1}\cdots\clrtwo{k}}\ctxhole{\tm_1 \cdots \tm_k}
\\[4pt]
		\end{array}$ 
		\end{tabular}
		& 
		\begin{tabular}{c}
		$\begin{array}{r\hcolspace rll}
			\textsc{Silent} & \ccapp{}{(\cla{}\var\tm)}{\tmtwo} &\rtobnoint& \tm\isub\var{\tmtwo}
			\\
			\textsc{Interaction} & \appp{\clrd{}}{(\cla{}\var\tm)}\tmtwo &\rtobint& \tm\isub\var\tmtwo
			\end{array}$
			\\[9pt]
			$\begin{array}{r@{\hspace{.25cm}} rll}
			\textsc{Silent $\beta$} & \tobnoint &\defeq& \chcontexts\ctxholep\rtobnoint
			\\
			\textsc{Interaction $\beta$} & \tobint &\defeq& \chcontexts\ctxholep\rtobint
			\\
			\textsc{Checkers $\beta$} & \tobch &\defeq& \tobnoint \cup \tobint
			\\
			\textsc{Silent head} & \tohnoint &\defeq& \chhcontexts\ctxholep\rtobnoint
			\\
			\textsc{Interaction head} & \tohint &\defeq& \chhcontexts\ctxholep\rtobint
			\\
			\textsc{Checkers head} & \tohch &\defeq& \tohnoint \cup \tohint

		\end{array}$
		\end{tabular}
		\end{tabular}
	\caption{The checkers calculus.}
	\label{fig:col-def}	
\end{figure}

\paragraph{Checkers Contexts.} The class $\chcontexts$ of \emph{checkers contexts}, and the subclass $\chhcontexts$ of \emph{checkers head contexts} are defined in \reffig{col-def}. 
\refdef{ctxclosure} of contextual closure with respect to a class of contexts, generalizes to this setting in the obvious way. 
Checkers contexts shall play a key role in the definition of contextual preorders and equivalences on checkers terms (see \refdef{interactionpreorder}, below).

\paragraph{Silent and Interaction Steps.} There are two kinds of colored $\beta$-redexes $\capp{}{(\cla{}\var\tm)}{\tmtwo}$, the silent one and the interaction one, each one with its own $\beta$-rule. 
Silent redexes are $\beta$-redexes where the color of the abstraction $\clr{}$ matches the color of the application $\clrtwo{}$. Intuitively, these steps are internal to each player's world. 
In interaction redexes, instead, the color of the abstraction and the color of the application are different, \ie $\clr{}\neq \clrtwo{}$. 
This represents the scenario where the two players interact with each other, which, from each player's perspective, amounts to interacting with the external world. Our focus shall be on the number of head interaction steps.

The internal/external dichotomy, and the idea of having two players are strong guiding intuitions but note that, in general, a checkers term can arbitrarily interleave black and white constructors: there is no neat frontier between the parts of a term corresponding to the two players.
The key point is that, even if we start with cleanly separated black and white parts, they still end up mixing during the reduction.

\begin{example}\label{ex:delta} 
Consider the black identity $\bId \defeq \bla\var\var$ and the white term $\comb{D}_\circ \defeq \lam_\circ y.\lam_\circ x.x\circ (y\circ x)$. 
\begin{enumerate}
\item
If the black identity is black-applied to $\comb{D}_\circ$, then it gives rise to a silent step $\bapp{\bId} \comb{D}_\circ \tobnoint \comb{D}_\circ$, while if it is white-applied to $\comb{D}_\circ$ then the step is an interaction one, namely $\rapp{\bId} \comb{D}_\circ \tobint \comb{D}_\circ$. 
\item $\bapp{\bapp{\comb{D}_\circ}{\bId}}{\bId} \tohint \bapp{(\lam_\circ x.x\circ (\bId\circ x))}{\bId} \tobint \bapp{(\lam_\circ x.x\circ  x)}{\bId}\tohint \rapp{\bId}{\bId}\tohint\bId$
\end{enumerate}
\end{example}

\paragraph{Basic Rewriting Properties} As the $\l$-calculus, the checkers $\l$-calculus is an example of \emph{orthogonal higher-order rewriting system} \cite{AczelHO,phdklop,DBLP:conf/lics/Nipkow91}, that is a class of rewriting systems for which confluence always holds, because of the good shape of their rewriting rules. Similarly, the silent and interaction sub-relations $\tobnoint$ and $\tobint$ are also confluent and commute. 

\begin{theorem}[Confluence]
Reductions $\tobch$, $\tobnoint$, and $\tobint$ are confluent.
\end{theorem}

Similarly, checkers head reduction inherits various expected properties of head reduction, such as determinism. In particular, we shall use the following immediate substitutivity property.

\begin{lemma}[Substitutivity]
	\label{l:color-substitutivity}
	Let $\tm,\tmp,\tmtwo\in\Lambdac$ and $\relation{R}\in\{\bnointsym,\bintsym,\bchsym,\hnointsym,\hintsym,\hchsym\}$. If $\tm\Rew{\relation{R}}\tmp$ then $\tm\isub\var\tmtwo\Rew{\relation{R}}\tmp\isub\var\tmtwo$. 
\end{lemma}

\paragraph{$\eta$-Conversion} Note that there are no checkers $\eta$-rules in \reffig{col-def}. A key point in our work, indeed, is that $\eta$-conversion can change the number of interaction steps.
\begin{example}[The Delicate Role of $\eta$]\label{ex:eta}
Consider the black $\eta$-expansion  $\bOne \defeq\bla\var\bla\vartwo{\bapp\var\vartwo}$ of the black identity $\bId$ and the following diagram:
\begin{center}
\begin{tikzpicture}
		\node at (0,0)[align = center](source){$\rapp{\rapp{\bOne}\varthree}\varfour$};
		\node at (source.east)[above right = 6pt and 25pt](source-up-right){ $\rapp{(\bla\vartwo\bapp\varthree\vartwo)}\varfour$};
		\node at (source-up-right.east)[right = 25pt](source-up-right-right){ $\bapp\varthree\varfour$};
		\node at (source.east)[below right = 6pt and 25pt](source-down-right){ $\rapp{\rapp{\bId}\varthree}\varfour$};
		\node at (source-up-right-right|-source-down-right)(source-down-right-right){$\rapp\varthree\varfour$};
		\node at \med{source-up-right-right}{source-down-right-right}(neq){$\neq$};
		\draw[->, out =35, in =185](source.north east) to node[above] {\scriptsize $\hintsym$} (source-up-right.west);
		\draw[->, out =-35, in =175](source.south east) to node[above] {\scriptsize $\blackcol\eta$}(source-down-right.west);
		\draw[->](source-down-right) to node[above] {\scriptsize $\hintsym$}(source-down-right-right);
		\draw[->](source-up-right) to node[above] {\scriptsize $\hintsym$}(source-up-right-right);
	\end{tikzpicture}
	\end{center}
It shows that $\eta$ changes the number of interaction steps. It also shows that standard properties of ordinary $\eta$ do not lift to the checkers case (see \withproofs{\refapp{invariance}}\withoutproofs{Appendix A on arXiv \cite{DBLP:journals/corr/abs-2409-18709}} for definitions): $\to_\eta$ cannot be postponed after $\tohint$, $\to_\eta$ and $\tohint$ do not commute, and adding $\eta$ to $\tobch$ breaks confluence.
\end{example}

\paragraph{Big-Step Notation and Interaction Index.} 
Recall that $\tm\bshcols$ stands for ``$\tm$ is $\hchsym$-normalizing''. 
We introduce the notation $\tm\bshcol{k}$ for: $\tm$ is $\hchsym$-normalizing \emph{and} the number of head interaction steps $\tohint$ in its $\hchsym$-evaluation is $k$. Note that the number $k$ is well-defined because $\tohch$ is deterministic.

\paragraph{Interaction Equivalence and Preorders} We now introduce interaction equivalence $\equivchcol$ as a form of quantitative contextual equivalence for the checkers calculus. Similarly to weak similarity for labeled transition systems, $\equivchcol$ ignores silent head steps. The quantitative aspect is that it requires to preserve the number of interaction head steps. The richer quantitative setting in fact gives rise to \emph{two} possible preorders, both generating $\equivchcol$ when symmetrized.

\begin{definition}[Checkers interaction preorders and equivalence]\label{def:interactionpreorder}
	We define the \emph{interaction preorder} $\leqchcol$, the \emph{interaction improvement (preorder)} $\leqchcolr$, and the \emph{interaction equivalence} $\equivchcol$ on checkers terms $\tm,\tmtwo\in\Lambdac$ as follows:
	\bsub
		\item $\tm \leqchcol \tmtwo$ if $\ctxp{{\tm}}\bshcol{k}$ implies $\ctxp{\tmtwo}\bshcol{k}$ for all checkers contexts $\ctx\in\chcontexts$ and $k\in\nat$;
		\item $\tm \leqchcolr \tmtwo$ if $\ctxp{{\tm}}\bshcol{k}$ implies $\ctxp{\tmtwo}\bshcol{k'}$ with $k' \leq k$, for all contexts $\ctx\in\chcontexts$ and $k\in\nat$;
		\item ${\tm} \equivchcol {\tmtwo}$ is the equivalence relation induced by $\leqchcol$, that is, ${\tm} \equivchcol {\tmtwo}$ if  ${\tm} \leqchcol {\tmtwo}$ and ${\tmtwo} \leqchcol {\tm}$.
	\esub
\end{definition}
The interaction improvement $\leqchcolr$ adapts Sands' improvements \citeyearpar{SandsImprovementTheory,SandsToplas,DBLP:journals/tcs/Sands96} to our \emph{up to silent steps} setting. Note that $\leqchcol\,\subseteq\, \leqchcolr$. It is easy to see that the two new preorders are different, \ie that $\leqchcol\,\subsetneq\, \leqchcolr$. Indeed, $\rapp\bId\rId \leqchcolr\rId$ but $\rapp\bId\rId \not\leqchcol\rId$, as both checkers terms have $\rId$ as head normal form but $\rapp\bId\rId$ requires one more interaction step to reach it. 

The interaction preorder $\leqchcol$ turns out to be more easily manageable than the interaction improvement $\leqchcolr$. Moreover, from the inclusion $\leqchcol\,\subseteq\, \leqchcolr$ we shall be able to transfer some properties of $\leqchcol$ to $\leqchcolr$. In this paper, then, we shall rather focus on $\leqchcol$.

\begin{example}\label{ex:whitedelta} 
Recall that $\bId$ and $\comb{D}_\circ$ have been defined in \refex{delta}, and $\bOne$ in \refex{eta}.
\bsub
\item $(\lam_\circ x.x\circ x)\circ(\lam_\circ x.x\circ x)\leqchcol \bId$, because the former is not $\hchsym$-normalizable.
\item\label{ex:whitedelta2} $\bId\not\leqchcol\bOne$ and $\bOne\not\leqchcol\bId$, as they are separated by $C = \ctxhole\circ\varthree\circ\varfour$, see \refex{eta}.
\item $\comb{D}_\circ\circ(\lam x_\circ.x)\circ(\lam x_\circ.x) \equivchcol (\lam x_\circ.x)$, because they are $\bnointsym$-convertible.
\esub
\end{example}

\paragraph{Back to Ordinary $\lambda$-Terms}
Composing the interaction relations for checkers terms defined above with the \emph{black embedding} of $\Lambda$ into $\Lambdac$, we obtain new \emph{interaction} relations on ordinary $\l$-terms.

\begin{definition}[Player Lifting and interaction preorders/equivalence]\ 
\bsub
\item
Given $\colr\in\{\redclr,\blueclr\}$, define the \emph{player $\colr$-lifting of ordinary $\l$-terms and ordinary contexts}, as the maps $\monoToPlayer{\colr}{\cdot} : \Lambda \rightarrow \Lambdac$  and $\monoToPlayer{\colr}{\cdot} : \mathcal{C} \rightarrow \chcontexts$ obtained by $\colr$-tagging every constructor:
\begin{center}
	$\begin{array}{c}
	\begin{array}{c@{\hspace{0.1cm}}c@{\hspace{0.1cm}}c\colspace\colspace c@{\hspace{0.1cm}}c@{\hspace{0.1cm}}c\colspace\colspace c@{\hspace{0.1cm}}c@{\hspace{0.1cm}}c\colspace\colspace c@{\hspace{0.1cm}}c@{\hspace{0.1cm}}c}
	\monoToPlayer{\colr}{\var} & \defeq & \var, 
	&
	\monoToPlayer{\colr}{\la\var\tm} & \defeq & \lap\colr\var\monoToPlayer{\colr}\tm,
	&
	\monoToPlayer{\colr}{\tm\tmtwo} & \defeq & \appp{\colr}{\monoToPlayer{\colr}{\tm}}{\monoToPlayer{\colr}{\tmtwo}};
	
	\\[3pt]
	
		\monoToPlayer{\colr}{\ctxhole} & \defeq & \ctxhole, 
		&
		\monoToPlayer{\colr}{\la\var\ctx} & \defeq & \lap\colr\var\monoToPlayer{\colr}\ctx,
		&
		\monoToPlayer{\colr}{\tm\ctx} & \defeq & \appp{\colr}{\monoToPlayer{\colr}{\tm}}{\monoToPlayer{\colr}{\ctx}},
		&
		\monoToPlayer{\colr}{\ctx\tmtwo} & \defeq & \appp{\colr}{\monoToPlayer{\colr}{\ctx}}{\monoToPlayer{\colr}{\tmtwo}}.
	\end{array}
	\end{array}$
\end{center}
\item 
The \emph{interaction preorder} $\intleq$, the \emph{interaction improvement} $\intrleq$, and the \emph{interaction equivalence} $\inteq$ are the following relations on ordinary $\l$-terms $t,u\in\Lambda$ defined via black lifting:
\begin{itemize}
	\item $\tm \intleq \tmtwo$ if $\monoToBlue{\tm} \leqchcol \monoToBlue{\tmtwo}$; 
	\item $\tm \intrleq \tmp$ if $\monoToBlue{\tm} \leqchcolr \monoToBlue{\tmtwo}$;
	\item ${\tm} \inteq {\tmtwo}$ is the equivalence relation induced by $\intleq$, that is, ${\tm} \inteq {\tmtwo}$ if ${\tm} \intleq {\tmtwo}$ and ${\tmtwo} \intleq {\tm}$.
\end{itemize}
\esub
\end{definition}

\begin{example} We use the \lam-terms introduced in \refnot{lam-terms}. We also consider $\comb{D} \defeq \lam yx.x(yx)$, having the property that $\comb{YD}$ is a fixed point operator whose head reduction is `slower' than that of $\comb{Y}$.
\bsub
\item $\comb{YK}\intleq t$, for all $t\in\Lambda$, because $\comb{YK}$ does not have an hnf. Similarly, $\comb{YI} =_\beta \Omega\intleq t$.
\item We have $\comb{I}\not\intleq \comb{1}$, nor $\comb{1}\not\intleq \comb{I}$  (Cf.\ \refex{whitedelta}\eqref{ex:whitedelta2}).
It is easily seen that $\comb{K}\not\intleq\comb{F}$ and $\comb{F}\not\intleq\comb{K}$.
\item $\comb{YD} \inteq\comb{Y}$, as the former needs more head-reduction steps to converge, but they are silent.
\esub
\end{example}

\paragraph{The Interaction Preorder is Inequational.} We now show that the interaction preorder $\intleq$ is a semantics of $\l$-calculus, \ie, it is an inequational $\l$-theory, from which the corresponding results for $\intrleq$ and $\inteq$ follow. This is our first main result, showing that our framework does solve the internal/external tension evoked in the introduction.

As for contextual equivalence, proving invariance (that is, that $\beta$-conversion is included) is non-trivial, and even harder because of the constraint on the number of interaction steps. It follows from the following theorem for the checkers interaction preorder, which we shall prove in a later section via a semantic proof based on multi types (page \pageref{app:delayed-proof}), rather than via rewriting-based techniques.

\begin{toappendix}
\begin{theorem}[The interaction preorder includes silent conversion]
\label{thm:silconv-included-in-inter-ctx}
For all checkers terms $t,u\in\Lambdac$, $t \silconv u$ entails $t \leqchcol u$. 
\end{theorem}
\end{toappendix}

\begin{corollary}
	\label{coro:requirements-for-inequational}
	\hfill
	\begin{enumerate}
	\item \label{coro:requirements-for-inequational1} The interactional preorder $\intleq$ is a consistent inequational $\lambda$-theory. Moreover, it is not semi-extensional, whence not extensional.
	\item The interactional improvement $\intrleq$ is a consistent inequational $\lambda$-theory.
	\item \label{coro:requirements-for-inequational3} Interaction equivalence $\inteq$ is a consistent $\l$-theory. Moreover, it is not extensional.
	\end{enumerate}
\end{corollary}

\begin{proof}
(\ref{coro:requirements-for-inequational3}) follows from (\ref{coro:requirements-for-inequational1}).
\begin{enumerate}
\item We unfold the definition of inequational \lam-theory, and check the following properties.
\begin{itemize}
\item \emph{Preorder}. Reflexivity and transitivity are straightforward.
\item \emph{Compatibility}. It follows from the compatibility of $\leqchcol$. If $\tm \intleq \tmtwo$, then we need to prove that $\ctxp\tm \intleq \ctxp\tmtwo$, for any context $\ctx$. Let $\ctxtwo$ be a context. If $\monoToBlue{\ctxtwop{\ctxp\tm}} \bshcol{k}$ then $\monoToBlue{\ctxtwop{\ctxp\tmtwo}} \bshcol{k}$ by $\monoToBlue\tm \leqchcol \monoToBlue\tmtwo$.

\item \emph{Invariance}. Let $\tm \bconv \tmtwo$. Then clearly $\monoToBlue\tm \silconv \monoToBlue\tmtwo$. By \refthm{silconv-included-in-inter-ctx}, $\monoToBlue\tm \leqchcol \monoToBlue\tmtwo$. Then $\tm \intleq \tmtwo$.
\end{itemize}
Consistency of $\intleq$ is given by the fact that $\Id \not\intleq \Omega$, as it can be seen by considering the empty context. The failure of semi-extensionality is shown in \refex{eta}. 

\item The proof that $\intrleq$ is a compatible and consistent goes as for $\intleq$. For invariance, just note that $\eqb\,\subseteq\,\intleq\,\subseteq\,\intrleq$.\qedhere
\end{enumerate}
\end{proof}

\paragraph{Interaction Improvement and $\eta$} Note that the corollary does not say anything about $\intrleq$ and $\eta$. We conjecture that $\eta$-reduction is included in $\intrleq$, which would imply that the obvious inclusion $\intleq\,\subseteq\, \intrleq$ is strict, since it is not included in $\intleq$ (\refex{eta}) (beware that strictness of the inclusion does not follows from $\leqchcol\,\subsetneq\, \leqchcolr$ as strictness in that case relies on black \emph{and} white terms). For instance, $\comb{1} \toeta \Id$ and we expect that  $\comb{1} \intrleq \Id$ because, intuitively, head termination of $\ctxp\bOne$ does require at most one more interaction step than for $\ctxp\bId$, for all $\ctx\in\mathcal{C}_{\checkerssym}$. 

At present, however, it is only a conjecture. There is a real technical difficulty because the properties of $\eta$ that are usually used for rewriting-based proofs of similar facts fail for the checkers calculus, see \refex{eta}. The semantic tools that we shall develop in the next sections for studying $\intleq$ are not able to deal with $\eta$ either.

\paragraph{Adding and Removing Players} The next proposition formalizes the fact that the black-only and white-only fragments of the checkers calculus are embeddings of the ordinary $\l$-calculus, not only statically but also dynamically. The statements concern head steps, but generalize to arbitrary steps.
	\begin{toappendix}
	\begin{proposition}[Lifting Properties]
		\label{prop:embedding-of-head-reduction}
		Let $\colr\in\{\redclr,\blueclr\}$ and $t,u\in\Lambda$:
		\NoteProof{prop:embedding-of-head-reduction}
		\begin{enumerate}
			\item \emph{Head steps are turned into silent ones}: if $\tm\toh\tmtwo$ then $\monoToPlayer{\colr}{\tm}\tohnoint \monoToPlayer\colr\tmtwo$.
			
			\item \emph{Mono-player head-steps can be pulled back to the ordinary $\lambda$-calculus}: if $~\monoToPlayer{\colr}{\tm}\tohcol \tmthree$ then there exists a \lam-term $\tmtwo\in\Lambda$ such that $\tm\toh\tmtwo$ and $\tmthree=\monoToPlayer{\colr}{\tmtwo}$. 
			
			\item \emph{Head normal forms are preserved}: $\tm$ is a hnf if and only if $~\monoToPlayer{\colr}{\tm}$ is $\hchsym$-normal.
		\end{enumerate}
\end{proposition}
\end{toappendix}

Let us also formalize the fact that adding tags to application and abstractions does not change the possible reductions. 
Let $\tm\in\Lambda$, a tagging $\tm^T$ is a checkers term that preserves the syntax but tags its abstractions and applications with players. 
\begin{toappendix}
	\begin{proposition}\label{prop:ordinary-reductions-can-happen-in-checkers}
	
	Let $\tm\toh\tmtwo$ be a reduction in the ordinary $\l$-calculus. For any tagging $T$ of $\tm$, there exists a tagging $T'$ of $\tmtwo$ such that $\tm^T \tohcol \tmtwo^{T'}$ in the checkers calculus.\NoteProof{prop:ordinary-reductions-can-happen-in-checkers}
\end{proposition}
\end{toappendix}

It follows that sequences of head reductions in the ordinary $\l$-calculus are preserved in the checkers calculus, under some tagging determined by the first term in the reduction.

\paragraph{Hierarchy} Since $\leqchcol\,\subseteq\, \leqchcolr$, we have $\intleq\,\subseteq\, \intrleq$ (we showed that $\leqchcol\,\subsetneq\, \leqchcolr$, but not by using black lifted terms, so that does not give us $\intleq\,\subsetneq\, \intrleq$). The lifting properties above are used to prove the following expected lemma.
\begin{toappendix}
\begin{lemma}
$\intrleq\,\subseteq\, \obsle$.\label{l:improvements-are-ctx-equivalent}\NoteProof{l:improvements-are-ctx-equivalent}
\end{lemma}
\end{toappendix}
From \refex{eta}, it follows that $\Id\not\intrleq \comb{1}$, while $\Id \obsle \comb{1}$ holds. Thus, $\intrleq\,\subsetneq\, \obsle$. Summing up, we obtain the following hierarchy of preorders.
\begin{lemma}[Hierarchy]
\label{l:hierarchy-preorders}
$\intleq\,\subseteq\, \intrleq\,\subsetneq\, \obsle$. Similarly, $\inteq\, \subsetneq\, \obseq$.
\end{lemma}

%
%
%
%
%

\section{\bohm Trees}
Here starts the second part of the paper, where interaction equivalence $\inteq$ shall be characterized as the equational theory $\cB$ induced by the equality of \bohm trees. In this section, we recall \bohm trees and two notions of equality between them.

Barendregt proposed to represent the (possibly infinite) behavior of a \lam-term~$t$ as a (possibly infinite) tree, obtained by repeatedly slicing it with respect to head termination. We first present the idea informally, and then formally, following the similarity-based approach of \citet{lassen1999bisimulation}.

\begin{definition}[\citeauthor{Bare77} \citeyear{Bare77}]
The \emph{B\"{o}hm tree} $\BT{t}$ of a \lam-term $t$ is defined as follows: 
\begin{itemize}
\item
	If $t$ is head terminating then $t\to^*_\head\lambda x_{1}\ldots x_{n}.y\,t_{1}\cdots t_{k}$, for some $n,k\ge 0$, and we define:
	\begin{center}
\begin{tikzpicture}
\node (root) at (-70pt,8pt) {~};
\node (BT) at (0,0) {$\BT{t}\defeq$};
\node[right of = BT, node distance =2cm] (head) {$\lam x_{1} \ldots x_{n}.y\, $};
\node (BTN1) at ($(head)+(-.25,-.75)$) {$\BT{t_1}$};
\draw ($(head.south east)+(-0.5,0.1)$) -- ($(BTN1.north)-(0.,0)$);
\node (BTNk) at ($(head)+(2.,-.75)$) {$\BT{t_k}$};
\draw ($(head.south east)+(-0.1,0.1)$) -- ($(BTNk.north)-(0.,0)$);
\node at ($(head)+(.9,-.75)$) {$\cdots$};
\end{tikzpicture}
\end{center}
\item
	Otherwise, $t$ is head diverging and we define $\BT{t} \defeq \bot$.
\end{itemize}
B\"ohm trees are naturally ordered as follows: $\BT{t} \le_\bot \BT{u}$ whenever $\BT{u}$ is obtained from $\BT{t}$ by replacing some (possibly zero, or infinitely many) occurrences of $\bot$ by arbitrary trees.
\end{definition}
Intuitively, the constant $\bot$ represents the complete lack of information and, accordingly, the relation $\BT{t} \le_\bot \BT{u}$ captures the fact that the behavior of $u$ is more defined than that of $t$.

\begin{example}\label{ex:Bohmtrees}
Some examples of B\"ohm trees of notable \lam-terms:
\begin{center}
\begin{tikzpicture}
\node (root) at (0,0) {};
\node (BTP) at ($(root)+(10pt,-.8)$) {$\BT{\comb{1}} =\hspace{5pt} \lam xy.x$};
\node (M2) at ($(BTP)+(25pt,-19pt)$) {$y$};
\draw ($(BTP.south east)+(-11pt,2pt)$) -- ($(M2.north)+(0pt,-2pt)$);
\node (BTM) at (root) {$\BT{\Omega} = \bot$};
\node (BTP) at ($(root)+(2.75,0)$) {$\BT{\comb{P}z} =\hspace{5.5pt} \lam x_0.x_0$};
\node (M1) at ($(BTP)+(26.7pt,-18pt)$) {$\lam x_1.x_1$};
\node (M2) at ($(M1)+(4pt,-17pt)$) {$\lam x_2.x_2$};
\draw ($(BTP.south east)+(-8.75pt,2pt)$) -- ($(M1.north)+(4pt,-2pt)$);
\draw ($(M1.south)+(7pt,1pt)$) -- ($(M2.north)+(3pt,-3pt)$);
\draw[densely dotted] ($(M2.south)+(6pt,2pt)$) -- ($(M2.south)+(6pt,-4pt)$);
\node (BTY) at ($(root)+(5.75,0)$) {$\BT{\comb{Y}} =\hspace{5.5pt} \lam f.f$};
\node (M1) at ($(BTY)+(25.5pt,-18pt)$) {$f$};
\node (M2) at ($(M1)+(0pt,-17pt)$) {$f$};
\draw ($(BTY.south east)+(-8.75pt,2pt)$) -- ($(M1.north)+(0,-1pt)$);
\draw ($(M1.south)+(-0pt,2pt)$) -- ($(M2.north)+(0,-1pt)$);
\draw[densely dotted] ($(M2.south)+(-0pt,2pt)$) -- ($(M2.south)+(0,-4pt)$);
\node (BTY3) at ($(root)+(7.75,0)$) {$\ge_\bot\ \lam f.f$};
\node (M1) at ($(BTY3)+(13pt,-18pt)$) {$f$};
\node (M2) at ($(M1)+(0pt,-17pt)$) {$\bot$};
\draw ($(BTY3.south east)+(-8.75pt,2pt)$) -- ($(M1.north)+(0,-1pt)$);
\draw ($(M1.south)+(-0pt,2pt)$) -- ($(M2.north)+(0,-1pt)$);
\node (BTY4) at ($(root)+(9.5,0)$) {$\ge_\bot\ \lam f.f$};
\node (M1) at ($(BTY4)+(13.2pt,-18pt)$) {$\bot$};
\draw ($(BTY4.south east)+(-8.75pt,2pt)$) -- ($(M1.north)+(0,-1pt)$);
\node at ($(BTY4)+(1.5,0)$) {$\ge_\bot\ \bot$};
\end{tikzpicture}
\end{center}
where $\comb{P} = \comb{Y}(\lam yzx.x(yz))$ satisfies $\comb{P}z =_\beta \lam x.x(\comb{P}z)$.
Note that $z$ is never erased by $\comb{P}$, rather ``pushed into infinity'' in the sense that $\comb{P}z\to^*_\beta t$ entails $z\in\FV{t}$, but $z$ does not occur in $\BT{\comb{P}z}$.
\end{example}
It is well-known that B\"ohm trees are invariant under $\beta$-conversion~\cite[Ch.~10]{Barendregt84}. Informally, the \bohm preorder $\btpre$ on terms is defined by pulling back the preorder on the associated trees, that is, $\tm \btpre \tmtwo$ if $\BT{t} \le_\bot \BT{u}$. Formally, we rather define the \bohm preorder as a notion of similarity, following \citet{lassen1999bisimulation}.

\begin{definition}[\bohm preorder, formally]\label{def:bohm-preorder}The \bohm preorder $\btpre$, also known as \emph{head (normal form) similarity}, is the largest relation $\tm \btpre \tmtwo$ closed under the following clauses:
	\begin{itemize}
	\item[(bot)] $\tm\bshdiv$  \ie ~ $\tm$ has no head normal form.
	\item[($\relation{H}$)] $t\bshs \lam x_1\dots x_n.y\,t_1\cdots t_k$ and $u\bshs \lam x_1\dots x_n.y\,u_1\cdots u_k$ with $(\tm_i\btpre u_i)_{i\leq k}$.
	\end{itemize}
\end{definition}

\begin{theorem}[{\cite[Cor.~14.3.20(iii)]{Barendregt84}}]
The \bohm preorder $\btpre$ is an inequational $\l$-theory. Moreover, $\btpre$ is neither extensional nor semi-extensional.
 \end{theorem}
The main result of the paper is that $\btpre$ coincides with the interaction preorder $\intleq$. For the direction $\intleq\,\subseteq\ \btpre$, we shall partly rely on a similar important result in the theory of the $\l$-calculus, namely Hyland's semi-separation theorem, which concerns an extensional variant of $\btpre$ below.

\paragraph{Extensional \bohm Preorder} 
We now formally introduce an extensional version of $\btleq$ capturing the preorder induced by Scott's model $\mathcal{D}_\infty$. 

\begin{definition}[Extensional \bohm preorder]\label{def:HNFBISIM}
	\label{def:hnf-nfs}The \emph{extensional \bohm preorder} $\etabtle$ is the largest relation $\tm \etabtle \tmtwo$ closed under the following clauses:
	\begin{itemize}
	\item[(bot)] $\tm\bshdiv$  \ie ~ $\tm$ has no head normal form.
	\item[($\relation{H}\eta$)] $t\bshs h$,  $u\bshs h' $, and there exist $n,k\ge 0$ such that $h =_\eta \lam x_1\dots x_n.y\,t_1\cdots t_k$ and $h' =_\eta  \lam x_1\dots x_n.y\,u_1\cdots u_k$ with $(\tm_i\etabtle u_i)_{i\leq k}$.
	\end{itemize}
\end{definition}

\begin{example} Recall that $\comb{I}$, $\comb{1}$, $\comb{K}$, and $\comb{F}$ have been defined in Notation~\ref{not:lam-terms}.
\bsub
\item We start with some negative examples: $\comb{K}\not\etabtle\comb{F}$ and $\comb{F}\not\etabtle\comb{K}$ (their head variables differ).
\item Both $\comb{I}\etabtle\comb{1}$ and $\comb{1}\etabtle\comb{I}$ hold. This entails the extensionality of $~\etabtle$, hence of $~\etabteq$.
\item $\lam xy.x\Omega\etabtle \comb{I}$, holds since $\lam xy.x\Omega\btpre\comb{1}\etabtle \comb{I}$ and $\btpre\,\subseteq\ \etabtle$. Conclude by transitivity.
\item To understand the role of the infinitary $\eta$-expansion, consider the \lam-term $\comb{J} = \comb{Y}(\lam jx.x(jx))$. 
It is easy to check that $\comb{J}\etabtle\comb{I}$ by looking at the in-line depiction of its Böhm tree:
\begin{center}$
	\BT{\comb{J}} = \lam xy_0.x(\lam y_1.y_0(\lam y_2.y_1(\lam y_3.y_2(\cdots))))
$\end{center}
\esub
\end{example}

The relations $\etabtle$ and $\etabteq$ have characterizations as contextual preorders/equivalences.
\begin{theorem}[\citeauthor{Hyland75} \citeyear{Hyland75}/\citeauthor{Wadsworth76} \citeyear{Wadsworth76}]\label{thm:HW76} 
For all \lam-terms $t,u$, we have $t\obsle u$ if and only if $t\etabtle u$.
Therefore $t\obseq u$ if and only if $t\etabteq u$.
\end{theorem}

\section{Completeness, or Separating \bohm Different Terms}\label{sec:bohmout}
In this section, we prove that the interaction preorder $\intleq$ is included in the \bohm preorder $\btleq$.

\paragraph{Proof Technique} The standard way of proving that a contextual preorder $\sqsubseteq$ is included in a tree similarity $\precsim$ is to proceed by proving the contrapositive: one supposes $\BT{t}\not\precsim\BT{u}$, and constructs a context $C$ that:
\begin{enumerate}
\item \emph{Extraction of the difference}: brings up this difference from possibly deep down the tree structure of $\BT{t}$ and $\BT{u}$, that is, $\ctx$ is such that $\ctxp\tm$ and $\ctxp\tmtwo$ head reduce to two terms $\tm'$ and $\tmtwo'$ for which the difference $\BT{t'}\not\precsim\BT{u'}$ is on the root node, and
\item \emph{Root separation}: exploits the root difference to make $\tm'$ head terminating and $\tmtwo'$ head divergent, thus obtaining that $\tm\not\sqsubseteq\tmtwo$.
\end{enumerate} 

In the literature, this extraction process is known as \emph{B\"ohm out} technique \cite{Boehm68}; see also \cite{Boehm68,Dezani98,BOUDOL199683,IntrigilaMP19,BarendregtM22}. It is a concise and yet sophisticated technique. The culprit is that the extracting context needs to first \emph{reorganize} the applicative structure of $\tm$ and $\tmtwo$ (via \emph{tupler combinators}, see below), to then apply the suitable selectors for extracting the discriminating sub-terms. 

What raises difficulties is when one has two different terms, say, $\tm \defeq \var\tm_1\tm_2$ and $\tmtwo \defeq \var\tm_1\tm_3$ with the difference deep down the structure of $\tm_2$ and $\tm_3$. Intuitively, to 'extract' $\tm_2$ and $\tm_3$ one would simply substitute for $\var$ a term that selects the second argument, namely $\lam\vartwo\varthree.\varthree$. Unfortunately, this approach is too simple to work, because then extracting the difference deep down $\tm_2$ and $\tm_3$ might require to substitute a \emph{different} selecting term (say, of the first argument) for another occurrence of $\var$ in $\tm_2$ and $\tm_3$, but clearly all the occurrences of $\var$ must receive the \emph{same} selecting term. An example of how the \bohm out technique solves this \emph{colliding selectors issue} is discussed below.

\paragraph{Definitions for the \bohm Out Technique} On closed \lam-terms, the technique amounts to applying the \emph{tupler} $\Tupler{n}$ and the \emph{$i$-th selector $\Proj{n}{i}$} defined as follows:
\begin{center}$
\begin{array}{rrllllllll}
\textsc{$n$-tuples} & \Tuple{t_1,\dots,t_n} & \defeq & \lam x.xt_1\cdots t_n, & \mbox{ with $x$ fresh};
\\
\textsc{Tuplers}& \Tupler n &\defeq & \lam x_1\ldots x_n.\Tuple{x_1,\dots,x_n};
\\
\textsc{Selectors} & \Proj{n}{i} & \defeq & \lam x_1\ldots x_n.x_i,&\textrm{with } 1\le i \le n.
	\end{array}$
\end{center}
So, the tupler $\Tupler n$ takes $n$ arguments $\tm_1,\ldots,\tm_n$ and returns the tuple $\Tuple{\tm_1,\dots,\tm_n}$, while the selector $\Proj{n}{i}$ takes $n$ arguments $\tm_1,\dots,\tm_n$ and returns the $i$-th argument $\tm_i$.
Note that $\Proj{1}1 = \comb{I}$. Then, $\Tupler nt_1\cdots t_nu \toh^* ut_1\cdots t_n$ and $\Proj{n}{i}t_1\cdots t_n\toh^* t_i$, whence we have the following combined \emph{extraction property}:
\begin{equation}
\begin{array}{ccc}
\Tupler nt_1\cdots t_n\Proj{n}{i} &\toh^*& t_i. 
\end{array}
\label{eq:extraction-property}
\end{equation}

In the following, we shall need to locate nodes/hnfs occurring at a certain path in a B\"ohm tree. 
\begin{definition}
\hfill
\begin{itemize}
\item \emph{Path}: a path is a (possibly empty) finite list of natural numbers $\alpha=\Tuple{a_1,\ldots,a_n}$, where $a_i\geq 1$, for each $1\leq i\leq n$. 
\item \emph{Concatenation}: given $i\in\nat$ and a path $\alpha$ as above, their concatenation is $i\cdot\alpha \defeq \Tuple{i,a_1,\ldots,a_n}$.
\item \emph{Node occurring at a path}: let $\tm$ be a $\lambda$-term such that $t\to^*_\head \lam \vec x.yt_1\cdots t_k$. If $\alpha = \Tuple{}$ then $t\restr_{\alpha} = \lam \vec x.yt_1\cdots t_k$, if $\alpha = i\cdot\alpha'$ with $i \leq k$ then $t_{\alpha} = (t_i)_{\alpha'}$.
	Otherwise, $t\restr_{\alpha}$ is undefined.
	\end{itemize}
\end{definition}

\begin{example}We show how the B\"ohm out technique is able to construct a context separating the \lam-term $\tm\defeq \lam xy.x(x\Omega y)\Omega$ from the \lam-term $\tmtwo\defeq \lam xy.x(xy\Omega)\Omega$, which provide an example of the colliding selectors issue. We represent below their B\"ohm trees:
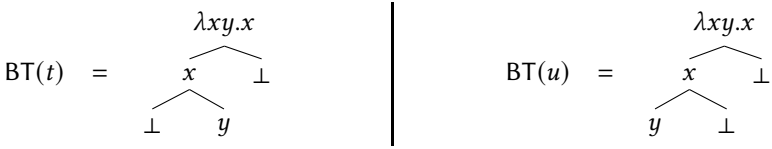
\begin{center}
$\begin{array}{ccc\colspace\colspace\colspace|\colspace\colspace\colspace ccc}
	\BT\tm
	& =&
	\begin{tikzpicture}[ocenter]
	\node (t) {$\lam xy.x$};
	
	\node at (t.center)[below left=13pt and 7pt] (t1) {$x$};
	\node at (t.center)[below right=13pt and 7pt] (t2) {$\bot$};
	\draw (t.south) -- (t1.north);
	\draw (t.south) -- (t2.north);
	
	\node at (t1.center)[below left=13pt and 7pt] (t11) {$\bot$};
	\node at (t1.center)[below right=13pt and 7pt] (t12) {$y$};
	\draw (t1.south) -- (t11.north);
	\draw (t1.south) -- (t12.north);		
	\end{tikzpicture}
	
	&
	
	\BT\tmtwo
	& =&
	\begin{tikzpicture}[ocenter]
	\node (t) {$\lam xy.x$};
	
	\node at (t.center)[below left=13pt and 7pt] (t1) {$x$};
	\node at (t.center)[below right=13pt and 7pt] (t2) {$\bot$};
	\draw (t.south) -- (t1.north);
	\draw (t.south) -- (t2.north);
	
	\node at (t1.center)[below left=13pt and 7pt] (t11) {$y$};
	\node at (t1.center)[below right=13pt and 7pt] (t12) {$\bot$};
	\draw (t1.south) -- (t11.north);
	\draw (t1.south) -- (t12.north);		
	\end{tikzpicture}
\end{array}$
\end{center}

	
	
	Showing that $\tm \not\obsle \tmtwo$ requires a context $\ctx$ making $\tm$ converge and $\tmtwo$ diverge. The path to extract along is $\alpha'=\Tuple{1,2}$, which showcases the colliding selectors issue: the first occurrence of $\var$ needs to select the first argument, while the second occurrence needs the second argument.  \bohm's trick consists in using tuplers, as we now attempt to explain. The \bohm out context is $C \defeq \ctxhole\Tupler{2}\comb{I}\Proj{2}{1}\Proj{2}{2}$. The idea is that, by substituting a tupler $\Tupler 2$ on $\var$ and having as further arguments the right selectors for each occurrence, one exploits the extraction property \refeq{extraction-property} above to select the right sub-term in each case. The identity $\comb{I}$ is added like a padding when more arguments are needed. Concretely, we have:
	 \[
	 \begin{array}{c\colspace | \colspace c}
	 \begin{array}{rllllllll}
	 C\ctxholep{\tm}&=&(\lam xy.x(x\Omega y)\Omega)\Tupler{2}\comb{I}\Proj{2}{1}\Proj{2}{2}
\\& \to_\head^2 &
	 \Tupler{2}(\Tupler{2}\Omega\comb{I})\Omega\Proj{2}{1}\Proj{2}{2}
	 
	 \\
	\mbox{by } \refeq{extraction-property}  & \to_\head^* &\Tupler{2}\Omega\comb{I} \Proj{2}{2} &
	 \\
\mbox{by } \refeq{extraction-property} & \to_\head^* &\comb{I} &
\end{array}
&
\begin{array}{rllllllll}
	 C\ctxholep{\tmtwo}&=&(\lam xy.x(xy\Omega)\Omega)\Tupler{2}\comb{I}\Proj{2}{1}\Proj{2}{2}
\\& \to_\head^2 &
	 \Tupler{2}(\Tupler{2}\comb{I}\Omega)\Omega\Proj{2}{1}\Proj{2}{2}
	 
	 \\
	\mbox{by } \refeq{extraction-property}  & \to_\head^* &\Tupler{2}\comb{I}\Omega \Proj{2}{2} &
	 \\
\mbox{by } \refeq{extraction-property} & \to_\head^* &\Omega &
\end{array}
\end{array}
\]
\end{example}

There are more technicalities of the \bohm out technique, unfortunately. Firstly, since the extraction process works through substitutions of tuplers and selectors, one in general extracts a \emph{substitution instance of a sub-term}, and not the sub-term itself. Secondly, in our example both occurrences of $\var$ have two arguments, but in general different occurrences might have different numbers of arguments. Then, one needs to apply a tupler $\Tupler n$ with $n$ ``large enough'', and this over-approximation may destroy some $\eta$-differences between the two trees. If one observes only termination, then $\eta$-equivalent terms are \emph{not} separable. 

In our setting, we are able to discriminate $\eta$-convertible \lam-terms because we observe termination \emph{and} count the number of interaction steps, which are changed by $\eta$, as showed by \refex{eta}. The idea is to black dress the differing terms $\tm$ and $\tmtwo$ and to white dress the separating context $\ctx$, so that the black $\eta$-differences in $\monoToBlue\tm$ and $\monoToBlue\tmtwo$ that are erased by $\monoToRed\ctx$ are turned into interaction steps.

\paragraph{Interaction \bohm Out} The following lemma shows how to separate those \lam-terms which have B\"ohm trees differing \emph{only} by some (possibly infinitary) $\eta$-expansions, that is, the case when $t\etabtle u$ and $t \not\btpre u$, because the case $t\not\etabtle u$ is handled by \refthm{HW76}, that is, via standard \bohm out. Because of its technical nature, it is labeled as a \emph{lemma} and yet it is one of the main technical contributions of the paper. The inclusion $\intleq\,\subseteq\,\btleq$ then follows easily.

\begin{lemma}[Interaction \bohm-out]
Let $t,u\in\Lambda$ such that $t\etabtle u$ and $t \not\btpre u$.
Then, there exists a context $\ctx\in\Cctx$ such that $
\monoToRed{\ctx}\ctxholep{\monoToBlue{t}}\bshcol{i}$ and
$\monoToRed{\ctx}\ctxholep{\monoToBlue{u}}\bshcol{i'}$ with $i'\neq i$.
\label{lem:separatingeta} 
\end{lemma}

	\paragraph{Terminology and Notations for the Proof}
		As customary in mathematical analysis, we say that a relation $\relation{P}(-)$ holds for all $K\in\nat$ \emph{large enough} whenever there exists a $K'\in\nat$ such that $\relation{P}(K)$ holds for all $K\ge K'$. We also use the notation $tu^{\sim n}$ for $(\cdots((tu)u)\cdots)u$ ($n$ times). Finally, two head normal forms $h,h'$ are \emph{spine equivalent}, written $h \speq h'$, if there are $n,k\ge 0$ such that:
\begin{equation}\label{eq:t-le-u}
	h= \las{\var}{n} \vartwo \,\apps{\tm}{k} \quad\textrm{ and }\quad
	h'= \lambda \var_1\ldots \var_{n}. \vartwo \,\apps{\tmtwo}{k}.
\end{equation}

\begin{proof} 
\applabel{lem:separatingeta} 
We prove a stronger statement, \ie that there exist closed terms $\vec s\in\Lambda$ such that, for all $\vec y$ containing $\FV{t}\cup\FV{u}$  and for all $K\in\nat$ large enough, the following holds:
\begin{center}$
\rapp{
	\monoToBlue{t}\isub{\vec y}{\monoToRed{\Tupler{K}}}
}{\monoToRed{\vec{s}}}\bshcol{i}
\quad\text{ and }\quad
\rapp{
	\monoToBlue{u}\isub{\vec y}{\monoToRed{\Tupler{K}}}
}{\monoToRed{\vec{s}}}\bshcol{i'}\textrm{ with }i'\neq i.
$\end{center}
Given variables $\vec x$ and a \lam-term $t$ we write $\sigma_{\vec x}$ for $\isub{\vec x}{\monoToRed{\Tupler{K}}}$, and $t^{\sigma_{\vec x}}$ for $ t\isub{\vec x}{\monoToRed{\Tupler{K}}}$.

Note that $t \not\btpre u$ is only possible if $t\bsh{}$. Moreover, $t\bsh{}\!\!h$ and $t\etabtle u$ entail $u\bsh{}\!\!h'$, for some $h'$.
We proceed by induction on the length of a minimal path $\delta\in\nat^*$ such that $t\restr_\delta\ \not\speq u\restr_\delta$.

\underline{Base case} $\delta = \emptyseq$, \ie $h\ \not\speq h'$. Then $t\etabtle u$ is only possible if the amount of spine abstractions and applications in $h,h'$ can be matched via $\eta$-expansion, say:
\begin{center}$
	t \to^*_\head h = \lam x_1\dots x_n.y\,t_1\cdots t_k
	\qquad
	\textrm{ and }
	\qquad	
	u \to^*_\head h'= \lam x_1\dots x_nz_1\dots z_m.y\,u_1\cdots u_{k+m}
$\end{center}
for $n,k\ge0$ and $m>0$. (The symmetrical case where $h$ has more abstractions/applications than $h'$ is omitted because analogous.)
There are two subcases to consider, depending on whether $y$ is free.
\begin{enumerate}
\item \emph{$y$ is free}, \ie $y\in\vec y$. Take any $K \ge k+m$, and  empty $\vec s$. For $t$, we have:
{ \begin{center}$
	\begin{array}{c|ll}
	\intsym\textsc{-Steps}&\textsc{Terms and $\nointsym$-steps}
	\\[2pt]
	\hline
	&
	\monoToBlue{t}\isub{\vec y}{\monoToRed{\Tupler{K}}}
	\ \ \tohnoint^* \ \ 
	\monoToBlue{h}\isub{\vec y}{\monoToRed{\Tupler{K}}},\phantom{X^{X^{X^{}}}}\hfill\textrm{ by Pr.~\ref{prop:embedding-of-head-reduction}(1) \&  \reflemmaeq{color-substitutivity}}
		\\[2pt]
	=	
	&
	\bapp{\bapp{\bapp{
		\manyblam{n}{x}\monoToRed{\Tupler{K}}
	}{{\monoToBlue{t_1}}^{\sigma_{\vec y}}}}{\cdots}}{{\monoToBlue{t_k}}^{\sigma_{\vec y}}}
	\\[2pt]	
	\tohint^k	
	&
	\manyblam{n}{x}\manyrlam[k+1]{K}{w}\Tuple{{\monoToBlue{t_1}}^{\sigma_{\vec y}},\dots,{\monoToBlue{t_k}}^{\sigma_{\vec y}},w_{k+1},\dots,w_K}_{\redclr}\\[2pt]
	\end{array}
$\end{center}}
where $\Tuple{-,\dots,-}_{\redclr} $ denotes the tuple with white applications $ \lam z.z-\circ\cdots\circ -$. For $u$, we have:
{\begin{center}$
	\begin{array}{c|ll}
	\intsym\textsc{-Steps}&\textsc{Terms and $\nointsym$-steps}
	\\[2pt]
	\hline
	&

		\monoToBlue{u}\isub{\vec y}{\monoToRed{\Tupler{K}}}
	\ \ \tohnoint^* \ \ 
		\monoToBlue{h'}\isub{\vec y}{\monoToRed{\Tupler{K}}},\phantom{X^{X^{X^{}}}}\hfill\textrm{ by Pr.~\ref{prop:embedding-of-head-reduction}(1) \& \reflemmaeq{color-substitutivity}},
		\\[2pt]
=	&
	\bapp{\bapp{\bapp{
\monoToRed{\Tupler{K}}
	}{{\monoToBlue{u_1}}^{\sigma_{\vec y}}}}{\cdots}}{{\monoToBlue{u_{k+m}}}^{\sigma_{\vec y}}}
	\\[2pt]
\tohint^{k+m}	&
		\manyblam{n}{x}\manyrlam[k+1]{K}{w}\Tuple{{\monoToBlue{u_1}}^{\sigma_{\vec y}},\dots,{\monoToBlue{u_{k+m}}}^{\sigma_{\vec y}},w_{k+m+1},\dots,w_K}_{\redclr}
	\end{array}
$\end{center}}
Summing up, ${\monoToBlue{t}}^{\sigma_{\vec y}}\bshcol{k}$ and  ${\monoToBlue{u}}^{\sigma_{\vec y}}\bshcol{k+m}$. The statement holds because $m>0$.

\item \emph{$y$ is bound}, \ie $y = x_j\in\vec x$. Take any $K\ge k+m$, and let the arguments $\vec s$ be $n$ copies of $\Tupler{K}$ (noted $\Tupler{K}^{\sim n}$ for short). On the one hand:
{\begin{center}$
	\begin{array}{c|lll}
	\intsym\textsc{-Steps}&\textsc{Terms and $\nointsym$-steps}
	\\[2pt]
	\hline
	&\rapp{\monoToBlue{t}\isub{\vec y}{\monoToRed{\Tupler{K}}}}{{\monoToRed{\Tupler{K}}}^{\sim n}}
	\ \ \tohnoint^* \ \ 
\rapp{
		\monoToBlue{h}\isub{\vec y}{\monoToRed{\Tupler{K}}}
	}{{\monoToRed{\Tupler{K}}}^{\sim n}},
	&\textrm{by Pr.~\ref{prop:embedding-of-head-reduction}(1) \& \reflemmaeq{color-substitutivity}},
	\\[2pt]
=	&
	\rapp{
		\big(\manyblam{n}{x}\bapp{\bapp{\bapp{x_j}{{\monoToBlue{t_1}}^{\sigma_{\vec y}}}}{\cdots}}{{\monoToBlue{t_k}}^{\sigma_{\vec y}}}\big)
	}{{\monoToRed{\Tupler{K}}}^{\sim n}}\\[2pt]
\tohint^n	&
	\bapp{\bapp{\bapp{
		\monoToRed{\Tupler{K}}
	}{{\monoToBlue{t_1}}^{\sigma_{\vec x\vec y}}}}{\cdots}}{{\monoToBlue{t_k}}^{\sigma_{\vec x\vec y}}}\\[2pt]	
\tohint^k	&
		\manyrlam[k+1]{K}{w}\Tuple{{\monoToBlue{t_1}}^{\sigma_{\vec x\vec y}},\dots,{\monoToBlue{t_k}}^{\sigma_{\vec x\vec y}},w_{k+1},\dots,w_K}_{\redclr}
	\end{array}
$\end{center}}
On the other hand:
{\begin{center}$
	\begin{array}{c|lll}
	\intsym\textsc{-Steps}&\textsc{Terms and $\nointsym$-steps}
	\\[2pt]
	\hline
	&\rapp{\monoToBlue{u}\isub{\vec y}{\monoToRed{\Tupler{K}}}}{{\monoToRed{\Tupler{K}}}^{\sim n}}
	\ \ \tohnoint^* \ \ 
	\rapp{
		\monoToBlue{h'}\isub{\vec y}{\monoToRed{\Tupler{K}}}
	}{{\monoToRed{\Tupler{K}}}^{\sim n}},
		\ \ \ \ \textrm{by Pr.~\ref{prop:embedding-of-head-reduction}(1) \& \reflemmaeq{color-substitutivity}},
	\\[2pt]
=	&
	\rapp{
		\big({\lambda_{\blueclr\cdots \blueclr}}x_1\ldots x_{n}\vec z.\,\bapp{\bapp{\bapp{x_j}{{\monoToBlue{u_1}}^{\sigma_{\vec y}}}}{\cdots}}{{\monoToBlue{u_{k+m}}}^{\sigma_{\vec y}}}\big)
	}{{\monoToRed{\Tupler{K}}}^{\sim n}}\\[2pt]
\tohch^n	&
	\bapp{\bapp{\bapp{
				\monoToRed{\Tupler{K}}
	}{{\monoToBlue{u_1}}^{\sigma_{\vec x\vec y}}}}{\cdots}}{{\monoToBlue{u_{k+m}}}^{\sigma_{\vec x\vec y}}}\\[2pt]	
\tohch^{k+m}&
		\manyrlam[k+m+1]{K}{w}\Tuple{{\monoToBlue{u_1}}^{\sigma_{\vec x\vec y}},\dots,{\monoToBlue{u_{k+m}}}^{\sigma_{\vec x\vec y}},w_{k+m+1},\dots,w_K}_{\redclr}.
	\end{array}
$\end{center}}
Summing up, ${\monoToBlue{t}}^{\sigma_{\vec y}}\bshcol{n+k}$ and  ${\monoToBlue{u}}^{\sigma_{\vec y}}\bshcol{n+k+m}$. The statement holds because $m>0$.
\end{enumerate}

 \underline{Inductive case} $\delta = j\cdot\gamma$. In this case, we must have:
\begin{center}$
	t \to^*_\head h = \lam x_1\dots x_n.y\,t_1\cdots t_k
	\qquad
	\textrm{ and }
	\qquad	
	u \to^*_\head h'= \lam x_1\dots x_n.y\,u_1\cdots u_k
$\end{center}
with $t_j \not\btpre u_j$ and $(t_l \etabtle u_l)_{l \le k}$. By \ih, there exists $K'$ and $\vec{s'}$ such that for all $K\ge K'$:
\begin{center}$
	\rapp{{\monoToBlue{t_j}}^{\sigma_{\vec x\vec y}}}{\monoToRed{\,\vec{s'}}}
	\bshcol{i}
	\quad\mbox{ and }\quad
	\rapp{{\monoToBlue{u_j}}^{\sigma_{\vec x\vec y}}}{\monoToRed{\,\vec{s'}}}\bshcol{i'}\textrm{ with }i'\neq i.
$\end{center}
We consider any $K\ge \max\{K',k\}$.
We assume wlog.\ that $y$ is free, the other case being analogous.
{ \begin{center}$
	\begin{array}{c|lll}
	\textsc{Steps}&\textsc{Terms}
	\\[2pt]
	\hline
&
\rapp{
		\rapp{
			\rapp{
				\monoToBlue{t}\isub{\vec y}{\monoToRed{\Tupler{K}}}
			}{{\monoToRed{\Tupler{K}}}^{\sim n+K-k}}
		}{\monoToRed{\Proj{K}{j}}}
	}{\monoToRed{\,\vec{s'}}}, 
	\hfill \textrm{by Pr.~\ref{prop:embedding-of-head-reduction}(1) \& \reflemmaeq{color-substitutivity}},
	\\[2pt]
	\tohnoint^* &
	\rapp{
		\rapp{
			\rapp{
				\monoToBlue{h}\isub{\vec y}{\monoToRed{\Tupler{K}}}
			}{{\monoToRed{\Tupler{K}}}^{\sim n+K-k}}
		}{\monoToRed{\Proj{K}{j}}}
	}{\monoToRed{\,\vec{s'}}}
	&
	\\[2pt]
	=
	&
	\rapp{
		\rapp{
			\rapp{
				\big(\manyblam{n}{x}\bapp{\bapp{\bapp{\monoToRed{\Tupler{K}}}{{\monoToBlue{t_1}}^{\sigma_{\vec y}}}}{\cdots}}{{\monoToBlue{t_k}}^{\sigma_{\vec y}}}\big)
			}{{\monoToRed{\Tupler{K}}}^{\sim n+K-k}}
		}{\monoToRed{\Proj{K}{j}}}
	}{\monoToRed{\,\vec{s'}}}\\[2pt]
	\tohint^n&
	\rapp{
		\rapp{
			\rapp{
				\bapp{\bapp{\bapp{\monoToRed{\Tupler{K}}}{{\monoToBlue{t_1}}^{\sigma_{\vec x\vec y}}}}{\cdots}}{{\monoToBlue{t_k}}^{\sigma_{\vec x\vec y}}}
			}{{\monoToRed{\Tupler{K}}}^{\sim K-k}}
		}{\monoToRed{\Proj{K}{j}}}
	}{\monoToRed{\,\vec{s'}}}\\[2pt]				
	\tohint^k\tohnoint^*&
	\rapp{{\monoToBlue{t_j}}^{\sigma_{\vec x\vec y}}}{\monoToRed{\,\vec{s'}}} 
	\hfill \textrm{by \refeq{extraction-property}}.
	\end{array}
$\end{center}}
An identical sequence of steps extracts $\monoToBlue{u_j}$ from the other term, that is, we have:
\begin{center}
$\begin{array}{ccc}
\rapp{\rapp{\rapp{
				\monoToBlue{u}\isub{\vec y}{\monoToRed{\Tupler{K}}}
			}{{\monoToRed{\Tupler{K}}}^{\sim n+K-k}}
		}{\monoToRed{\Proj{K}{j}}}
	}{\monoToRed{\,\vec{s'}}}

&\tohnoint^*\tohint^{n+k}\tohnoint^*&
	\rapp{{\monoToBlue{u_j}}^{\sigma_{\vec x\vec y}}}{\monoToRed{\,\vec{s'}}}
	\end{array}$
	\end{center}
Note the same number of $\intsym$-steps. By defining $\monoToRed{\vec s}$ as the arguments ${\monoToRed{\Tupler{K}}}^{\sim n+K-k},\monoToRed{\Proj{K}{j}}, \monoToRed{\,\vec{s'}}$, and by composing with what is obtained by the \ih, we obtain:
\begin{center}$
\rapp{
	\monoToBlue{t}\isub{\vec y}{\monoToRed{\Tupler{K}}}
}{\monoToRed{\vec{s}}}\bshcol{n+k+i}
\quad\text{and}\quad
\rapp{
	\monoToBlue{u}\isub{\vec y}{\monoToRed{\Tupler{K}}}
}{\monoToRed{\vec{s}}}\bshcol{n+k+i'},
$\end{center}
which is an instance of the statement because $i\neq i'$ by \ih\qedhere
\end{proof}


\begin{theorem}[Completeness] \label{th:intleq-included-in-bohm}
	Let $t,u\in\Lambda$. If $t\intleq u$ then $t\btpre u$.
\end{theorem}
\begin{proof} 
Assume $t\not\btpre u$, towards a contradiction. There are two cases:
	\bsub
	\item If $t\not\etabtle u$, then by Theorem~\ref{thm:HW76} there exists a context $C$ such that $C\ctxholep{t}\!\Downarrow_\head$, while $C\ctxholep{u}\!\not\Downarrow_\head$.
	By~\refprop{ordinary-reductions-can-happen-in-checkers}, it follows that $\monoToRed{\ctx}\ctxholep{\monoToBlue{t}}\bshcols$, while $\monoToRed{\ctx}\ctxholep{\monoToBlue{u}}\bshcoldiv$. This shows $t\not\intleq u$.
	\item If $t\etabtle u$ then $t\not\intleq u$ follows directly from interaction \bohm out (Lemma~\ref{lem:separatingeta}).\qedhere
	\esub
\end{proof}

\section{Multi Types and Relational Semantics}\label{sec:relsem}
We now start preparing the ground for the proof of the inclusion $\btleq\, \subseteq\, \intleq$. The main tool shall be a system of checkers multi types. The needed background on multi types, a.k.a. \emph{non-idempotent intersection types}, is recalled here, the checkers variant shall be introduced in the next section.

Here, we present Engeler's relational model \cite{HylandNPR06} in terms of \citeauthor{Carvalho07}'s system of multi types \citeyearpar{Carvalho07,DBLP:journals/mscs/Carvalho18}, together with some classic results. 
In particular, we recall the multi type characterization of head normalizability, and the bounds of the length of head evaluations\ignore{, and of the head size of hnfs,} that can be extracted from the type derivations.

\begin{definition} Types and typing rules of the multi types system are given in \reffig{plain-multi-type-system}.
\end{definition}

\paragraph{Multi Types.} There are two categories of types: \emph{linear types} $\ltype$, which include a single\footnote{One may ask for more atomic types, but this choice does not really affect the results presented in the paper.} atomic type $\tvar$ and arrow types $\mtype \typearrow \ltype$; \emph{multi types} $\mtype$, which are possibly empty multisets of linear types. 
Multi types are generally represented as unordered lists $\mset{\ltype_1, \dots, \ltype_n}$ of linear types $\ltype_1, \dots, \ltype_n$, possibly with repetitions.
The \emph{empty} multi type $\mset{\,}$, obtained by taking $n = 0$, is also denoted by $\emptymset$.

A multi type $\mset{\ltype_1, \dots, \ltype_n}$ should be intended as a conjunction $\ltype_1 \land \dots \land 
\ltype_n$, for a commutative, associative, non-idempotent conjunction 
$\land$ (morally a tensor $\otimes$), having $\emptymset$ as a neutral element.
The intuition is that a linear type corresponds to a single use of a term $\tm$, which is typed with a 
multiset $\mtype$ of cardinality $n$ if it is going to be used $n$ times. 
In particular, if $n>0$ and $\tm$ is part of a larger term $\tmtwo$, then a copy of $\tm$ shall end up in evaluation (\ie head) position during the evaluation of $\tmtwo$.

\begin{figure}[t!]	
\centering
\arraycolsep=2pt
\begin{tabular}{c|c}
\textsc{Types} & \textsc{Typing rules}
\\\hline
\\[-5pt]
	$\begin{array}{rrll}
	\textsc{Linear} & \ltype, \ltypetwo &\grameq& \vartype \mid \mtype \typearrow \ltype
	\\
	\textsc{Multi} & \mtype, \mtypetwo &\grameq& \multitype{n}{\ltype} \ \ n\geq 0
		\\[4pt]
	\textsc{Generic} & \gtype, \gtypetwo &\grameq& \ltype \mid \mtype
	\\
		\textsc{Zero} & \zero &\grameq& \mset{\ }
	\end{array}$
	&
	\begin{tabular}{c\colspace c}
	\infer[\typingruleAx]{\var \hastype [\ltype] \vdash \var \hastype \ltype}{}
	&	
	\infer[\typingruleMany]{\uplus_{i\in I}\typectx_i \vdash  \tm \hastype [\ltype_i]_{i\in I}}{(\typectx_i \vdash \tm \hastype \ltype_i)_{i\in I}  & I~ \text{finite}}
	\\[5pt]
	\infer[\typingruleAbs]{\typectx \vdash \la\var\tm \hastype \mtype \typearrow \ltype}{\typectx, \var \hastype \mtype \vdash \tm \hastype\ltype}
	 &
	\infer[\typingruleApp]{\typectx \uplus \typectxtwo \vdash \tm\tmtwo \hastype \ltype}{ \typectx \vdash \tm \hastype \mtype \typearrow \ltype & \typectxtwo \vdash \tmtwo \hastype \mtype  }
\end{tabular}
\end{tabular}
\caption{De Carvalho's multi type system.}
\label{fig:plain-multi-type-system}
\end{figure}

\paragraph{Typing Rules.} 
\emph{Judgments} have shape $\typctx \vdash \tm \hastype \ltype$ or $\typctx \vdash \tm \hastype \mtype$, where $\tm$ is a \lam-term, $\mtype$ is a multi type, $\ltype$ is a 
linear type, and $\typctx$ is a \emph{type environment}, \ie, a total function from variables to multi types such that $\domain{\typctx} \defeq \{\var \mid \typctx(\var) \neq \emptymset\}$ is finite. We say that $\typctx$ is \emph{empty} if $\domain{\typctx} = \emptyset$.  
We write $\var_1 \hastype \mtype_1, \dots, \var_n \hastype \mtype_n$ for the environment $\Gamma$ such that $\typctx(y) = \mtype_{i}$, if $y = \var_i\in\vec x$, $\Gamma(y) = \emptymset$, otherwise. 

Note that the application rule $\ruleAp$ requires the argument to be typed with a multi type $\mtype$, which is necessarily introduced by rule $\ruleMany$,  having as hypotheses a multiset of derivations, indexed by a possibly empty set $I$. When $I$ is empty, the rule $\ruleAp$ has no premises and can type every term with~$\emptymset$. For instance, $\vdash \Omega\hastype\zero$ is derivable, but no linear type can be assigned to $\Omega$. 
Intuitively, $\zero$ is the type of erasable terms, and every \lam-term is erasable in the (call-by-name) $\l$-calculus.

\paragraph{Technicalities about Types.} 
The \emph{multiset union} is denoted by $\mplus$ and is extended to type environments pointwisely,
\ie\  $(\typctx \mplus \typctxtwo)(\var) \defeq \typctx(\var) \mplus \typctxtwo(\var)$, for all $\var\in\Var$.
This notion is extended further to a finite family of type environments as expected. 
In particular, if $J = \emptyset$ we let $\bigmplus_{i \in J\!} \typctx_i$ be the empty environment.
Given two type environments $\typctx$ and $\typctxtwo$ having disjoint domain $\domain{\typctx} \cap \domain{\typctxtwo} = \emptyset$, we simply write $\typctx, \typctxtwo$ for $\typctx\mplus\typctxtwo$.
Note that $\typctx, \var \hastype \emptymset = \typctx$, where we implicitly assume $\var \notin \domain{\typctx}$. 
We write $\concl{\tderiv}{\typctx}{\tm}{\type}$ whenever $\tderiv$ is a (\emph{type}) \emph{derivation} (\ie a finite tree constructed bottom up by applying the rules in \reffig{plain-multi-type-system}) with as conclusion the judgment $\typctx \vdash \tm \hastype \type$.
We write $\derive{\tderiv}{\tm}$ if $\concl{\tderiv}{\typctx}{\tm}{\type}$, for some type environment $\typctx$ and some type $\type$.

\paragraph{Type Preorder and Relational Semantics.} 
The multi type system induces a notion of semantic interpretation into what is known as \emph{relational model} of (the call-by-name) \lam-calculus.
The interpretations of \lam-terms are naturally ordered by set-theoretical inclusion, and this induces a preorder on \lam-terms, namely the inequational theory of the model. 

\begin{definition}[Relational interpretation and type preorder]\label{def:relsem}~
\bsub\item
The \emph{relational interpretation} $\sem\tm$ \emph{of a \lam-term} $\tm$ is defined as follows:
\begin{center}$\begin{array}{lllllll}
	\sem\tm &\defeq &\{(\typctx,\ltype) \mid 
	\exists 
	\, 
	\concl{\tderiv}{\typctx}{\tm}{\ltype} \}.
\end{array}$\end{center}
	\item The \emph{type preorder} $\leqtypepl$ on ordinary $\l$-terms is defined as $\tm \leqtypepl \tmtwo$ if $\sem\tm \subseteq \sem \tmtwo$, 
	and the induced \emph{type equivalence} is noted $\equivtype$.
	\esub
\end{definition}

We have the following fundamental properties of the multi type preorder.

\begin{theorem}[{\citet{BreuvartMR18}}]\label{thm:standardprops}
 Let $t,u\in\Lambda$.
\begin{enumerate}
\item \emph{Compatibility}: if $\tm \leqtypepl \tmtwo$ then $\ctxp\tm \leqtypepl \ctxp\tmtwo$, for every context $\ctx$.
\item \emph{$\beta$-invariance}: if $\tm \to_\beta \tmtwo$ then $\sem\tm = \sem\tmtwo$.
\item\label{thm:standardprops3} \emph{$\eta$-reduction}: if $\tm \to_\eta \tmtwo$ then $\sem\tm \subseteq \sem\tmtwo$.
\item\label{thm:standardprops4} \emph{No $\eta$-expansion}: $\lam y.xy \to_\eta x$, but $\sem{x} \not\subseteq \sem{\lam y.xy}$.
\end{enumerate}
\end{theorem}
The failure of $\eta$-expansion, \ie Point \ref{thm:standardprops4} of the theorem, is due to the fact that $\var$ can be typed with the atomic type $\vartype$ using an axiom, while there is no way of typing $\lam y.xy$ with $\vartype$, since it can only be typed with an arrow type.

\begin{corollary}
The relation $\leqtypepl$ is an inequational $\l$-theory. Moreover, $\leqtypepl$ is semi-extensional but not extensional.
\end{corollary}

The corollary captures the soundness of Engeler's relational model. It is possible to prove that $\leqtypepl\,\subseteq\,\obsle$, thus obtaining a semantic proof that $\obsle$ is an inequational theory (\refthm{head-ctx-inequational}). In \refsect{bohm-included-in-interaction}, we shall follow this approach for proving that the interaction preorder is an inequational theory.

\paragraph{Adequacy with Respect to Head Reduction.} For showing that typable terms are head terminating, 
we need a notion of size for type derivations, that shall bound the number of head steps.

\begin{definition}[Size]
 Let $\tderiv$ be a type derivation. 	The \emph{(applicative) size} $\insize{\tderiv}$ of $\tderiv$ is the number of  occurrences of rules $\ruleAp$ in $\tderiv$.
\label{def:two-sizes}
\end{definition}

\begin{proposition}[\cite{BarendregtM22}]
	Let $\tm, \tm'\in\Lambda$ be such that $\tm \toh \tm'$.
	\label{prop:qual-subject} 
	\begin{enumerate}
		\item \label{p:qual-subject-reduction}
		\emph{Quantitative subject reduction}: 	if $\tderiv\exder \typctx \vdash \tm\hastype \ltype$ then there exists a derivation $\tderiv'\exder \typctx \vdash \tm' \hastype \ltype$ such that $\insize{\tderiv'} = \insize{\tderiv} -1$.
		
		\item \label{p:qual-subject-expansion}
		\emph{Subject expansion}: if $\tderiv'\exder \typctx \vdash \tm' \hastype \ltype$ then there is a derivation $\tderiv\exder \typctx \vdash \tm\hastype \ltype$.
	\end{enumerate}
\end{proposition}

Note the quantitative aspect of subject reduction (\refpropp{qual-subject}{reduction}), stating that the derivation size \emph{strictly decreases} along head steps. 
It does not say that it decreases at \emph{arbitrary} $\beta$-steps because the contraction of redexes occurring in sub-terms typed with rule $\ruleMany$ might not change the size. 
For instance, if $\var \tm \to_\beta \var\tm'$ and $\tm$ is typed using an empty $\ruleMany$ rule (\ie with 0 premises), which is a sub-derivation of size 0, then also $\tm'$ is typed using an empty $\ruleMany$ rule, of size 0. In fact, not all typable terms are $\beta$-normalizable: $\var\Omega$ is typable as follows, for any linear type $\ltype$, but it has no $\beta$-nf:\begin{equation}
\AxiomC{}
\RightLabel{$\ruleAx$}
\UnaryInfC{$\var\hastype \mset{\zero\typearrow\ltype} \vdash \var\hastype \zero\typearrow\ltype$}
\AxiomC{}
\RightLabel{$\ruleMany$}
\UnaryInfC{$\vdash \Omega\hastype \zero$}
\RightLabel{$\ruleAp$}
\BinaryInfC{$\var\hastype \mset{\zero\typearrow\ltype} \vdash \var\Omega\hastype \ltype$}
\DisplayProof
\label{eq:var-omega-tderiv}
\end{equation}
Since the size of type derivations decreases at every head step, it provides a termination measure (only) for the head reduction of typable terms. 
The fact that typable terms are head terminating is also called \emph{correctness} of the type system. 

\emph{Completeness} of the type system---\ie every head terminating term is typable---is obtained via typability of all head normal forms, proved next, and subject expansion (\refpropp{qual-subject}{expansion}). 

\begin{proposition}[Typability of head normal forms. \cite{BarendregtM22}]
\label{prop:typability-hnf}
Let $\hnf\in\Lambda$ be a head normal form. Then there exists a derivation $\tderiv \derives \typctx\vdash \hnf\hastype \ltype$.
\end{proposition}


Summing up, we obtain the following characterization of head normalization.

\begin{theorem}[Typability characterizes head normalization. \cite{BarendregtM22}]\label{thm:head-characterization}
Let $\tm\in\Lambda$. 
\begin{enumerate}
\item \emph{Correctness}: if $\derive{\tderiv}{\tm}$ then there exists a head normalizing evaluation $\tm \toh^n \hnf$ with $\hnf$ normal and $n  \leq \insize{\tderiv}$.

\item \emph{Completeness}: if $\tm \toh^* \hnf$ is a head normalizing sequence,
	then there exists a derivation $\derive{\tderiv}{\tm}$.
\end{enumerate}
Therefore $\sem\tm\neq\emptyset$ if and only if $\tm$ is head normalizable. In particular, $\leqtypepl$ is consistent.
\end{theorem}
\paragraph{Exact Bounds?} It is natural to wonder whether there are type derivations $\tderiv$ for which correctness holds with $n  = \insize{\tderiv}$. 
The answer is no: the type derivation in \refeq{var-omega-tderiv}, as well as any other type derivation for $\var\Omega$, has at least one $\ruleAp$ rule even if the \lam-term $\var\Omega$ is already a head normal form, whence $n=0$. 

The question has been studied and refined in the literature. Such a mismatch can be improved in two ways, both studied in-depth by \citet{DBLP:journals/jfp/AccattoliGK20}. The first one traces back to \citet{Carvalho07,DBLP:journals/mscs/Carvalho18}, and takes into account the number $\hdsize\hnf$ of application constructors in the spine of the head normal form $\hnf$. Then type derivations $\tderiv$ satisfying a certain \emph{tight predicate} verify $n +\hdsize\hnf = \insize{\tderiv}$. The second one is developed in \citet{DBLP:journals/jfp/AccattoliGK20}. It introduces:
\begin{itemize}
\item A second set of typing rules assigning some new type constants to the constructors that occur in the head normal form, and;
\item A tight predicate forcing all such constructors to be typed with these alternative rules.
\end{itemize}
In such a system, one can actually obtain $n = \insize{\tderiv}$ when the tight predicate holds. 
The drawback is that in this case one obtains a constants-only type that cannot be composed with any other type. 

For the new type system of the next section, we shall give in \refsect{bohm-included-in-interaction} a refined technique for exact bounds that exploits player tags. We shall measure the exact number of interaction head steps without resorting to constants-only types, which is a novelty.

\section{Checkers Multi Types}
\label{sect:multi-types}
In this section, we introduce a system $\ctypes{}$ of multi types for the checkers calculus, that can be seen as an annotated version of the standard one presented in Section~\ref{sec:relsem}. We shall prove that the new system characterizes termination of checkers head reduction $\tohch$, similarly to the standard system. Despite the similarity, however, the two systems are inherently different because the new one shall \emph{not} be invariant under $\eta$-reduction, while the standard one is (\refthm{standardprops}(\ref{thm:standardprops3})).

\begin{definition} Types and typing rules of the checkers multi types system $\ctypes{}$ are given in \reffig{colored-types-for-cbn}.
\end{definition}

\begin{figure}[t!]	
\begin{center}
\arraycolsep=3pt
\begin{tabular}{c}
	$\begin{array}{rrll @{\hspace{15pt}} l@{\hspace{12pt}} |@{\hspace{12pt}} rrll }
	\textsc{Linear types} & \ltype, \ltypetwo &\grameq& \vartype \mid \mtype \typearrowpp{\clr{}}{\clrtwo{}} \ltype & \clr{},\clrtwo{}\in\{\redclr,\blueclr\}
&
	\textsc{Generic types} & \gtype, \gtypetwo &\grameq& \ltype \mid \mtype
	\\
	\textsc{Multi types} & \mtype, \mtypetwo &\grameq& \multitype{n}{\ltype} & n\geq 0
&
		\textsc{Zero} & \zero &\grameq& \mset{\ }

	\end{array}$
\\[15pt]
\hline
\\[-6pt]
	\begin{tabular}{c\hcolspace c\hcolspace | \hcolspace c}
	\infer[\typingruleAx]{\var \hastype [\ltype] \ctypes{0} \var \hastype \ltype}{}
	&	
	\infer[\typingruleMany]{\uplus_{i\in I}\typectx_i \ctypes{\sum_{i\in I} k_i}  \tm \hastype [\ltype_i]_{i\in I}}{(\typectx_i \ctypes {k_i} \tm \hastype \ltype_i)_{i\in I}  & I~ \text{finite}}
	 &
	\infer[\typingruleAppNoInt]{\typectx \uplus \typectxtwo \ctypes {k_1+k_2} \ccapp{}\tm\tmtwo \hastype \ltype}{ \typectx \ctypes {k_1} \tm \hastype \mtype \typearrowpp{\clr{}}{\clr{}} \ltype & \typectxtwo \ctypes{k_2} \tmtwo \hastype \mtype  }
	
	\\[5pt]
	\infer[\typingruleAbs]{\typectx \ctypes k \cla{}\var\tm \hastype \mtype \typearrowpp{\clr{}}{\clrtwo{}} \ltype}{\typectx, \var \hastype \mtype \ctypes k \tm \hastype\ltype}
	 &
	\infer[\typingruleApp]{\typectx \uplus \typectxtwo \ctypes {k} \capp{}\tm\tmtwo \hastype \ltype}{ \typectx \ctypes {k_1} \tm \hastype \mtype \typearrowpp{\clr{}}{\clrtwo{}} \ltype & \typectxtwo \ctypes{k_2} \tmtwo \hastype \mtype  }
		 &
	\infer[\typingruleAppInt]{\typectx \uplus \typectxtwo \ctypes {k_1+k_2+1} \appp{\clr{}^\bot}\tm\tmtwo \hastype \ltype}{ \typectx \ctypes {k_1} \tm \hastype \mtype \typearrowpp{\clr{}}{\clr{}^\bot} \ltype & \typectxtwo \ctypes{k_2} \tmtwo \hastype \mtype  }

\end{tabular}
\end{tabular}
\end{center}\smallskip
In rule $\typingruleApp$, $k=k_1 {+} k_2$ if $\clr{}=\clrtwo{}$, otherwise $k= k_1 {+} k_2 {+} 1$. Rule $\typingruleApp$ compactly sums up rules $\typingruleAppNoInt$ and $\typingruleAppInt$.

\caption{Checkers multi type system $\ctypes{}$.}
\label{fig:colored-types-for-cbn}
\end{figure}
\paragraph{Main Ideas 1: Checkers Arrows.} We start by turning the arrow type $\mtype\typearrow\ltype$ into a checkers arrow type $\mtype\typearrowpp{\clr{}}{\clrtwo{}}\ltype$ carrying \emph{two} tags $\clr{},\clrtwo{}\in\set{\bullet,\circ}$, thus giving rise to four possible player combinations. 
The idea is that if, say, $\tm:\mtype\typearrowpp{\redclr{}}{\blueclr{}}\ltype$ then $\tm$ can only be applied via $\blueclr{}$-applications and if it reduces to an abstraction then it must be a $\redclr{}$-abstraction. 
In the typing rule $\typingruleAbs$ for abstractions, the first player $\clr{}$ is determined by the external abstraction, while the second player $\clrtwo{}$ can be freely chosen. 
We shall refer to $\typearrowpp{\redclr{}}{\blueclr{}}$ and $\typearrowpp{\blueclr{}}{\redclr{}}$ as \emph{interaction arrow types}, and to $\typearrowpp{\blueclr{}}{\blueclr{}}$ and $\typearrowpp{\redclr{}}{\redclr{}}$ as \emph{silent arrow types}.

\paragraph{Main Ideas 2: Interaction Application Rule and Index.} The second main ingredient is that there are now two application rules $\typingruleAppNoInt$ and $\typingruleAppInt$, one for the application of silent arrows and one for interaction arrows. We provide also a rule $\typingruleApp$ that sums them up compactly. 
Moreover, typing judgement $\typctx \ctypes k \tm\hastype \gtype$ now carry an index $k$ which counts the number of $\typingruleAppInt$ rules in the derivation. 
These rules, intuitively, type interaction steps, which can be \emph{factual} or \emph{potential}, as we now explain. 
The application of an abstraction typed with an interaction arrow type gives rise to an interaction step $\tobint$, whence it is a factual interaction. 
The application of a free variable typed with an interaction arrow, for instance, does \emph{not} give rise to an interaction step, but that \emph{potential} interaction is recorded in the type and it might arise if the term is plugged in a context, according to its type.

When the interaction index $k$ is irrelevant, we omit it and simply write $\typctx \ctypes{} \tm\hastype \gtype$.

\begin{definition}[Checkers relational interpretation and type preorder]\ 
\bsub
\item
The \emph{checkers relational interpretation} $\semint\tm$ \emph{of a checkers term} $\tm\in\Lambdac$ is defined by:
\begin{center}
$\begin{array}{lllllll}
	\semint\tm &\defeq &\{((\typctx,k,\ltype) \mid 
	\exists 
	\, 
	\tderiv\exder \typctx \ctypes k \tm \hastype \ltype \}.
\end{array}$
\end{center}
	\item The  type preorder $\leqctype$ on checkers $\l$-terms $\tm,\tmtwo\in\Lambdac$ is defined as $\tm \leqctype \tmtwo$ if $\semint\tm \subseteq \semint\tmtwo$, and the induced type equivalence is noted $\equivctype$. 
	\item The black type preorder $\leqbtype$ on ordinary $\l$-terms $\tm,\tmtwo\in\Lambda$ is defined as $\tm \leqbtype \tmtwo$ if $\monoToBlue\tm \leqctype \monoToBlue\tmtwo$, and the induced type equivalence is noted $\equivbtype$. 
\esub
\end{definition}

We have the following fundamental properties of the multi type preorder.

\begin{toappendix}
\begin{theorem}
\label{thm:checkers-type-compatibility}
Let $\tm,\tmtwo\in \Lambdac$.\NoteProof{thm:ch-types-compatibility}\label{thm:ch-types-compatibility}
\begin{enumerate}
\item \emph{Compatibility}: if $\tm \leqctype \tmtwo$ then $\ctxp\tm \leqctype \ctxp\tmtwo$, for every context $\ctx$.
\item \emph{Silent $\beta$-invariance}: if $\tm \tobnoint \tmtwo$ then $\semint\tm = \semint\tmtwo$.
\item \emph{No $\eta$-reduction}: for all players $\clr{},\clrtwo{}\in\set{\bullet,\circ}\,.\,\semint{\cla{}\vartwo\capp{}\var\vartwo} \not\subseteq \semint\var$.
\item \emph{No $\eta$-expansion}: for all players $\clr{},\clrtwo{}\in\set{\bullet,\circ}\,.\,\semint\var \not\subseteq \semint{\cla{}\vartwo\capp{}\var\vartwo}$.
\end{enumerate}
\end{theorem}
\end{toappendix}
The fact that the checkers relational interpretation invalidates $\eta$-reduction, is specific to interaction arrow types.
As an example of \refthm{ch-types-compatibility}(3), consider the black $\eta$-expansion of $\var$:
\begin{center}
$\infer[\typingruleAbs]{
	\var\hastype [\emptytype \typearrowpp{\redclr}{\blueclr} \ltype] \ctypes{1} \bla\vartwo \bapp{\var}{\vartwo}\hastype \emptytype \typearrowpp{\blueclr}{\clr{}} \ltype 
	}{
	\infer[\typingruleAppInt]{
		\var\hastype [\emptytype \typearrowpp{\redclr}{\blueclr} \ltype] \ctypes{1} \bapp{\var}{\vartwo}\hastype \ltype 
		}{ 
		\infer[\typingruleAx]{
			\var\hastype [\emptytype \typearrowpp{\redclr}{\blueclr} \ltype] \ctypes{0} \var \hastype \emptytype \typearrowpp{\redclr}{\blueclr} \ltype
			}{}
		~~~~&~~~~
		\infer[\typingruleMany]{
			\ctypes{0} \vartwo \hastype \emptytype
			}{}
		}
	}$
\end{center}
and note that instead $\var\hastype [\emptytype \typearrowpp{\redclr}{\blueclr} \ltype] \not \ctypes{~k} \var\hastype \emptytype \typearrowpp{\blueclr}{\clr{}} \ltype$, for all indices $k$ and players $\clr{}$. An analogous typing derivation shows that $\bla\vartwo\rapp\var\vartwo\not\leqctype\var$.

\begin{toappendix}
\begin{corollary}
The black type preorder $\leqbtype$ is an inequational \lam-theory. Moreover, $\leqbtype$ is neither extensional nor semi-extensional.\label{coro:btype-preorder-inequational}\NoteProof{coro:btype-preorder-inequational}
\end{corollary}
\end{toappendix}

\paragraph{Adequacy with Respect to Checkers Head Reduction.} As in the plain type system, the applicative size $\insize{\tderiv}$ of type derivations (which is still defined as the number of rules $\ruleAp$ in $\tderiv$, even if the rule itself has changed) decreases with each head step $\tohch$. 
The difference, however, is that if the step is an interaction one---and only in that case---then also the index $k$ decreases by exactly 1. 

\begin{toappendix}
\begin{proposition}\label{prop:ch-subject}
	Let $\tm,\tm'\in\Lambdac$ be such that $\tm \tohch \tm'$.
	\NoteProof{prop:ch-subject} 
	\begin{enumerate}
		\item 
		\emph{Quantitative subject reduction}: 	if $\tderiv\exder \typctx \ctypes k \tm\hastype \ltype$ then there is a derivation $\tderiv'\exder \typctx \ctypes {k'} \tm' \hastype \ltype$ such that $\insize{\tderiv'} = \insize{\tderiv} -1$. Moreover, if $\tm \tohint \tm'$ then $k'=k-1$ and if $\tm \tohnoint \tm'$ then $k'=k$.
		
		\item 
		\emph{Subject expansion}: if $\tderiv'\exder \typctx \ctypes{} \tm' \hastype \ltype$ then there is a derivation $\tderiv\exder \typctx \ctypes{} \tm\hastype \ltype$.
	\end{enumerate}
\end{proposition}
\end{toappendix}

As before, quantitative subject reduction entails the correctness of the checkers type system, and the following typability of head normal forms gives its completeness.
\begin{toappendix}\begin{proposition}[Typability of head normal forms]
\label{prop:ch-typability-hnf}
Let $\hnf\in\Lambdac$ be a head normal form. Then there exists a derivation $\tderiv \derives \typctx\ctypes{} \hnf\hastype \ltype$.\NoteProof{prop:ch-typability-hnf}
\end{proposition}\end{toappendix}

Typability of all head normal forms (\refprop{ch-typability-hnf}) together with subject expansion (\refprop{ch-subject}(2)) implies the \emph{completeness} of the type system: every head terminating term is typable. Summing up, we obtain the following characterization of head normalization.

\begin{toappendix}
\begin{theorem}[Typability characterizes head normalization]
	\label{thm:ch-head-characterization}
	Let $\tm\in\Lambdac$. 
	\NoteProof{thm:ch-head-characterization}
	\begin{enumerate}
		\item \emph{Correctness}: if $\tderiv \exder \typctx \ctypes{k} \tm\hastype\ltype$ then there exists a $\hchsym$-normal form $\hnf$ and an evaluation sequence $\evseq: \tm \tohch^n \hnf$ with $n  \leq \insize{\tderiv}$ and such that the number of interaction steps in $\evseq$ is $\leq k$.
		
		\item \emph{Completeness}: if $\tm \tohch^* \hnf$ is a head normalizing sequence,
		then there exists $\tderiv \exder \typctx \ctypes{} \tm\hastype\ltype$.
	\end{enumerate}
	Therefore, $\semint\tm\neq\emptyset$ if and only if $\tm$ is $\hchsym$-normalizable. 
\end{theorem}
\end{toappendix}

\section{From the \bohm Preorder to the Interaction One, via Multi Types}
\label{sect:bohm-included-in-interaction}
In this section, we use the obtained results about checkers multi types to prove the chain of inclusions $\btpre\,\subseteq\, \leqbtype\,\subseteq\, \intleq$, as to complete the proof that $\btpre\,=\, \intleq$. The first inclusion $\btpre\,\subseteq\, \leqbtype$ is simple: it follows from an easy induction on the size of checkers type derivations, exploiting the properties of the checkers type system. The second inclusion $\leqbtype\,\subseteq\, \intleq$ requires slightly more work. The key point is to show that the type preorder preserves the number of interaction steps during head normalization in $\Lambdac$, that is, that if $\tm\bshcol{k}$ and $\tm\leqctype\tmtwo$ then $\tmtwo\bshcol{k}$. Such a property requires to characterize a class of checkers type derivations whose index captures exactly the number of head interaction steps to head normal form. We do it via a notion of \emph{tight} typing.

\paragraph{The First Inclusion} The proof of the following proposition goes by a simple induction on the size of derivations. The use of \emph{quantitative} subject reduction (hence of multi types, instead of idempotent intersection types) is critical in order for the induction argument to go through.

%
%
%


\begin{theorem}[The \bohm preorder is included in the checkers type preorder]
	\label{thm:bisimulation-preserves-typeder}
Let $\tm,\tmtwo\in\Lambda$. If $\tm \btpre \tmtwo$ then $\tm \leqbtype \tmtwo$.
\end{theorem}

\begin{proof} 
Assume $\tm \btpre \tmtwo$. If $\tm\bshdiv$ then by Proposition~\ref{prop:embedding-of-head-reduction} also $\monoToBlue{\tm}$ is not $\hchsym$-normalizable. From~\refthm{ch-head-characterization}, we obtain $\semint{\monoToBlue{\tm}} = \emptyset \subseteq \semint{\monoToBlue{\tmtwo}}$ whence $\tm \leqbtype \tmtwo$ (by definition).

Otherwise, we must have $\tm\bsh{}\!\! h_t \defeq\lam x_1\dots x_m.y\,\tm_1\cdots \tm_n$ and $\tmtwo\bsh{}\!\! h_u \defeq\lam x_1\dots x_m.y\,\tmtwo_1\cdots \tmtwo_n$, with $(\tm_i\btpre\tmtwo_i)_{i \le n}$.
Let us take any typing $(\typctx,k,\ltype) \in\semint{\monoToBlue{\tm}}$, and show that it belongs to $\semint{\monoToBlue{\tmtwo}}$.
We proceed by induction on the size $\insize{\tderiv}$ of a derivation $\tderiv\exder \typectx \ctypes {k} \monoToBlue\tm \hastype \ltype$.
By \refprop{embedding-of-head-reduction}(1) we get $\monoToBlue{\tm}\tohnoint^*\monoToBlue h_{\tm}$, and by quantitative subject reduction (\refprop{ch-subject}(1)) there is a derivation $\tderiv_{\mathrm{hnf}}\exder\ctypes{k'} \monoToBlue h_\tm \hastype \ltype$ such that $k = k'$ (since all head reductions involved are silent) and $\insize{\tderiv_{\mathrm{hnf}}} \le \insize{\tderiv}$. Moreover, 
	\[\begin{array}{ccc}
	\tderiv_{\mathrm{hnf}} &=&
	\infer=[m\,\typingruleAbs]{
		\typctx \ctypes{k} \monoToBlue{\nftm} \hastype \ltype
		=
		\mtype_1 \typearrowpp{\blueclr}{\clrtwo{1}} \cdots \mtype_n \typearrowpp{\blueclr}{\clrtwo{m}} \ltype'
		}{\infer=[n\, \typingruleApp]{
			\typctx, (\var_i \hastype \mtype_i)_{1 \leq i\leq m} \ctypes{k} \monoToBlue{\vartwo\, \tm_1\, \cdots \,\tm_{n}} \hastype \ltype'
			}{
			\typctxtwo_0 \ctypes{0} \vartwo \hastype 
			\mtypetwo_1\typearrowpp{\clr{1}}{\blueclr}\cdots\mtypetwo_n \typearrowpp{\clr{n}}{\blueclr} \ltype'
			& 
			(\tderiv_{i}\exder \typctxtwo_i \ctypes{k_i} \monoToBlue\tm_{i} \hastype \mtypetwo_{i})_{1 \leq i\leq n}
			}
		}\end{array}\]
where $n\,\typingruleApp$ (resp.\ $m\,\typingruleAbs$) denotes $n$ (resp.\ $m$) consecutive applications of the rule, and
\begin{itemize}
		\item[-] $\typctx, (\var_i \hastype \mtype_i)_{1 \leq i\leq m} = \biguplus_{0 \leq i \leq n} \typctxtwo_i$;
		\item[-] $\sum_{1 \leq i\leq n} (k_i + \gamma_{\blueclr}(\clr{i}))= k$, where $\gamma_{\blueclr}(a) \defeq 0$ if $a = \blueclr$, and $\gamma_{\blueclr}(a) \defeq 1$ otherwise.
	\end{itemize}
Note that, for every $i\le n$, $\tm_i\btpre\tmtwo_i$ and $\insize{\tderiv_i} <  \insize{\tderiv_{\mathrm{hnf}}} \le \insize{\tderiv}$, so we can apply the \ih to each linear type in the multi type $\mtypetwo_i$ and get a derivation $\typctxtwo_i \ctypes{k_i} \monoToBlue{\tmtwo_{i}} \hastype \mtypetwo_{i}$. We use these derivations to build the appropriate derivation for $\nftmtwo$:
	\[
	\infer=[m\,\typingruleAbs]{
		\typctx \ctypes{k} \monoToBlue{\nftmtwo} \hastype \ltype
		=
		\mtype_1 \typearrowpp{\blueclr}{\clrtwo{1}} \cdots \mtype_n \typearrowpp{\blueclr}{\clrtwo{m}} \ltype'
		}{\infer=[n\, \typingruleApp]{
			\typctx, (\var_i \hastype \mtype_i)_{1 \leq i\leq m} \ctypes{k} \monoToBlue{\vartwo\, \tmtwo_1\, \cdots \,\tmtwo_{n}} \hastype \ltype'
			}{
			\typctxtwo_0 \ctypes{0} \vartwo \hastype 
			\mtypetwo_1\typearrowpp{\clr{1}}{\blueclr}\cdots\mtypetwo_n \typearrowpp{\clr{n}}{\blueclr} \ltype'
			& 
			(\tderiv_{\tmtwo_i}\exder \typctxtwo_i \ctypes{k_i} \monoToBlue\tmtwo_{i} \hastype \mtypetwo_{i})_{1 \leq i\leq n}
			}
		}\]
By \refprop{embedding-of-head-reduction}(1) we get $\monoToBlue{\tmtwo}\tohnoint^*\monoToBlue h_{\tmtwo}$ and by subject expansion (Prop.~\ref{prop:ch-subject}(2)) we obtain $\typctx \ctypes{k'} \monoToBlue{\tmtwo} \hastype \ltype$.
Since all head reductions involved are silent, we conclude $k = k'$ and $(\typctx,k,\ltype) \in\semint{\monoToBlue{\tmtwo}}$. 

(Note that a `hidden' base case is when $n = 0$, as it does not require the induction hypothesis.)\end{proof}

\paragraph{The Second Inclusion, via Tight Typings} Now, we consider a \emph{tight} predicate on type judgements that forces the index $k$ on the type derivation to be exactly the number of head interaction steps to head normal form, as we show below. An essential aspect is the fact that the predicate concerns \emph{types}, and not \emph{type derivations}, so that it can be transferred from $\tm$ to $\tmtwo$ when they are related by the type preorder, that is, when $\tm\leqctype\tmtwo$.

\begin{definition}[Tight typings and derivations]
	Let $\tm\in\Lambdac$ and $\typctx\ctypes{} \tm\hastype \ltype$ be a type judgement. The pair $(\typctx,\ltype)$ is a \emph{tight typing} if: 
	\begin{enumerate}
	\item All multi types $\mtype$ occurring in $\typctx$ and $\ltype$ are empty except for one, and 
	\item All arrows in $(\typctx,\ltype)$ are silent. 
	\end{enumerate}
	For ease of language, we shall also say that a derivation $\tderiv\exder \typctx\ctypes{} \tm\hastype \ltype$ is tight if $(\typctx,\ltype)$ is a tight typing.
\end{definition}
The definition of tight typing can actually be slightly weakened, by asking that only the arrows in the non-empty multi type are silent, without loosing any of its properties. This weakened definition, however, is slightly more technical, which is why we avoid it.

The crucial property ensured by tightness is that tightly typed head normal forms have interaction index 0, and that a tight typing can be derived for every head normal form.

\begin{toappendix}
\begin{proposition}[Tightness and head normal forms]\label{prop:nf-are-typable-tight}
Let $\htm\in\Lambdac$ be a $\hchsym$-normal form.\NoteProof{prop:nf-are-typable-tight}
\begin{enumerate}
	\item \emph{Existence:} there exists a tight derivation $\typctx \ctypes{k} \htm \hastype \ltype$;
	\item \emph{Zero interaction:} if $\typctx \ctypes k \htm \hastype \ltype$ is tight then $k=0$.
\end{enumerate}
\end{proposition}
\end{toappendix}

From the properties of tightness for head normal forms and the characterization of head reduction (\refthm{ch-head-characterization}), we obtain the following refined characterization.

\begin{theorem}[Tight characterization]
	\label{th:tight-characterization}
	Let $\tm\in\Lambdac$. 
	\begin{enumerate}
	\item\label{th:tight-characterization1} \emph{Correctness}: if $\tderiv : \typctx \ctypes k \tm \hastype \ltype$ is tight then there exists a head normal form $\htm$ and an evaluation sequence $\evseq: \tm \tohch^* \htm$ such that the number of interaction steps in $\evseq$ is exactly $k$.
	\item\label{th:tight-characterization2} \emph{Completeness}: if $\evseq: \tm \tohch^* \htm$ with $\htm$ head normal, and $k$ is the number of interaction steps in $\evseq$, then there exists a tight derivation $\tderiv : \typctx \ctypes k \tm \hastype \ltype$.
	\end{enumerate}
\end{theorem}


\begin{proof}
\hfill
\begin{enumerate}
	\item 	By correctness (\refthm{ch-head-characterization}(1)), there exists an evaluation sequence $\evseq: \tm \tohch^* \htm$ such that the number of interaction steps in $\evseq$ is $k'\leq k$.
	By quantitative subject reduction (\refprop{ch-subject}(1)), $\typctx \ctypes {k-k'} \htm \hastype \ltype$. Since $(\typctx,\ltype)$ is tight, by the zero interaction property of tight typings (\refprop{nf-are-typable-tight}(2)) we obtain $k-k'=0$, that is, $k'=k$.

\item 	By the existence of tight derivations for head normal forms (\refprop{nf-are-typable-tight}(1)), we obtain a tight derivation $\tderiv \exder \typctx \ctypes 0 \htm \hastype \ltype$. By subject expansion (\refprop{ch-subject}), the same typing types $\tm$ but with an index $k'$, \ie $\typctx \ctypes{k'} \tm \hastype \ltype$. By Point \ref{th:tight-characterization1}, $\tm \tohch^* \htmtwo$ for some head normal form $\htmtwo$ doing $k'$ interaction steps. By determinism of $\tohch$, $\htm=\htmtwo$ and $k=k'$.\qedhere	
	\end{enumerate}
\end{proof}

The tight characterization is then used to show that the type preorder is sound with respect to the interaction preorder.

\begin{corollary}
\label{coro:type-preorder-included-ctxc}\hfill
\begin{enumerate}
		\item \emph{Tightness of the checkers type preorder}:  let $\tm,\tmtwo\in\Lambdac$. If $\tm\bshcol k$ and $\tm\leqctype\tmtwo$ then $\tmtwo\bshcol k$.
\item\label{coro:type-preorder-included-ctxc2} \emph{Soundness of the checkers type preorder}: let $\tm,\tmtwo\in\Lambdac$. If $\tm\leqctype\tmtwo$ then $\tm\leqchcol\tmtwo$.
\item\label{coro:type-preorder-included-ctxc3} \emph{On ordinary \lam-terms}: let $\tm,\tmtwo\in\Lambda$. If $\tm\leqbtype\tmtwo$ then $\tm\intleq\tmtwo$.
\end{enumerate}
\end{corollary}
\begin{proof}
(\ref{coro:type-preorder-included-ctxc3}) follows immediately from (\ref{coro:type-preorder-included-ctxc2}).
\begin{enumerate}
\item Let $\tm$ and $\tmtwo$ such that $\tm \leqctype \tmtwo$ and $\tm \bshcol{k}$. 
By tight completeness (\refth{tight-characterization}(2)),  there exists a tight typing $\typctx,\ltype$ such that $\typctx \ctypes{k} \tm \hastype \ltype$. 
By $\tm\leqctype\tmtwo$, we obtain a derivation of $\typctx \ctypes{k} \tmtwo \hastype \ltype$. 
By tight correctness (\refth{tight-characterization}(1)) and tightness of $(\typctx,\ltype)$, we obtain $\tmtwo \bshcol{k}$.
			\item Let $\tm,\tmtwo\in\Lambdac$ such that $\tm\leqctype\tmtwo$. By compatibility of $\leqctype$ (\refthm{ch-types-compatibility}(1)), $\ctxp\tm\leqctype\ctxp\tmtwo$ for all $\ctx\in\chcontexts$. By tightness of $\leqctype$ (Point 1), if $\ctxp\tm\bshcol{k}$ then $\ctxp\tmtwo\bshcol{k}$. Hence, $\tm\leqchcol\tmtwo$.\qedhere
	\end{enumerate}
\end{proof}

We can now put the all the inclusions together, thus obtaining our main theorem.
\begin{theorem}[Tree and type characterizations of interaction equivalence]
The preorders $\btleq$, $\leqbtype$, and $\intleq$ coincide.
	Therefore, the \lam-theories $=_{\mathcal{B}}$, $\eqbtype$, and $\inteq$ coincide.\label{thm:final}
\end{theorem}
\begin{proof} $(\btleq \,\subseteq\, \leqbtype)$
By \refthm{bisimulation-preserves-typeder}.
$(\leqbtype \,\subseteq\, \intleq)$ By \refcoro{type-preorder-included-ctxc}\eqref{coro:type-preorder-included-ctxc3}.
 $(\intleq \,\subseteq\, \btleq)$ By \refth{intleq-included-in-bohm}.
\end{proof}

\paragraph{Back to a Delayed Proof} We can now finally prove \refthm{silconv-included-in-inter-ctx}, stating that silent conversion $\silconv$ is included in the checkers interaction preorder $\intleq$, which is the key point of the proof that the interaction preorder $\intleq$ is an inequational \lam-theory (\refcoro{requirements-for-inequational}).

\gettoappendix{thm:silconv-included-in-inter-ctx}
\begin{proof}
\applabel{delayed-proof}
From \refthm{checkers-type-compatibility}(2), if $\tm\silconv\tmtwo$ then $\tm \leqctype\tmtwo$. By \refcoro{type-preorder-included-ctxc}(2), we obtain $\tm\leqchcol\tmtwo$.
\end{proof}

\paragraph{Interaction Improvement and $\eta$} We provided \bohm tree characterizations of $\intleq$ and $\inteq$, but not of the interaction improvement $\intrleq$. Nonetheless, we almost have one. In the previous sections, we obtained the following chain of relationships:
\begin{equation}
\begin{array}{cccccccccc}
\btleq &=_{T.\ref{thm:final}} &\intleq &\subseteq_{\reflemmaeq{hierarchy-preorders}} &\intrleq& \subsetneq_{\reflemmaeq{hierarchy-preorders}} & \obsle & =_{T.\ref{thm:HW76}} & \etabtle
\end{array}
\label{eq:chain-preorders}
\end{equation}
That is, interaction improvement $\intrleq$ possibly enlarges the interaction preorder $\intleq$ and yet stays confined within the further fence of $\eta$-equivalence. Additionally, the fact that $\eta$-reduction can decrease the number of interactions (\refex{eta}) suggests the following.
\begin{conjecture}
\label{conj:improvement}
Interaction improvement $\intrleq$ is characterized by $\slbtleq$, the variant of $\etabtle$ (\refdef{HNFBISIM}) up to possibly infinite $\eta$-\emph{reduction} (rather than \emph{$\eta$-equivalence}).
 \end{conjecture}
\citet{BreuvartMR18} prove that the preorder $\slbtleq$ coincides with the preorder $\leqtypepl$ induced by (plain) multi types (\refdef{relsem}), which validates $\eta$-reduction (\refthm{standardprops}.\ref{thm:standardprops3}). As already mentioned, the difficulty for proving the conjecture is managing $\eta$-reduction in the checkers framework, since both rewriting techniques and checkers multi types fail to handle it. Actually, the variant of \bohm out technique of Section \ref{sec:bohmout} smoothly adapts from preorder to improvement, giving $\intrleq\,\subseteq\,\slbtleq$, that is, $\btleq\,\subseteq\,\intrleq\,\subseteq\,\slbtleq$. Therefore, conjecture \ref{conj:improvement} reduces to prove $\slbtleq\,\subseteq\,\intrleq$.

	\paragraph{White Contexts Are Enough} In the proof of the completeness theorem (\refth{intleq-included-in-bohm}), black terms are separated using only white contexts. The next corollary guarantees that, for the interaction preorder, white contexts are as discriminating as general checkers contexts. Such a strong fact validates the intuition that the interaction preorder amounts to consider the program and the  context as \emph{different} players.
	
	\begin{corollary}[The interaction preorder can be restricted to white contexts]
	Let $\tm,\tmtwo\in\Lambda$ and the preorder $\tm \wintleq \tmtwo$ be defined as: for all ordinary contexts $\ctx\in\Cctx$, if there exists $k$ such that $\monoToRed{\ctx}\ctxholep{\monoToBlue{\tm}}\bshcol{k}$ then $\monoToRed{\ctx}\ctxholep{\monoToBlue{\tmtwo}}\bshcol{k}$. Then, $\tm\intleq\tmtwo$ if and only if $\tm\wintleq\tmtwo$.
	\end{corollary}

\begin{proof}
\refth{intleq-included-in-bohm} shows that $\tm\wintleq\tmtwo$ implies $\tm\btpre\tmtwo$. \refthm{final} shows that $\tm\btpre\tmtwo$ implies $\tm\intleq\tmtwo$. Clearly $\tm\intleq\tmtwo$ implies $\tm\wintleq\tmtwo$, as $\monoToRed{\ctx}$ is a checkers context for any $\ctx\in\Cctx$. \qedhere
\end{proof}

\section{Related Work}\label{sec:relatedworks}
\paragraph{Improvements} Improvements \cite{SandsImprovementTheory,SandsToplas,DBLP:journals/tcs/Sands96} were developed in the '90s by Sands and co-authors to prove that various program transformations are time or space improvements \cite{DBLP:journals/tcs/Sands96,DBLP:conf/popl/MoranS99,DBLP:conf/icfp/GustavssonS01}, in the context of call-by-need evaluation, and have seen a revival in recent years \cite{DBLP:conf/icfp/HackettH14,DBLP:conf/lics/HackettH15,DBLP:journals/pacmpl/HackettH18,DBLP:conf/ppdp/Schmidt-Schauss18,DBLP:conf/sas/RielyP00,DBLP:conf/flops/MuroyaH24}.

\paragraph{Observational Equivalences and Trees} The characterization of observational equivalences in terms of equalities on trees originates in \cite{Hyland75,Hyland76,Wadsworth76}, and is an ongoing line of research whose state of the art is presented in \cite{IntrigilaMP19}. The question appears in the TLCA list of open problems \cite{Dezani01}. This paper contributes to this line of research, somewhat backwards: we introduce a new observational equivalence that matches a known equality on trees, namely non-extensional \bohm tree equality $\cB$.

\paragraph{\bohm Tree Equality} For the specific case of $\cB$, the literature already presents some corresponding observational equivalences. As hinted at in the introduction, however, they are \emph{partial} answers since they all extend the $\l$-calculus with computational primitives, or embed it into process algebras:
\begin{itemize}
\item \citet{Dezani98} add numerals, tests, and a non-deterministic choice operator, and show that $\cB$ corresponds to \emph{may convergence to a natural number} in their calculus. 
\item In the slightly different setting of \emph{weak} head reduction, \citet{DBLP:journals/iandc/Sangiorgi94} shows that \levy-Longo tree equality (the weak variant of $\cB$) corresponds to various contextual equivalences all obtained via extended settings, namely the $\pi$-calculus, $\l$-calculus with non-determinism, and $\l$-calculus with well-formed operators. Similarly, \citeauthor{DEZANICIANCAGLINI1999153} use another non-deterministic $\lambda$-calculus \citeyearpar{DEZANICIANCAGLINI1999153}, and \citet{BOUDOL199683} use a $\lambda$-calculus with modified syntax and rewriting rules where---crucially---terms can get stuck.
\item Recently, \citet{DBLP:conf/lics/SakayoriS23} provide a characterization of $\cB$ via the $\pi$-calculus.
\end{itemize}
In this paper, we do modify the $\l$-calculus via the checkers calculus, but we do \emph{not} add extra features, we only refine the analysis of $\beta$-reduction in a quantitative and interactive way.


\paragraph{Clocked $\l$-Calculus} The idea of discriminating \lam-terms by counting the head reduction steps in the construction of their B\"ohm trees also underlies the \emph{clocked \lam-calculus} of \citet{EndrullisHKP14,EndrullisHKP17}. The analogy ends there, as the clocked \lam-calculus does not allow one to separate the internal/silent steps of terms and contexts from the external/interaction steps between the two. 

\paragraph{Game Semantics} Our notion of interaction is inspired by the one between Player and Opponent in game semantics \cite{DBLP:journals/iandc/HylandO00,DBLP:journals/iandc/AbramskyJM00}---see \citet{Clairambault24} for the state of the art---in particular, silent steps are inspired by the hiding mechanism for composing games. We end up, however, with what seems to be a different setting. While the study of the relationship is left to future work, our interaction seems more basic (no need of the technical apparatus of game semantics), perfectly symmetrical, and allows for arbitrary shufflings of the two players in checkers terms, rather than keeping them apart as in game semantics. Moreover, with a few exceptions---namely, \cite{DBLP:conf/icalp/OngG02,DBLP:journals/tcs/KerNO03}---game models usually validate $\eta$-equivalence, while interaction equivalence does not. Notably, the game model by \citet{DBLP:journals/tcs/KerNO03} captures exactly $\cB$. 

\paragraph{Multi Types} The quantitative analysis of the relational semantics via multi types was pioneered by \citet{DBLP:journals/mscs/Carvalho18}. Then \citet{DBLP:journals/tcs/CarvalhoPF11,DBLP:conf/fossacs/BernadetL11} and especially \citet{DBLP:journals/jfp/AccattoliGK20} developed and extended that approach, that has been applied to a variety of settings, including classical logic \cite{DBLP:conf/lics/KesnerV20}, the probabilistic $\l$-calculus \cite{DBLP:journals/pacmpl/LagoFR21}, reasonable space \cite{DBLP:journals/pacmpl/AccattoliLV22}, and Bayesian networks \cite{DBLP:journals/pacmpl/FaggianPV24}.

\paragraph{Variants of Relational Semantics} Relational semantics is generalized along several directions by \citet{LairdMMP13,GrelloisM15,Ong17}.
\citeauthor{LairdMMP13}'s generalization of relations $r\subseteq A\times B$ from $r:A\times B\to \mathrm{Bool}$ to $r:A\times B\to \rel$ \citeyearpar{LairdMMP13}, where $\rel$ is a continuous semi-ring, might be used---in principle---to count interaction steps. The exact meaning of the coefficients in the interpretation of programs, however, is well-understood only at ground types, and remains unclear in the untyped case.
Any temptation of seeing our checkers type system as a dichromatic version of \citeauthor{GrelloisM15}'s \emph{colored linear logic} \citeyearpar{GrelloisM15} should be avoided: in their work, colors need to correctly ``match'' in a type derivation because they represent different levels of priority.

\paragraph{Cost-Aware Denotational Semantics} 
Another line of research aims at defining cost-sensitive denotational semantics based on \emph{sized} domains
\cite{DBLP:journals/pacmpl/KavvosMLD20,DBLP:journals/jfp/DannerL22}. These models are used to interpret a syntactic recurrence extracted from a given program, and prove a bounding theorem about such extraction. 
\citet{DBLP:conf/lics/Niu023} propose a synthetic language for cost-aware denotational semantics, endowed with phase-separated constructions of intensional and extensional computations. These approaches are designed for typed languages with constants, where observations are made at ground type, but they are hardly generalizable to the untyped case. 

There also are cost-aware game models such as \citeauthor{DBLP:conf/popl/Ghica05}'s slot games \citeyearpar{DBLP:conf/popl/Ghica05}---which capture Sands' improvements---and \citeauthor{DBLP:conf/fossacs/AlcoleiCL19}'s resource-tracking concurrent games \citeyearpar{DBLP:conf/fossacs/AlcoleiCL19}. The relationship between these concurrency-driven models and our approach is deferred to future investigation.
\section{Future Work}
\label{sect:future-work}




\paragraph{Weak Head.} Our study focuses on the paradigmatic case of head reduction, but it could be adapted to weak head reduction (sometimes called \emph{lazy} reduction \cite{DBLP:journals/iandc/AbramskyO93}). It is folklore that the Lévy-Longo tree preorder (a weak variant of \bohm trees) matches the weak head type preorder. We conjecture that the Lévy-Longo tree preorder  matches exactly the weak head variant of our interaction preorder.
Most of our study would adapt smoothly, but for the interaction \bohm out. 
The culprit is that there is no analogous of \bohm separation theorem for weak head reduction. 
There exist separation results but they all involve extensions of the $\l$-calculus (see above among related works). Crafting a weak interaction \bohm out might require some work.

\paragraph{Call-by-Value} In call-by-name, the gap between head normal form bisimilarity (aka interaction equivalence), and contextual equivalence is, roughly, extensionality. In call-by-value, the gap contains more computational principles, as stressed in particular by the recent works \cite{DBLP:journals/corr/abs-2303-08161,DBLP:conf/fscd/AccattoliL24}. It would be interesting to adapt interaction equivalence to call-by-value and explore whether variants of the definition catch theories in between.

\paragraph{Back to Improvements} Now that there is a notion of equational improvement, it is natural to adapt it to call-by-need and compare / revisit / extend the results about Sands' (non-equational) improvements and program transformations in the literature \cite{DBLP:journals/tcs/Sands96,DBLP:conf/popl/MoranS99,DBLP:conf/icfp/GustavssonS01,DBLP:conf/icfp/HackettH14,DBLP:conf/lics/HackettH15,DBLP:journals/pacmpl/HackettH18,DBLP:conf/ppdp/Schmidt-Schauss18}.

\paragraph{Game Semantics} To relate interaction equivalence with game semantics, it is natural to look at the \bohm tree game model by \citet{DBLP:journals/tcs/KerNO03}. Operational game semantics \cite{LevyLICSGames,jaber_et_al-games,DBLP:conf/icalp/Laird07}, a sub-area of game semantics based on labeled transition systems, is another natural candidate. More generally, our work provides a good justification for a systematic exploration of non-extensional game semantics.

\paragraph{PCF} Game semantics was introduced to capture denotationally \textsf{PCF} contextual equivalence, which is obtained via a quotient on games \cite{DBLP:journals/iandc/HylandO00,DBLP:journals/iandc/AbramskyJM00}. Is there a relationship between game models \emph{before the quotient} and \textsf{PCF} interaction equivalence?

\paragraph{Higher-Order Model Checking} Higher-order model checking is strongly connected to game semantics \cite{DBLP:conf/lics/Ong06}, intersection/multi types \cite{DBLP:conf/popl/Kobayashi09}, and \bohm trees \cite{DBLP:conf/fsttcs/ClairambaultM13}. It is reasonable to expect a connection with interaction equivalence.
 
\paragraph{Complete Normal Form Bisimilarities} In some $\l$-calculi with effects, normal form bisimilarities are complete for contextual equivalences \cite{DBLP:conf/birthday/StovringL09,DBLP:conf/fossacs/BiernackiLP19}. We conjecture that, therein, contextual equivalence and interaction equivalence coincide.

\paragraph{Interaction Cost} It is natural to wonder what kind of interaction cost emerges from our study and how it relates to both the actual cost of computation and interaction (in)equivalence. 
The characterization of our interaction improvement in terms of \bohm trees reveals that one can improve the interactive cost of a \lam-term by performing $\eta$-reductions, or by replacing some head-diverging sub-term—that is, an idle looping execution branch—with a non-looping one. 
Beyond \bohm tree characterizations, this may also relate to the length of interaction sequences in game semantics for the untyped call-by-name \lam-calculus, where interaction improvements could reflect optimizations in communication between player and opponent.

We believe that our interaction semantics can also effectively capture communication costs in the evaluation of functional programs within distributed environments. 
In this framework, black and white terms correspond to processes running on different machines, with interaction steps modeling inter-process communication, while locally executable steps remain silent.
To formalize this idea, we shall first adapt our interaction relations to a process calculus, and then investigate how the resulting improvements in interaction relates with measures of communication complexity.

\section{Conclusions}
\label{sect:conclusions}
Our work stems from the recognition of a tension between the equational aspect of contextual equivalences and the desire to observe the time cost of programs, expressed as the number of evaluation steps. We solve the tension by introducing the \emph{checkers calculus}, a $\l$-calculus where the internal and external aspects of computation receive a first-class status. The new setting is then used to define \emph{interaction equivalence}, which induces an equational theory for the ordinary $\l$-calculus: a very well-known one,  namely the equality $\cB$ of \bohm trees without $\eta$.

Beyond the technical aspects, the main takeaway is probably the framework, which is considerably simpler than other theories of interaction such as game semantics or the geometry of interaction, and not \emph{ad-hoc}, as witnessed by the relationship with $\cB$ and with multi types.

\begin{acks}
The authors would like to thank Guilhem Jaber for many discussions about operational game semantics for the $\l$-calculus, which eventually led us to the idea of tagging terms to distinguish between silent and interaction steps.
\end{acks}
\bibliographystyle{ACM-Reference-Format}
\bibliography{main.bib}

\withproofs{
\newpage
\appendix
\setboolean{appendix}{true}

\section{Proving the $\beta$ and $\eta$ Invariance of Contextual Equivalence}
\label{app:invariance}
The aim of this Appendix is providing the schema of rewriting-based proofs of the invariance with respect to $\beta$ and $\eta$ of head contextual equivalence $\obseq$ (that is, that $\eqb\,\subseteq\, \obseq$ and that $\eqeta\,\subseteq\, \obseq$), to convince the reader that these are non-trivial theorems resting on various properties, themselves non-trivial. This should not be surprising: it is well-known that establishing contextual equivalence is hard, because of the universal quantification on contexts.

\paragraph{$\beta$-Invariance of $\obseq$} The theorem rests on the following three auxiliary properties, the first two of which themselves require non-trivial proofs; the third one is instead a simple lemma.
\begin{enumerate}
\item \emph{Confluence of $\beta$}: if $\tm \tob^* \tmtwo_1$ and $\tm \tob^* \tmtwo_2$ then there exists $\tmthree$ such that $\tmtwo_1 \tob^* \tmthree$ and $\tmtwo_2 \tob^* \tmthree$.

\item \emph{Untyped normalization of head reduction}: if $\tm \tob^* \htm$ with $\htm$ a head normal form then head reduction $\toh$ terminates on $\tm$.

\item \emph{Stability of head normal forms by $\beta$}: if $\htm$ is a head normal form and $\htm \tob \tmtwo$ then $\tmtwo$ is a head normal form.
\end{enumerate}
The trickiest of them is the untyped normalization theorem. The statement is not a tautology because the sequence $\tm \tob^* \htm$ in the hypothesis is of \emph{arbitrary} $\beta$ steps, that is, not necessarily of head steps. The conclusion is also correctly stated because in general the head normal form of $\tm$ is not $\htm$. For instance, $\tm \defeq \Id \var (\Id \vartwo) \tob \Id \var \vartwo \tob \var \vartwo$ is a sequence in the hypothesis (whose first step is not a head step), but the head normal form of $\tm$ is $\var (\Id \vartwo)$ and not $\var\vartwo$. For a modern account of untyped normalization theorems, see \cite{DBLP:conf/aplas/AccattoliFG19}.

\begin{theorem}[$\beta$-Invariance of $\obseq$]
\hfill
\begin{enumerate}
\item If $\tm \tob \tmtwo$ then $\tm\obsle\tmtwo$.
\item If $\tm \tob \tmtwo$ then $\tmtwo\obsle\tm$.
\end{enumerate}
\end{theorem}
\begin{proof}
\hfill
\begin{enumerate}
\item $\tm\obsle\tmtwo$. Let $\ctx$ be a context and let $\ctxp\tm \toh^*\htm$ with $\htm$  head normal. Since $\ctxp\tm \tob \ctxp\tmtwo$, by confluence there is $\tmfour$ such that $\htm \tob^* \tmfour$ and $\ctxp\tmtwo \tob^* \tmfour$. By stability of head normal forms, $\tmfour$ is  head normal. By the untyped normalization theorem, $\ctxp\tmtwo$ is head normalizing.

\item $\tmtwo\obsle\tm$.  Let $\ctx$ be a context and let $\ctxp\tmtwo\toh^*\htm$ with $\htm$  head normal. Since $\ctxp\tm \tob \ctxp\tmtwo$, we have that $\ctxp\tm \tob^*\htm$. By the untyped normalization theorem, $\ctxp\tm$ is head normalizing.\qedhere
\end{enumerate}
\end{proof}

\paragraph{$\eta$-Invariance of $\obseq$} Also in this case, we need some auxiliary properties:

\begin{enumerate}
\item \emph{Commutation of $\toh$ and $\toeta$}: if $\tm \toh^* \tmtwo$ and $\tm \toeta^* \tmthree$ then there exists $\tmfour$ such that $\tmtwo \toeta^*\tmfour$ and $\tmthree \toh^* \tmfour$.
\item \emph{Postponement of $\toeta$ with respect to $\toh$}: if $\tm (\toh\cup\toeta)^* \tmtwo$ then there exists $\tmthree$ such that $\tm\toh^*\tmthree\toeta^*\tmtwo$.
\item \emph{Head adequacy of $\eta$}: if $\tm \toeta^* \tmtwo$ and $\tmtwo$ is head normal then $\tm$ is head normalizing.
\item \emph{Stability of head normal forms by $\eta$}: if $\tm \toeta^* \tmtwo$ and $\tm$ is head normal then $\tmtwo$ is head normal.
\end{enumerate}
The commutation and postponement properties are non-trivial. Let us focus on postponement, the argument for commutation is similar. In general, postponements can be proved easily if the local swaps between the two reductions do not duplicate steps. Unfortunately, it is not the case here. Let us recall that $\comb{1}\defeq \la \var\la\vartwo\var\vartwo \toeta \Id$. Then consider:
\begin{center}
$\begin{array}{cccccccccc}
(\la\var\var\var) \comb{1} & \toeta & (\la\var\var\var) \Id & \toh & \Id\Id
\end{array}$
\end{center}
It swaps as follows, duplicating the $\eta$ step:
\begin{center}
$\begin{array}{cccccccccc}
(\la\var\var\var) \comb{1} & \toh & \comb{1}\comb{1} & \toeta & \Id\comb{1} & \toeta & \Id\Id
\end{array}$
\end{center}
Consider now this example:
\begin{center}
$\begin{array}{cccccccccc}
\comb{1} \tm \tmtwo& \toeta & \Id \tm\tmtwo& \toh & \tm\tmtwo
\end{array}$
\end{center}
It swaps as follows, duplicating the head step:
\begin{center}
$\begin{array}{cccccccccc}
\comb{1} \tm \tmtwo& = & (\la \var\la\vartwo\var\vartwo)\tm \tmtwo& \toh & (\la\vartwo\tm\vartwo) \tmtwo & \toh & \tm\tmtwo
\end{array}$
\end{center}
For a proof of the postponement of $\eta$, see \cite[Cor.~15.1.6]{Barendregt84}.

\begin{theorem}[$\eta$-Invariance of $\obseq$]
\hfill
\begin{enumerate}
\item If $\tm \toeta \tmtwo$ then $\tm\obsle\tmtwo$.
\item If $\tm \toeta \tmtwo$ then $\tmtwo\obsle\tm$.
\end{enumerate}
\end{theorem}
\begin{proof}
\hfill
\begin{enumerate}
\item $\tm\obsle\tmtwo$. Let $\ctx$ be a context and let $\ctxp\tm \toh^*\htm$ with $\htm$  head normal. Since $\ctxp\tm \toeta \ctxp\tmtwo$, by commutation there is $\tmfour$ such that $\htm \toeta^* \tmfour$ and $\ctxp\tmtwo \toh^* \tmfour$. By stability of head normal forms by $\eta$, $\tmfour$ is head normal. Then, $\ctxp\tmtwo$ is head normalizing.

\item $\tmtwo\obsle\tm$.  Let $\ctx$ be a context and let $\ctxp\tmtwo\toh^*\htm$ with $\htm$  head normal. By postponement applied to $\ctxp\tm \toeta \ctxp\tmtwo\toh^*\htm$, we have that there exists $\tmthree$ such that $\ctxp\tm \toh^*\tmthree\toeta^*\htm$. By head adequacy of $\eta$, $\tmthree$ is head normalizing.\qedhere
\end{enumerate}
\end{proof}

\section{Proofs from Section \ref{sec:checkers-calculus} (The Checkers Calculus)}
In this section we provide omitted proofs about the checkers calculus, especially proving that the ordinary $\lambda$-calculus embeds nicely in the checkers calculus.

\gettoappendix{prop:embedding-of-head-reduction}
\begin{proof}
	\applabel{prop:embedding-of-head-reduction}
	\begin{enumerate}
		\item If $\la{\vec{\var}}(\la{\vartwo}\tmtwo)\tmthree\vec{\tm} \toh 
		\la{\vec{\var}}\tmtwo\isub{\vartwo}{\tmthree}\vec{\tm}$, then\\ $\lambda_{\clr{}\cdots\clr{}}\vec\var.\appp{\clr{}\cdots\clr{}}{(\cla{}\vartwo\monoToPlayer{\clr{}}{\tmtwo})\cdot^\clr{}\monoToPlayer{\clr{}}{\tmthree}}
		{\vec{\monoToPlayer{\clr{}}{\tm}}}
		\tohnoint \lambda_{\clr{}\cdots\clr{}}\vec\var.\appp{\clr{}\cdots\clr{}}{\monoToPlayer{\clr{}}{\tmtwo}\isub{\vartwo}{\monoToPlayer{\clr{}}{\tmthree}}}
		{\vec{\monoToPlayer{\clr{}}{\tm}}}$
		
		\item Suppose that $\monoToPlayer{\colr}{\la{\vec{\var}}(\la{\vartwo}\tmtwo)\tmthree\vec{\tm}} = \lambda_{\clr{}\cdots\clr{}}\vec\var.\appp{\clr{}\cdots\clr{}}{(\cla{}\vartwo\monoToPlayer{\clr{}}{\tmtwo})\cdot^\clr{}\monoToPlayer{\clr{}}{\tmthree}}
		{\vec{\monoToPlayer{\clr{}}{\tm}}}
		\tohnoint \lambda_{\clr{}\cdots\clr{}}\vec\var.\appp{\clr{}\cdots\clr{}}{\monoToPlayer{\clr{}}{\tmtwo}\isub{\vartwo}{\monoToPlayer{\clr{}}{\tmthree}}}
		{\vec{\monoToPlayer{\clr{}}{\tm}}}$
		then $ \lambda_{\clr{}\cdots\clr{}}\vec\var.\appp{\clr{}\cdots\clr{}}{\monoToPlayer{\clr{}}{\tmtwo}\isub{\vartwo}{\monoToPlayer{\clr{}}{\tmthree}}}
		{\vec{\monoToPlayer{\clr{}}{\tm}}}= \monoToPlayer{\colr}{\la{\vec{\var}}\tmtwo\isub{\vartwo}{\tmthree}\vec{\tm}}$ and $\la{\vec{\var}}(\la{\vartwo}\tmtwo)\tmthree\vec{\tm} \toh 
		\la{\vec{\var}}\tmtwo\isub{\vartwo}{\tmthree}\vec{\tm}$.
		
			\item Head normal forms (resp. checkers head normal forms) are exactly those terms of shape $\la{\vec{\var}}~\vartwo\,\vec{\tm}$ (resp. $\lap{\clr{1}\cdots\clr{k}}{\vec{\var}}~\appp{\clr{1}\cdots\clr{n}}{\vartwo}{\vec{\tm}}$). We conclude as $\monoToPlayer{\colr}{\la{\vec{\var}}~\vartwo\,\vec{\tm}}=\lap{\clr{}\cdots\clr{}}{\vec{\var}}~\appp{\clr{}\cdots\clr{}}{\vartwo}{\vec{\tm}}$.\qedhere
	\end{enumerate}
\end{proof}

\gettoappendix{prop:ordinary-reductions-can-happen-in-checkers}

\begin{proof}
	\applabel{prop:ordinary-reductions-can-happen-in-checkers}
	If $\la{\vec{\var}}(\la{\vartwo}\tmtwo)\tmthree\vec{\tm} \toh 
	\la{\vec{\var}}\tmtwo\isub{\vartwo}{\tmthree}\vec{\tm}$, then for a tagging $T$ of $\la{\vec{\var}}(\la{\vartwo}\tmtwo)\tmthree\vec{\tm}$,  $(\la{\vec{\var}}(\la{\vartwo}\tmtwo)\tmthree\vec{\tm})^T = \lambda_{\clr{1}\cdots\clr{k}}\vec\var.\appp{\clrtwo{1}\cdots\clrtwo{n}}{(\cla{}\vartwo{\tmtwo^T})\cdot^\colrtwo{\tmthree^T}}
	{\vec{{\tm}}^T}
	\tohcol
	\lambda_{\clr{1}\cdots\clr{k}}\vec\var.\appp{\clrtwo{1}\cdots\clrtwo{n}}{\tmtwo^T\isub\vartwo{\tmthree^T}}
	{\vec{{\tm}}^T}$ which is clearly a tagging of $\la{\vec{\var}}\tmtwo\isub{\vartwo}{\tmthree}\vec{\tm}$. 
\end{proof}

Thus, interaction equivalence is shown to be a refinement of the ordinary contextual equivalence.

\gettoappendix{l:improvements-are-ctx-equivalent}
\begin{proof}
\applabel{l:improvements-are-ctx-equivalent}
Let $\tm \intrleq \tmtwo$ and $\ctxp\tm\conv{\hsym}$, that is, $\ctxp\tm \toh^*\htm$ with $\htm$ in head normal form. By \refprop{embedding-of-head-reduction}(1), $\monoToBlue{\ctxp\tm} \tohch^* \monoToBlue\htm$, and by \refprop{embedding-of-head-reduction}(3) $\monoToBlue\htm$ is $\hchsym$-normal. That is, we have $\monoToBlue{\ctxp\tm}\conv{\hchsym}$, and so $\monoToBlue{\ctxp\tm}\bshcol{0}$ (with $0$, since everything is black). By $\tm \intrleq \tmtwo$, we obtain $\monoToBlue{\ctxp\tmtwo}\bshcol{0}$, that is, $\monoToBlue{\ctxp\tmtwo}\tohch^*\htmtwo$ for some head normal form $\htmtwo$. By \refprop{embedding-of-head-reduction}(2), there exists $\tmthree$ such that $\ctxp\tmtwo\toh^*\tmthree$ with $\monoToBlue\tmthree=\htmtwo$. By \refprop{embedding-of-head-reduction}(3), $\tmthree$ is head normal, thus $\ctxp\tmtwo \conv{\hsym}$. We conclude that $\tm \obsle \tmtwo$.
\end{proof}
\section{Proofs from \refsect{multi-types} (Checkers Multi Types and Relational Semantics)}

In this section we use the following convenient notation, that allows us to treat the indices in the two applications rules $\typingruleApp$ in a uniform manner.
For all players $\clr{},\clrtwo{}\in\{\redclr,\blueclr\}$, we define:
\[
	\gamma_{\clr{}\clrtwo{}} = 
	\begin{cases}
	0,&\textrm{if $\clr{} = \clrtwo{}$;}\\
	1,&\textrm{otherwise.}\\
	\end{cases}
\]
(For the readers used to Kronecker $\delta$, our $\gamma$ is the opposite: $\gamma_{\clr{}\clrtwo{}} = 1 - \delta_{\clr{}\clrtwo{}}$.)

We first prove some technical lemmas to reason on multi types and derivations.

\begin{lemma}[Splitting multisets with respect to derivations]
\label{l:types-splitting-multisets}
	Let $\tm$ be a term, $\tderiv \derives  \typectx \ctypes{k} \tm : \mtype$ a derivation, and $\mtype = \mtypetwo \mplus \mtypethree$  a splitting. Then there exist two derivations $\tderiv_{\mtypetwo} \derives  \typectx_{\mtypetwo} \ctypes{k_1} \tm : \mtypetwo$, and $\tderiv_{\mtypethree} \derives  \typectx_{\mtypethree} \ctypes{k_2} \tm : \mtypethree$ such that $\typctx = \typctx_{\mtypetwo} \mplus \typctx_{\mtypethree}$, $k = k_1 + k_2$, and $\insize{\tderiv}=\insize{\tderiv_{\mtypetwo}}+\insize{\tderiv_{\mtypethree}}$.
\end{lemma}

\begin{proof}
\applabel{l:types-splitting-multisets}
	The last rule of $\tderiv \derives  \typectx_{\mtype} \ctypes{k} \tm : \mtype$ can only be $\typingruleMany$, thus it is enough to re-group its hypotheses according to $\mtypetwo$ and $\mtypethree$. Since $\typingruleMany$ rules do not count in the measure of type derivations, it is immediate that $\insize{\tderiv}=\insize{\tderiv_{\mtypetwo}}+\insize{\tderiv_{\mtypethree}}$.
\end{proof}

\begin{lemma}[Substitution]
\label{l:types-substitution}
	If $\tderiv \derives \typctx, \var\hastype \mtype \ctypes k \tm \hastype \gtype$ and $\tderivtwo \derives \typctxtwo \ctypes {k'} \tmtwo \hastype \mtype$ then there exists a derivation $\tderivp \derives \typctx\uplus\typctxtwo \ctypes {k+k'} \tm\isub\var\tmtwo \hastype \gtype$ such that $\insize{\tderiv'}=\insize{\tderiv}+\insize{\tderivtwo}$.
\end{lemma}

\begin{proof}
	By induction on $\tderiv$. Cases of $\gtype$:
	\begin{enumerate}
		\item Linear types, \ie $\gtype \defeq \ltype$. Cases of the last derivation rule:
		
		\begin{enumerate}
			\item \emph{Unsubstituted variable}, \ie: $$\tderiv = \infer[\typingruleAx]{\var\hastype \zero,\vartwo \hastype [\ltype] \ctypes {0} \vartwo \hastype \ltype}{}$$
			Then as $\mtype = \zero$, $\tderivtwo$ must be of the following shape: 
			
			$$\infer[\typingruleMany]{\typctxtwo \ctypes {0} \tmtwo \hastype \zero}{}$$ 
			
			Hence $\tderivp \defeq \tderiv$ works, as $\vartwo\isub\var\tmtwo = \vartwo$ and $k'=0$. Clearly, $\insize{\tderiv}=\insize{\tderiv'}=0$, as required.
			\item \emph{Substituted variable}, \ie: $$\tderiv = \infer[\typingruleAx]{\var \hastype [\ltype] \ctypes 0 \var \hastype \ltype}{}$$
			Then, as $\mtype = [\ltype]$, $\tderivtwo$ must be of the following shape: 
			
			$$\infer[\typingruleMany]{\typctxtwo \ctypes {k'} \tmtwo \hastype [\ltype]}{\tderivtwo' \exder \typctxtwo \ctypes {k'} \tmtwo \hastype \ltype}$$
			
			Hence $\tderivp \defeq \tderivtwo'$ works, as $\var\isub\var\tmtwo = \tmtwo$ and $k=0$. Clearly, $\insize{\tderiv'}=\insize{\tderivtwo}=\insize{\tderiv} + \insize{\tderivtwo}$.
			
			\item \emph{Abstraction}, \ie: $$\tderiv = \infer[\typingruleAbs]{\typctx, \var \hastype \mtype \ctypes {k} \cla{}\vartwo\tmp \hastype \mtypetwo \typearrowpp{\clr{}}{\clrtwo{}} \ltype}{\rho\derives\typctx, \var \hastype \mtype, \vartwo \hastype \mtypetwo \ctypes {k} \tmp \hastype \ltype}$$
				
			By induction, there exists a derivation  $\rho'$ of final judgment $\typctx,\vartwo\hastype\mtypetwo \ctypes {k+k'} \tmp\isub\var\tmtwo \hastype \ltype$ such that $\insize{\rho'}=\insize{\rho} +\insize{\tderivtwo}$, from which we can conclude obtaining $\tderiv'$ by applying the derivation rule $\typingruleAbs$. We observe that $\insize{\tderiv'}=\insize{\rho'}=\insize{\rho} +\insize{\tderivtwo}= \insize{\tderiv} +\insize{\tderivtwo}$.
			
			\item \emph{Application}, \ie: 
			$$	\infer[\typingruleApp]{\typectx_1 \uplus \typectx_2, \var \hastype \mtype  \ctypes {k} \capp{}{\tm_1}{\tm_2} \hastype \ltype}{ \rho_1\derives\typectx_1, \var\hastype\mtype_1 \ctypes{l_1}{\tm_1} \hastype \mtype \typearrowpp{\clr{}}{\clrtwo{}} \ltype & \rho_2 \derives\typectx_2, \var\hastype\mtype_2  \ctypes {l_2} {\tm_2} \hastype \mtype  }$$ where $k =  l_1 + l_2 + \gamma_{\clr{}\clrtwo{}}$.
			
			We have that $\mtype = \mtype_1 \mplus \mtype_2$. By \reflemma{types-splitting-multisets}, the derivation $\tderivtwo$ splits into two derivations of final judgments $\tderivtwo_1 \derives \typctxtwo_1 \ctypes {k'_1} \tmtwo \hastype \mtype_1$ and $\tderivtwo_2 \derives \typctxtwo_2 \ctypes {k'_2} \tmtwo \hastype \mtype_2$ such that $k' = k'_1 + k'_2$, and $\insize{\tderivtwo}=\insize{\tderivtwo'} + \insize{\tderivtwo''}$.
			
By \ih, there are two derivations of final judgment  $\rho_1\derives\typectx_1 \ctypes  {l_1+k'_1} {\tm_1\isub\var\tmtwo} \hastype \mtype \typearrowpp{\clr{}}{\clrtwo{}} \ltype$ and  $\rho_2\derives \typectx_2 \ctypes {l_2+ k'_2} {\tm_2\isub\var\tmtwo} \hastype \mtype $, such that $\insize{\rho_i'}=\insize{\rho_i} +\insize{\tderivtwo_i}$. 
We conclude obtaining $\tderiv'$ by applying to these two derivations the rule $\typingruleApp$, as $k+k' = l_1 + k'_1+ l_2 + k'_2+\gamma_{\clr{}\clrtwo{}}$. We observe that $\insize{\tderiv'}=1+\insize{\rho'_1} + \insize{\rho'_2} =1+\insize{\rho_1} +\insize{\tderivtwo_1} + \insize{\rho_2} +\insize{\tderivtwo_2}= \insize{\tderiv} +\insize{\tderivtwo}$.
				
		\end{enumerate}
		
		\item Multi types, \ie $\gtype \defeq [\ltype_i]_{i\in I}$ with $\iset$ finite. The last rule of the derivation must be $\typingruleMany$:		
		$$
		\infer[\typingruleMany]{\uplus_{i\in I}\typectx_i, \var\hastype\mtype \ctypes{\sum_{i\in I} l_i}  \tm \hastype [\ltype_i]_{i\in I}}
		{(\tderiv_i \derives \typectx_i, \var\hastype\mtype_i \ctypes {l_i} \tm \hastype \ltype_i)_{i\in I}}$$		
		with $\mtype = \uplus_{i\in I} \mtype_i$ and $k = \sum_{i\in I}l_i$ and $\uplus_{i\in I}\typectx_i = \typctx$.
		
		By the multiset splitting lemma (\reflemma{types-splitting-multisets}), the derivation $\tderivtwo$ for $\tmtwo$ in the hypotheses splits in several derivations of final judgments $\tderivtwo_i \derives \typctxtwo_i \ctypes {k'_i} \tmtwo \hastype \mtype_i$ such that $k' = \sum_{i\in I}k'_i$ and $\uplus_{i\in I}\typctxtwo_i = \typctxtwo$, and $\sum\insize{\tderiv_i}=\insize{\tderivtwo}$.
		
		Then by induction hypothesis, there exist several derivations $\tderivp_i \derives \typectx_i \cup \typctxtwo_i \ctypes  {l_i+k'_i} {\tm\isub\var\tmtwo} \hastype \ltype_i$, such that $\insize{\tderiv_{i}'} = \insize{\tderiv_{i}} + \insize{\tderivtwo_i}$. The derivation $\tderivp$ of the statement is obtained by applying rule $\typingruleMany$ to the family of derivations $\{\tderivp_i\}_{i\in I}$, as follows:
		$$
		\infer[\typingruleMany]{\uplus_{i\in I}\typectx_i\cup\typctxtwo_i \ctypes{\sum_{i\in I} (l_i+k'_i)}  \tm\isub\var\tmtwo \hastype [\ltype_i]_{i\in I}}
		{(\tderivp_i \exder \typectx_i \cup \typctxtwo_i \ctypes  {l_i+k'_i} {\tm\isub\var\tmtwo} \hastype \ltype_i )_{i\in\iset}}$$		
		Note indeed that $k+k' = \sum_{i\in I} (l_i + k'_i)$ and $\uplus_{i\in I}\typectx_i\uplus\typctxtwo_i = \typctx \uplus \typctxtwo$. Moreover, $\insize{\tderiv'}= \sum_i \insize{\tderiv_{i}'} = \sum_i\insize{\tderiv_{i}} + \insize{\tderivtwo_i}  = \insize{\tderiv} + \insize{\tderivtwo}$. \qedhere
	\end{enumerate}
\end{proof}

\refprop{silent-root-subject-reduction} shows that the interpretation of a checkers term is invariant by silent root reduction.

\begin{proposition}[Silent Root Subject Reduction]
	\label{prop:silent-root-subject-reduction}
	Let $\tm,\tmtwo\in\Lambdac$.
	If $\tm \rtobnoint \tmp$ and $\tderiv \derives \typctx \ctypes { k}  \tm \hastype \ltype$ then there exists $\tderivp \derives \typctx \ctypes { k}  \tmp \hastype \ltype$ with $\insize\tderivp = \insize\tderiv-1$.
\end{proposition}
\begin{proof}
	\emph{Root step}, \ie $\tm = (\cla{}\var\tmthree)\cdot^\clr{}\tmtwo\rtobnoint \tmthree\isub\var\tmtwo=\tmp$.

	Then the last rule of the derivation $\tderiv$ is $\typingruleApp$ and is followed by the $\typingruleAbs$ rule on the left:
	
	$$	\infer[\typingruleApp]{\typectx \uplus \typectxtwo \ctypes {l_1+l_2} (\cla{}\var\tmthree)\cdot^\clr{}\tmtwo \hastype \ltype}{ \infer{\typectx \ctypes {l_1} \cla{}\var\tmthree \hastype \mtype \typearrowpp{\clr{}}{\clr{}} \ltype}{\typectx, \var \hastype \mtype\ctypes {l_1} \tmthree \hastype \ltype} & \typectxtwo \ctypes {l_2} \tmtwo \hastype \mtype  }$$
	
	By the Substitution \reflemma{types-substitution}, there exists a derivation $\tderiv' \derives \typctx \uplus \typctxtwo \ctypes{l_1+l_2} \tmthree \isub\var\tmfour \hastype \ltype$ such that $\insize{\tderiv'}=\insize{\tderiv_s}+\insize{\tderiv_u}=\insize{\tderiv}-1$.
\end{proof}

\begin{lemma}[Merging multisets w.r.t.\ derivations]
	\label{l:silly-merging-multisets}
	Let $\tm\in\Lambdac$. Consider two derivations:
	\begin{itemize}
		\item $\tderiv_{\mtypetwo} \derives  \typectx_{\mtypetwo} \ctypes{k_1} \tm : \mtypetwo$, and
		\item $\tderiv_{\mtypethree} \derives  \typectx_{\mtypethree} \ctypes{k_2} \tm : \mtypethree$.
	\end{itemize}
	Then there exists a derivation $\tderiv_{\mtypetwo} \derives  \typectx_{\mtypetwo}\uplus\typctx_{\mtypethree} \ctypes{k_1+k_2} \tm : \mtypetwo\mplus \mtypethree$.
\end{lemma}

\begin{proof}
	The last rules of $\tderiv_{\mtypetwo} \derives  \typectx_{\mtypetwo} \ctypes{k_1} \tm : \mtypetwo$ and $\tderiv_{\mtypethree} \derives  \typectx_{\mtypethree} \ctypes{k_2} \tm : \mtypethree$ can only be the rule $\typingruleMany$, thus it is enough to re-group their hypotheses.
\end{proof}

\begin{lemma}[Anti-substitution]
	\label{l:types-anti-substitution}
	Let $\tderiv \derives  \typectx\ \uplus \typectxtwo \ctypes{k+k'} \tm\isub\var\tmtwo : \gtype$. Then there exist
	\begin{itemize}
		\item a multi type $\mtype$,
		\item a derivation $\tderiv_\tm\derives{\typectx,\var:\mtype}\ctypes{k}{\tm}:{\gtype}$, and 
		\item a derivation $\tderiv_\tmtwo\derives{\typectxtwo}\ctypes{k'} {\tmtwo}:{\mtype}$.
	\end{itemize}
\end{lemma}

\begin{proof}
	By lexicographic induction on $(\gtype,\tm)$. Cases of $\gtype$:
	\begin{enumerate}
		\item Linear types, \ie $\gtype \defeq \ltype$. Cases of the last derivation rule:
		
		\begin{enumerate}
			\item \emph{Unsubstituted variable}, \ie: $\vartwo\isub\var\tmtwo = \vartwo$ and
			
			$$\tderiv = \infer[\typingruleAx]{\vartwo \hastype [\ltype] \ctypes {0} \vartwo \hastype \ltype}{}$$
			Then, taking $\mtype \defeq \zero$, $\tderiv_\tmtwo$ must be of the following shape: 
			
			$$\infer[\typingruleMany]{\typctxtwo \ctypes {0} \tmtwo \hastype \zero}{}$$ 
			
			Hence $\tderiv_\tm \defeq \tderiv$ works.
			\item \emph{Substituted variable}, \ie: $\var\isub\var\tmtwo = \tmtwo$ and $\tderiv \derives \typctxtwo \ctypes {k'} \tmtwo \hastype \ltype$. We take $\mtype\defeq[\ltype]$. Then:
			\[
			\infer{\tderiv_\tmtwo \derives \typctxtwo \ctypes {k'} \tmtwo \hastype [\ltype]}{\tderiv \derives \typctxtwo \ctypes {k'} \tmtwo \hastype \ltype}
			\]
			
			and
			
			\[
			\infer[\typingruleAx]{\tderiv_\tm=\var \hastype [\ltype] \ctypes {0} \var \hastype \ltype}{}\]
			
			\item \emph{Abstraction}, \ie: $$\tderiv = \infer[\typingruleAbs]{\typctx\uplus\typectxtwo \ctypes {k+k'} \la\vartwo\tmp\isub\var\tmtwo \hastype \mtypetwo \typearrowpp{\clr{}}{\clrtwo{}} \ltype}{\rho\derives\typctx\uplus\typectxtwo, \vartwo \hastype \mtypetwo \ctypes {k+k'} \tmp\isub\var\tmtwo \hastype \ltype}$$
			
			By induction, there exists a derivation  $\rho'$ of final judgment $\typctx,\vartwo\hastype\mtypetwo,\var:\mtype \ctypes {k} \tmp \hastype \ltype$ and a type derivation $\tderiv_\tmtwo$ of final judgment $\typctxtwo\ctypes {k'} \tmtwo \hastype \mtype$. Then

			$$\tderiv_\tm = \infer[\typingruleAbs]{\typctx, \var \hastype \mtype \ctypes {k} \cla{}\vartwo\tmp \hastype \mtypetwo \typearrowpp{\clr{}}{\clrtwo{}} \ltype}{\rho'\derives\typctx, \var \hastype \mtype, \vartwo \hastype \mtypetwo \ctypes {k} \tmp \hastype \ltype}$$
			
			\item \emph{Application}, \ie: 
			
			$$	\infer[\typingruleApp]{\typectx_1 \uplus \typectx_2\uplus\typectxtwo_1\uplus\typectxtwo_2  \ctypes  {k+k'+ \gamma_{\clr{}\clrtwo{}}} (\capp{}{\tm_1}{\tm_2})\isub\var\tmtwo \hastype \ltype}{ \rho_1\derives\typectx_1\uplus\typectxtwo_1 \ctypes  {k_1+k_1'} {\tm_1}\isub\var\tmtwo \hastype \mtypetwo \typearrowpp{\clr{}}{\clrtwo{}} \ltype & \rho_2 \derives\typectx_2\uplus\typectxtwo_2  \ctypes {k_2+k_2'} {\tm_2}\isub\var\tmtwo \hastype \mtypetwo  }$$ where $k =  k_1+k_2$ and $k' =  k_1'+k_2'$.
			
			By \ih, we have a type derivation $\rho_1'$ such that $\rho_1'\derives\typectx_1, \var\hastype\mtype_1 \ctypes  {k_1} {\tm_1} \hastype \mtypetwo \typearrowpp{\clr{}}{\clrtwo{}} \ltype$, a type derivation $\tderiv_{\tmtwo_1}$ such that $\tderiv_{\tmtwo_1} \derives \typctxtwo_1 \ctypes {k'_1} \tmtwo \hastype \mtype_1$, a type derivation $\rho_2'$ such that $\rho_2' \derives\typectx_2, \var\hastype\mtype_2  \ctypes {k_2} {\tm_2} \hastype \mtypetwo$, and a type derivation $\tderiv_{\tmtwo_2}$ such that $\tderiv_{\tmtwo_2} \derives \typctxtwo_2 \ctypes {k'_2} \tmtwo \hastype \mtype_2$. $\tderiv_\tmtwo$ is obtained by merging $\tderiv_{\tmtwo_1}$ and $\tderiv_{\tmtwo_2}$ (\reflemma{silly-merging-multisets}). $\tderiv_\tm$ is as follows:
			$$	
			\infer[\typingruleApp]{\typectx_1 \uplus \typectx_2, \var \hastype \mtype_1\uplus\mtype_2  \ctypes  {k+ \gamma_{\clr{}\clrtwo{}}} \capp{}{\tm_1}{\tm_2} \hastype \ltype}{ \rho_1'\derives\typectx_1, \var\hastype\mtype_1 \ctypes  {k_2} {\tm_1} \hastype \mtypetwo \typearrowpp{\clr{}}{\clrtwo{}} \ltype & \rho_2' \derives\typectx_2, \var\hastype\mtype_2  \ctypes {k_2} {\tm_2} \hastype \mtypetwo  }
			$$
			
		\end{enumerate}
		
		\item Multi types, \ie $\gtype \defeq [\ltype_i]_{i\in I}$ with $\iset$ finite. The last rule of the derivation must be $\typingruleMany$:		
		$$
		\infer[\typingruleMany]{\uplus_{i\in I}\typectx_i\cup\typctxtwo_i \ctypes{\sum_{i\in I} (k_i+k'_i)}  \tm\isub\var\tmtwo \hastype [\ltype_i]_{i\in I}}
		{(\tderivp_i \exder \typectx_i \cup \typctxtwo_i \ctypes  {k_i+k'_i} {\tm\isub\var\tmtwo} \hastype \ltype_i )_{i\in\iset}}
		$$	
		
	By applying the \ih, we obtain for each $i\in I$: $\tderiv_i \derives \typectx_i, \var\hastype\mtype_i \ctypes {k_i} \tm \hastype \ltype_i$ and
	$\tderivtwo_i \derives \typctxtwo_i \ctypes {k'_i} \tmtwo \hastype \mtype_i$. 
	$\tderiv_\tmtwo$ is obtaining by merging the $\tderivtwo_i$ (\reflemma{silly-merging-multisets}). $\tderiv_\tm$ is as follows:	
		$$
		\infer[\typingruleMany]{\uplus_{i\in I}\typectx_i, \var\hastype\mtype \ctypes{\sum_{i\in I} k_i}  \tm \hastype [\ltype_i]_{i\in I}}
		{(\tderiv_i \exder \typectx_i, \var\hastype\mtype_i \ctypes {k_i} \tm \hastype \ltype_i)_{i\in I}}$$		
		with $\mtype = \uplus_{i\in I} \mtype_i$.
		\qedhere
	\end{enumerate}
\end{proof}

\refprop{silent-root-subject-expansion} shows that the interpretation of a checkers term is invariant by silent root expansion.

\begin{proposition}[Silent Root Subject Expansion]
	\label{prop:silent-root-subject-expansion}
	If $\tm \rtobnoint \tmp$ and $\tderiv \derives \typctx \ctypes { k}  \tm' \hastype \ltype$ then there exists $\tderivp \derives \typctx \ctypes { k}  \tm \hastype \ltype$.
\end{proposition}
\begin{proof}
	\emph{Root step}, \ie $\tm = (\cla{}\var\tmthree)\cdot^\clr{}\tmtwo\rtobnoint \tmthree\isub\var\tmtwo=\tmp$.
	
	By the Anti-Substitution \reflemma{types-anti-substitution}, there exist
	$\tderivtwo \derives \typctx',\var:\mtype \ctypes {k_1}  \tmthree \hastype \ltype$ and $\rho \derives \typctx'' \ctypes {k_2}  \tmtwo \hastype \mtype$ such that $\typctx=\typctx'\uplus\typctx''$ and $k=k_1+k_2$.
	
	Then we can build the type derivation $\tderiv'\derives \typctx \ctypes { k}  \tm \hastype \ltype$ as follows:
	
\[
	\infer[\typingruleApp]{\typectx' \uplus \typectx'' \ctypes {k_1+k_2} (\cla{}\var\tmthree)\cdot^\clr{}\tmtwo \hastype \ltype}{ \infer{\typectx' \ctypes {k_1} \cla{}\var\tmthree \hastype \mtype \typearrowpp{\clr{}}{\clr{}} \ltype}{\sigma\derives\typectx', \var \hastype \mtype\ctypes {k_1} \tmthree \hastype \ltype} & \rho\derives\typectx'' \ctypes {k_2} \tmtwo \hastype \mtype  }
\]
This concludes the proof.
\end{proof}

The next proposition sums up the core of the technical development, showing that checkers multi types are compositional, invariant by silent $\beta$ and do not satisfy any form of $\eta$-equivalence.

\gettoappendix{thm:ch-types-compatibility}
\begin{proof}
\applabel{thm:ch-types-compatibility}\hfill
\begin{enumerate}
\item \emph{Compatibility}. Let $\tm$ and $\tmtwo$ such that $\tm \leqctype \tmtwo$. By induction on $\ctx$:
		\begin{itemize}
			\item $\ctx =\ctxhole$, trivial.
			
			\item $\ctx=\capp{}\ctxtwo\tmthree$. 
			\newcommand{\tmdepart}{\capp{}{\ctxtwop\tm}\tmthree}
			\newcommand{\tmarrivee}{\capp{}{\ctxtwop\tmtwo}\tmthree}
			
			Let $\typctx \ctypes{k} \tmdepart \hastype \ltype$. As $\tmdepart$ is an application, the last rule of the derivation is $\typingruleApp$.
			
			$$\infer{\typctx \ctypes{k} \tmdepart \hastype \ltype}{\typctx_1 \ctypes{k_1} \ctxtwop\tm \hastype \mtype\typearrowpp{\clr{}}{\clrtwo{}} \ltype& \typctx_2 \ctypes{k_2} \tmthree \hastype \mtype}$$ where $k =  k_1 + k_2+ \gamma_{\clr{}\clrtwo{}}$.
			
			By induction, $\typctx_1 \ctypes{k_1} \ctxtwop\tmtwo \hastype \mtype\typearrowpp{\clr{}}{\clrtwo{}} \ltype$.
			
			Hence we conclude by building the appropriate derivation for $\ctxtwop\tmtwo$: $$\infer{\typctx \ctypes{k} \tmarrivee \hastype \ltype}{\typctx_1 \ctypes{k_1} \ctxtwop\tmtwo \hastype \mtype\typearrowpp{\clr{}}{\clrtwo{}} \ltype& \typctx_2 \ctypes{k_2} \tmthree \hastype \mtype}$$
			
			\item $\ctx=\capp{}\tmthree\ctxtwo$. 
			
			\renewcommand{\tmdepart}{\capp{}\tmthree{\ctxtwop\tm}}
			\renewcommand{\tmarrivee}{\capp{}\tmthree{\ctxtwop\tmtwo}}
			
			Let $\typctx \ctypes{k} \tmdepart \hastype \ltype$. As $\tmdepart$ is an application, the last rule of the derivation is $\typingruleApp$.
			
			$$\infer{\typctx \ctypes{k} \tmdepart \hastype \ltype}{\typctx_1 \ctypes{k_1} \tmthree \hastype \mtype\typearrowpp{\clr{}}{\clrtwo{}} \ltype& \typctx_2 \ctypes{k_2} \ctxtwop\tm \hastype \mtype}$$ where $k = k_1 + k_2 + \gamma_{\clr{}\clrtwo{}}$.
			
			By induction, $\typctx_2 \ctypes{k_2} \ctxtwop\tmtwo \hastype \mtype$.
			
			Hence we conclude by building the appropriate derivation for $\ctxtwop\tmtwo$: $$\infer{\typctx \ctypes{k} \tmarrivee \hastype \ltype}{\typctx_1 \ctypes{k_1} \tmthree \hastype \mtype\typearrowpp{\clr{}}{\clrtwo{}} \ltype& \typctx_2 \ctypes{k_2} \ctxtwop\tm \hastype \mtype}$$
			
			\item $\ctx=\cla{}\var\ctxtwo$. 
			
			\renewcommand{\tmdepart}{\cla{}\var\ctxtwop\tm}
			\renewcommand{\tmarrivee}{\cla{}\var\ctxtwop\tmtwo}
			
			Let $\typctx \ctypes{k} \tmdepart \hastype \mtype\typearrowpp{\clr{}}{\clrtwo{}}\ltype$. As $\tmdepart$ is an abstraction, the last rule of the derivation is $\typingruleAbs$.
			\[
			\infer{\typectx \ctypes k \tmdepart \hastype \mtype \typearrowpp{\clr{}}{\clrtwo{}} \ltype}{\typectx, \var \hastype \mtype \ctypes k \ctxtwop\tm \hastype\ltype}
			\]
			By induction, $\typectx, \var \hastype \mtype \ctypes k \ctxtwop\tmtwo \hastype\ltype$. 
			
			Hence we conclude by building the appropriate derivation for $\ctxtwop\tmtwo$: 
			\[
			\infer{\typectx \ctypes k \tmarrivee \hastype \mtype \typearrowpp{\clr{}}{\clrtwo{}} \ltype}{\typectx, \var \hastype \mtype \ctypes k \ctxtwop\tmtwo \hastype\ltype}
			\]
		\end{itemize}
	\item \begin{itemize}
		\item $\semint\tm \subseteq \semint\tmtwo$. By Prop.~\ref{prop:silent-root-subject-reduction} if $\tm\rtobnoint\tmtwo$, then $\semint\tm \subseteq \semint\tmtwo$. We conclude by compatibility.
		\item $\semint\tm \supseteq \semint\tmtwo$. By Prop.~\ref{prop:silent-root-subject-expansion} if $\tm\rtobnoint\tmtwo$, then $\semint\tmtwo \subseteq \semint\tm$. We conclude by compatibility.
	\end{itemize}
	\item Let us consider type derivations for $\cla{}\vartwo\capp{}\var\vartwo$:
	\[
	\infer[\typingruleAbs]{
		\var\hastype [\mtype \typearrowpp{\clr{}'}{\clrtwo{}} \ltype] \ctypes{k} \cla{}\vartwo\capp{}\var\vartwo \hastype \mtype \typearrowpp{\clr{}}{\clrtwo{}'} \ltype 
	}{
		\infer[\typingruleAppInt]{
			\var\hastype [\mtype \typearrowpp{\clr{}'}{\clrtwo{}} \ltype],\vartwo:\mtype \ctypes{k} \capp{}\var\vartwo\hastype \ltype 
		}{ 
			\infer[\typingruleAx]{
				\var\hastype [\mtype \typearrowpp{\clr{}'}{\clrtwo{}} \ltype] \ctypes{0} \var \hastype \mtype \typearrowpp{\clr{}'}{\clrtwo{}} \ltype
			}{}
			~~~~&~~~~
			\infer[\typingruleMany]{
				\vartwo:\mtype\ctypes{0} \vartwo \hastype \mtype}{\vdots}{}
		}
	}
	\]
	where $k=\gamma_{\clr{}'\clrtwo{}}$. Then $([\mtype \typearrowpp{\clr{}'}{\clrtwo{}} \ltype],k,\mtype \typearrowpp{\clr{}}{\clrtwo{}'} \ltype )\in\semint{\cla{}\vartwo\capp{}\var\vartwo} $ for each $\clr{},\clr{}', \clrtwo{},\clr{}'$. Instead, $([\mtype \typearrowpp{\clr{}'}{\clrtwo{}} \ltype],k,\mtype \typearrowpp{\clr{}}{\clrtwo{}'} \ltype )\not\in\semint{\var} $  if $\clr{}\neq \clr{}'$ or $\clrtwo{}\neq \clrtwo{}'$ or $k\geq 1$. Then $\semint{\cla{}\vartwo\capp{}\var\vartwo}\not\subseteq \semint{\var} $.
	\item One immediately realizes that $([\vartype],0,\vartype)\in\semint\var$, because \[
	\infer[\typingruleAx]{\var \hastype [\vartype] \ctypes{0} \var \hastype \vartype}{}
	\]
	But it is impossible to type an abstraction $\cla{}\vartwo\capp{}\var\vartwo$ with the atomic type $\vartype$. Then $([\vartype],0,\vartype)\not\in\semint{\cla{}\vartwo\capp{}\var\vartwo}$, and hence $\semint\var\not\subseteq \semint{\cla{}\vartwo\capp{}\var\vartwo}$.\qedhere
\end{enumerate}
\end{proof}

\gettoappendix{coro:btype-preorder-inequational}
\begin{proof}
\applabel{coro:btype-preorder-inequational}
	$\leqbtype$ is compatible because $\tm\leqbtype\tmtwo \implies \monoToBlue{\tm}\leqctype\monoToBlue{\tmtwo} \overset{\text{\refthm{checkers-type-compatibility}(1)}}{\implies} 
	\monoToBlue{\ctxp\tm}\leqctype\monoToBlue{\ctxp\tmtwo} \implies \ctxp\tm\leqbtype\ctxp\tmtwo $ for each $\ctx$. $\leqbtype$ contains $\beta$ conversion because $\tm=_\beta \tmtwo \overset{\text{\refprop{embedding-of-head-reduction}(1)}}{\implies} \monoToBlue{\tm}\silconv \monoToBlue{\tmtwo} \overset{\text{\refthm{checkers-type-compatibility}(2)}}{\implies} \monoToBlue{\tm}\leqctype \monoToBlue{\tmtwo} \implies \tm \leqbtype \tmtwo$. The fact that $\leqbtype$ is neither extensional nor semi-extensional directly follows from \refthm{checkers-type-compatibility}(3) and \refthm{checkers-type-compatibility}(4).
\end{proof}

The next proposition generalizes the (already proven) \refprop{silent-root-subject-reduction} and \refprop{silent-root-subject-expansion} by showing that head reduction/expansion of a typed checkers term preserves its typability. Crucially, in subject reduction, the number of $\typingruleAppInt$ of a derivation, that is the $k$-index of the final judgment, diminishes exactly by $1$ if the head step is an interaction step, and is stable if the head step is silent.

In the following proof, we use the following grammar generating head contexts:
\[\begin{array}{r@{\hspace{.15cm}}r@{\hspace{.1cm}}l@{\hspace{.1cm}}ll}
	\textsc{Weak Head Contexts} & \hauxctx & \grameq & \ctxhole \mid \hauxctx\tm
	\\
	\textsc{Head Contexts} & \hctx & \grameq & \la\var\hctx \mid \hauxctx
	
\end{array}\]
\gettoappendix{prop:ch-subject}
\begin{proof}
	\applabel{prop:ch-subject}
	\begin{enumerate}
		\item \begin{enumerate}
			\item \emph{Root step}, \ie $\tm = \capp{}{(\cla{}\var\tmthree)}\tmtwo\tohch \tmthree\isub\var\tmtwo=\tmp$.
			
			Then the last rule of the derivation $\tderiv$ is $\typingruleApp$ and is followed by the $\typingruleAbs$ rule on the left:
			
			$$	\infer[\typingruleApp]{\typectx \uplus \typectxtwo \ctypes {k} \capp{}{(\cla{}\var\tmthree)}\tmtwo \hastype \ltype}{ \infer{\typectx \ctypes {k_1} \cla{}\var\tmthree \hastype \mtype \typearrowpp{\clr{}}{\clrtwo{}} \ltype}{\typectx, \var \hastype \mtype\ctypes {k_1} \tmthree \hastype \ltype} & \typectxtwo \ctypes{k_2} \tmtwo \hastype \mtype  }$$ where $k =  k_1 + k_2 + \gamma_{\clr{}\clrtwo{}}$.
			
			By the Substitution \reflemma{types-substitution}, we have that there exists a derivation $\tderiv' \derives \typctx \uplus \typctxtwo \ctypes{k_1+k_2} \tmthree \isub\var\tmfour \hastype \ltype$ such that $\insize{\tderiv'}=\insize{\tderiv_s}+\insize{\tderiv_u}=\insize{\tderiv}-1$. We conclude by observing that if $\tm \tohnoint\tmp$ then $k = k_1 + k_2$, and otherwise if $\tm \tohint \tmp$ then $k_1+k_2=k-1$.
			
			\item \emph{Contextual closure}. We have two subcases:
			\begin{enumerate}
				\item Weak contexts, \ie $\tm= \capp{}{\hauxctxp{s}}\tmtwo\toh \capp{}{\hauxctxp{s'}}\tmtwo=\tm'$. Then the last rule of $\tderiv$ is $\typingruleApp$:
				\[
				\infer{\typectx \uplus \typectxtwo \ctypes {k} \capp{}{\hauxctxp{s}}\tmtwo \hastype \ltype}{\tderivtwo\derives
					\typectx \ctypes {k_1} \hauxctxp{s} \hastype \mtype \typearrowpp{\clr{}}{\clrtwo{}} \ltype & 
					\typectxtwo \ctypes{k_2} \tmtwo \hastype \mtype } 
				\]
				where $k =  k_1 + k_2+ \gamma_{\clr{}\clrtwo{}}$. By \ih, there exists a derivation $\tderivtwo'\derives\typectx\ctypes{k_1'} \hauxctxp{s'} \hastype \mtype \typearrowpp{\clr{}}{\clrtwo{}} \ltype$, where $\insize{\tderivtwo'}=\insize{\tderivtwo}-1$ and 
				\[
				k_1' = \begin{cases}
					k_1, & if ~\hauxctxp{s} \tohnoint\hauxctxp{s'},
					\\
					k_1-1, & if~ \hauxctxp{s} \tohint \hauxctxp{s'}.
				\end{cases}
				\] Then, we can build the type derivation $\tderiv'$ as:
				\[
				\infer{\typectx \uplus \typectxtwo \ctypes {k'} \capp{}{\hauxctxp{s'}}\tmtwo \hastype \ltype}{\tderivtwo'\derives
					\typectx \ctypes {k_1'} \hauxctxp{s'} \hastype \mtype \typearrowpp{\clr{}}{\clrtwo{}} \ltype & 
					\typectxtwo \ctypes{k_2} \tmtwo \hastype \mtype } 
				\]
				where clearly $\size{\tderiv'}=\size{\tderiv}-1$ and $k' = \begin{cases}
					k, & \textrm{if }~\capp{}{\hauxctxp{s}}\tmtwo \tohnoint\capp{}{\hauxctxp{s'}}\tmtwo,
					\\
					k-1, & \textrm{if }~ \capp{}{\hauxctxp{s}}\tmtwo \tohint \capp{}{\hauxctxp{s'}}\tmtwo.
				\end{cases}$
				\item Head contexts, \ie $\tm= \cla{}{\var}{H\ctxholep{s}}\toh \cla{}{\var}{H\ctxholep{s'}}=\tm'$. Analogous.
			\end{enumerate}
		\end{enumerate}
	\item \begin{enumerate}
		\item \emph{Root step}, \ie $\tm = (\cla{}\var\tmthree)\cdot^\clrtwo{}\tmtwo\rtobnoint \tmthree\isub\var\tmtwo=\tmp$.
		
		By the Anti-Substitution \reflemma{types-anti-substitution}, there exist
		$\tderivtwo \derives \typctx',\var:\mtype \ctypes { k_1}  \tmthree \hastype \ltype$ and $\rho \derives \typctx'' \ctypes {k_2}  \tmtwo \hastype \mtype$ such that $\typctx=\typctx'\uplus\typctx''$ and $k=k_1+k_2$.
		
		Then we can build the type derivation $\tderiv\derives \typctx \ctypes { k'}  \tm \hastype \ltype$ as follows:
		
		$$	\infer[\typingruleApp]{\typectx' \uplus \typectx'' \ctypes {k'} (\cla{}\var\tmthree)\cdot^\clrtwo{}\tmtwo \hastype \ltype}{ \infer{\typectx' \ctypes {k_1} \cla{}\var\tmthree \hastype \mtype \typearrowpp{\clr{}}{\clrtwo{}} \ltype}{\sigma\derives\typectx', \var \hastype \mtype\ctypes {k_1} \tmthree \hastype \ltype} & \rho\derives\typectx'' \ctypes{k_2} \tmtwo \hastype \mtype  }$$
		where $k'=k+\gamma_{\clr{}\clrtwo{}}$.
		\item \emph{Contextual closure}. We have two subcases:
		\begin{enumerate}
			\item Weak contexts, \ie $\tm= \capp{}{\hauxctxp{s}}\tmtwo\toh \capp{}{\hauxctxp{s'}}\tmtwo=\tm'$. Then the last rule of $\tderiv'$ is $\typingruleApp$:
			\[
			\infer{\typectx \uplus \typectxtwo \ctypes {k} \capp{}{\hauxctxp{s'}}\tmtwo \hastype \ltype}{\tderivtwo'\derives
				\typectx \ctypes {k_1} \hauxctxp{s'} \hastype \mtype \typearrowpp{\clr{}}{\clrtwo{}} \ltype & 
				\typectxtwo \ctypes{k_2} \tmtwo \hastype \mtype } 
			\]
			where $k = k_1 + k_2 + \gamma_{\clr{}\clrtwo{}}$. By \ih, there exists a derivation $\tderivtwo\derives\typectx\ctypes{k_1'} \hauxctxp{s} \hastype \mtype \typearrowpp{\clr{}}{\clrtwo{}} \ltype$
			Then, we can build $\tderiv$ as follows:
			\[
			\infer{\typectx \uplus \typectxtwo \ctypes {k'} \capp{}{\hauxctxp{s}}\tmtwo \hastype \ltype}{\tderivtwo\derives
				\typectx \ctypes {k_1'} \hauxctxp{s} \hastype \mtype \typearrowpp{\clr{}}{\clrtwo{}} \ltype & 
				\typectxtwo \ctypes{k_2} \tmtwo \hastype \mtype } 
			\]
			\item Head contexts, \ie $\tm= \cla{}{\var}{H\ctxholep{s}}\toh \cla{}{\var}{H\ctxholep{s'}}=\tm'$. Analogous.\qedhere
		\end{enumerate}
	\end{enumerate}
	\end{enumerate}
\end{proof}

\gettoappendix{prop:ch-typability-hnf}
\begin{proof}
	\applabel{prop:ch-typability-hnf}
		Proposition~\ref{prop:nf-are-typable-tight} provides a stronger statement.
\end{proof}

By combining the previous results of this section, we now prove the standard result of intersection types: typability is equivalent to normalization (here, checkers head normalization).

\gettoappendix{thm:ch-head-characterization}
\begin{proof}
	\applabel{thm:ch-head-characterization}
	
	\begin{enumerate}
		\item By induction on $\size{\tderiv}$ and case analysis on whether $\tm$ is head normal. If $\tm$ is head
		normal then $\tm$ is head normalizable in $0\leq k$ interaction steps. If $\tm\tohch\tmtwo$ then there two cases:
		\begin{itemize}
			\item $\tm\tohnoint\tmtwo$. Then by quantitative subject reduction
			(Prop.~\ref{prop:ch-subject}(1)), there is $\tderivtwo\derives\typectx\ctypes k\tmtwo \hastype\ltype$ such that $\insize{\tderiv} = \insize{\tderivtwo}+1$. By i.h., $\tmtwo$ is head normalizable in less than $k$ interaction steps. Then, the same holds for $\tm$.
			\item $\tm\tohint\tmtwo$. Then by quantitative subject reduction
			(Prop.~\ref{prop:ch-subject}(1)), there is $\tderivtwo\derives\typectx\ctypes {k-1}\tmtwo \hastype\ltype$ such that $\insize{\tderiv} = \insize{\tderivtwo}+1$. By i.h., $\tmtwo$ is head normalizable in less than $k-1$ interaction steps. Then, $\tm$ is head normalizable in less than $k$.\qedhere
		\end{itemize}
	\item By induction on $m$. 
	Cases:
	\begin{enumerate}
		\item If $m=0$ then $\tm = \htm$. Then we conclude by Proposition~\ref{prop:ch-typability-hnf}. 
		
		\item If $m>0$ then $\tm \tohch \tmtwo \tohch^{m-1} \ntm$. By \ih, there 
		exists 
		$\tderivtwo\derives\tctx \vdash {\tmtwo}:{\ltype}$. By subject expansion (Prop.~\ref{prop:ch-subject}(2)), there exists 
		$\tderiv\derives\typectx \vdash {\tm}:{\ltype}$.\qedhere
	\end{enumerate}
	\end{enumerate}
\end{proof}

\section{Proofs of Section~\ref{sect:bohm-included-in-interaction} (From the \bohm Preorder to the Interaction One, via Multi Types)}
In this section, we prove that the notion of (checkers) tight types is able to represent derivations of normal forms that can only contain $\typingruleAppNoInt$ as applications $\typingruleApp$ rules, \ie derivations with a null index $k$ in the final judgment.

\gettoappendix{prop:nf-are-typable-tight}

\begin{proof}
\applabel{prop:nf-are-typable-tight}\hfill
	\begin{enumerate}
		\item 
%
%
%
%
%
%
Let $\ntm =\manyclam{n}{\var}\manycapp{m}{\var}{\tm}$. Set $\ltype'\defeq\mset{\emptytype \typearrowpp{\clrtwo{1}}{\clrtwo{1}} \cdots \emptytype \typearrowpp{\clrtwo{m}}{\clrtwo{m}} \vartype}$. 
Firstly, we build the following derivation:
	\[
	\infer=[m \typingruleApp]{\var\hastype\mset{\ltype'} \ctypes{0} \manycapp{m}{\var}{\tm} \hastype \vartype}
			{\infer{\vdots}{
			\infer[\typingruleApp]{\var\hastype\mset{\ltype'} \ctypes{0} \capp{1}\var {\tm_1} \hastype \emptytype \typearrowpp{\clrtwo{2}}{\clrtwo{2}} \cdots \emptytype \typearrowpp{\clrtwo{m}}{\clrtwo{m}} \vartype
			}{
				\var\hastype[\ltype'] \ctypes{0} \var \hastype \ltype' & \infer[\typingruleMany]{\emptytypectx \ctypes{0} \tm_{1} \hastype \emptytype}{}}
			}
			& \infer[\typingruleMany]{\emptytypectx \ctypes{0} \tm_{m} \hastype \emptytype}{}
			}
	\]
Then, there are two cases:
\begin{itemize}
\item $\var = \var_i$ for some $i\in\set{1,\ldots,n}$. For the sake of simplicity, let us say that $i=n$, the other cases are analogous. Then, we obtain the following derivation:
	\[
\infer=[(n-1) \typingruleAbs]{\ctypes{0} \manyclam{n}{\var}\manycapp{m}{\var}{\tm} \hastype \emptytype \typearrowpp{\clr{1}}{\clr{1}} \cdots \emptytype \typearrowpp{\clr{n-1}}{\clr{n-1}}\mset{\ltype'}\typearrowpp{\clr{n}}{\clr{n}} \vartype
}{
	\infer[\typingruleAbs]{ \ctypes{0} \cla{n}{\var_n}\manycapp{m}{\var}{\tm} \hastype \mset{\ltype'}\typearrowpp{\clr{n}}{\clr{n}} \vartype
	}{
			\var\hastype\mset{\ltype'} \ctypes{0} \manycapp{m}{\var}{\tm} \hastype \vartype
	}
}
	\]
Let $\ltype \defeq \emptytype \typearrowpp{\clr{1}}{\clr{1}} \cdots \emptytype \typearrowpp{\clr{k-1}}{\clr{k-1}}\mset{\ltype'}\typearrowpp{\clr{n}}{\clr{n}} \vartype$. 	We show that $(\emptytypctx,\ltype)$ is a tight typing:
	
	\begin{enumerate}
		\item \emph{Exactly one multiset is not empty:} the only non-empty multiset is $\mset{\ltype'}$ in $\ltype$.
		
		\item \emph{The arrows are silent}: all arrows are indeed silent.
	\end{enumerate}	

\item $\var \neq \var_i$ for all $i\in\set{1,\ldots,n}$.  Then, we obtain the following derivation:
	\[
\infer=[n\typingruleAbs]{\var\hastype\mset{\ltype'} \ctypes{0} \manyclam{n}{\var}\manycapp{m}{\var}{\tm} \hastype \emptytype \typearrowpp{\clr{1}}{\clr{1}} \cdots \emptytype \typearrowpp{\clr{n}}{\clr{n}} \vartype
}{
			\var\hastype\mset{\ltype'} \ctypes{0} \manycapp{m}{\var}{\tm} \hastype \vartype
}
	\]
Let $\ltype \defeq \emptytype \typearrowpp{\clr{1}}{\clr{1}} \cdots \emptytype\typearrowpp{\clr{n}}{\clr{n}} \vartype$. 	Note that $(\var\hastype\mset{\ltype'},\ltype)$ is a tight typing:
	
	\begin{enumerate}
		\item \emph{Exactly one multiset is not empty:} the only non-empty multiset is $\mset{\ltype'}$.
		
		\item \emph{The arrows are silent}: all arrows are indeed silent.
	\end{enumerate}
\end{itemize}

		
		\item 
Suppose $\htm =\manyclam{n}{\var}\manycapp{m}{\var}{\tm}$ and let $(\typctx,\ltype)$ a tight typing and an integer $k$ such that $\typctx \ctypes k \htm \hastype \ltype$ (please note that $\var$ could be either a free variable $\vartwo$ or a bound variable $\var_i$ for $1\leq i\leq k$).
		
		The typing derivation can only look like (that is, the difference with the type derivations built at the previous point is that the atomic type $\vartype$ can more generally be a type $\ltypethree$ made out of empty multi-sets and silent arrows):
		\[\infer=[n \typingruleAbs]{\typctx \ctypes{k} \htm \hastype \ltype}{\infer=[m \typingruleApp]{\typctx, \var_1 \hastype \mtype_1,\ldots,\var_n \hastype \mtype_n \ctypes{k} \manycapp{m}{\var}{\tm} \hastype \ltypethree}{\typctxtwo_0 \ctypes{l_0} \var \hastype \ltype' & (\typctxtwo_i \ctypes{l_i} \tm_{i} \hastype \mtypetwo_{i})_{1 \leq i\leq m}}}\]
		Where: 
		\begin{itemize}
			
			\item $\typctx, \var_1 \hastype \mtype_1,\ldots,\var_n \hastype \mtype_n = \uplus_{0 \leq i \leq m} \typctxtwo_i$;
			\item $\ltype' = \mtypetwo_{1} \typearrowpp{\clrp{1}}{\clrtwo{1}}  \cdots \mtypetwo_{m} \typearrowpp{\clrp{m}}{\clrtwo{m}} \ltypethree$;
			
			\item $l_0 + \sum_{1 \leq i\leq m} (l_i +\gamma_{\clrp{i}\clrtwo{i}})= k$.
		\end{itemize}
	
		Then, we get that $\typctxtwo_0= \var\hastype [\ltype']$. By assumption, as $(\typctx,\ltype)$ is a tight typing and $\typctxtwo_0\subseteq\typctx, \var_1 \hastype \mtype_1,\ldots,\var_n \hastype \mtype_n$ we get that $\typctx, \var_1 \hastype \mtype_1,\ldots,\var_n \hastype \mtype_n = \var\hastype [\ltype']$ and that $\clrp{i}=\clrtwo{i}$ and $\mtypetwo_i=\emptytype$ for all $1\leq i \leq m$.
		
		Hence, for all $i$, $\gamma_{\clrp{i}\clrtwo{i}} = 0$ and $l_i=0$ (the derivations for $\tm_i\hastype\mtypetwo_i$ can only start by the many rule with no premises).
		
		As $l_0=0$, ($\typctxtwo_0 \ctypes{l_0} \var \hastype \ltype'$ only consists of an axiom rule), we conclude that $k=0$.
		\qedhere		
	\end{enumerate}
\end{proof}

}

\end{document}